\documentclass[10pt,reqno,a4paper]{article}

\usepackage{authblk}
\usepackage{mathtools}
\usepackage{graphicx,float}
\usepackage[margin=10pt,font=small,labelfont=bf,labelsep=endash]{caption}
\usepackage{a4wide}
\usepackage[english]{babel}
\usepackage[utf8]{inputenc}
\usepackage{csquotes}
\usepackage{mathrsfs}
\usepackage[makeroom]{cancel}
\usepackage{amssymb}
\usepackage{verbatim}
\usepackage{datetime}
\usepackage{bm}
\usepackage{enumerate}
\usepackage[dvipsnames,table]{xcolor}
\usepackage[colorlinks = true, allcolors=Blue, breaklinks=true]{hyperref}
\usepackage{silence}
\usepackage{tensor}
\usepackage{comment}
\usepackage{amsthm}

\usepackage[
  backend=biber,
  style=alphabetic,
  sorting=nty,
  doi=true,
  maxnames= 6,
  giveninits=true,
  url=false,
  isbn=false
]{biblatex}

\renewbibmacro{in:}{}
\DeclareFieldFormat[article]{volume}{\mkbibbold{#1}}

\DeclareFieldFormat[article]{title}{#1}

\DeclareFieldFormat{journaltitle}{#1}

\DeclareFieldFormat{volume}{\mkbibbold{#1}}

\DeclareFieldFormat{pages}{#1}
\DeclareFieldFormat{eid}{#1}

\renewbibmacro*{note+pages}{%
  \printfield{note}%
}

\newbibmacro*{linkedjournalblock}{%
  \printfield{journaltitle}%
  \setunit{\addspace}%
  \printfield{volume}%
  \setunit{\addcomma\space}%
  \iffieldundef{eid}
    {\printfield{pages}}
    {\printfield{eid}}%
  \setunit{\addspace}%
  \printtext{\mkbibparens{\printfield{year}}}%
}

\renewbibmacro*{journal+issuetitle}{%
  \iffieldundef{journaltitle}
    {}
    {%
      \iffieldundef{doi}
        {\usebibmacro{linkedjournalblock}}
        {\href{https://doi.org/\thefield{doi}}{\usebibmacro{linkedjournalblock}}}%
      \newunit
    }%
}

\DeclareFieldFormat[book]{title}{%
  \iffieldundef{doi}
    {#1}
    {\href{https://doi.org/\thefield{doi}}{#1}}%
}

\DeclareFieldFormat[inbook]{title}{%
  \iffieldundef{doi}
    {#1}
    {\href{https://doi.org/\thefield{doi}}{#1}}%
}
\DeclareFieldFormat[incollection]{title}{%
  \iffieldundef{doi}
    {#1}
    {\href{https://doi.org/\thefield{doi}}{#1}}%
}

\usepackage[capitalise]{cleveref}

\newcommand{\mailto}[1]{\thanks{\href{mailto:#1}{#1}}}

\newcommand{\hnabla}{\overset{\mathrm{h}}{\nabla}{}}
\newcommand{\vnabla}{\overset{\mathrm{v}}{\nabla}{}}

\numberwithin{equation}{section}
\newtheorem{lemma}{Lemma}[section]
\newtheorem{definition}{Definition}[section]

\newtheorem{proposition}{Proposition}[section]
\newtheorem{theorem}{Theorem}[section]
\newtheorem{remark}{Remark}[section]
\newtheorem{example}{Example}[section]

\newcommand{\J}[2]{\tensor[^{#1}]{\mathbb{J}}{^{#2}}}
\renewcommand{\:}{\colon}

\newcounter{mnotecount}[subsection]

\graphicspath{{Figures/}}

\title{Pseudodifferential Weyl calculus on vector bundles}

\author[1]{Lars Andersson \mailto{larand@kth.se}}
\affil[1]{Department of Mathematics, Royal Institute of Technology, Lindstedtsvägen 25, 114 28 Stockholm, Sweden}

\author[2]{Benjamin Moser \mailto{benjamin.moser@univie.ac.at}}
\affil[2]{University of Vienna, Faculty of Physics, Boltzmanngasse~5, 1090 Vienna, Austria}

\author[2]{Marius A. Oancea \mailto{marius.oancea@univie.ac.at}}

\author[3,4]{Claudio F. Paganini \mailto{claudio.paganini@ur.de}}
\affil[3]{Faculty for Mathematics, University of Regensburg, 93040 Regensburg, Germany}
 \affil[4]{Faculty for Mathematics, TU Chemnitz, 09111 Chemnitz, Germany}

\author[5,6]{Gabriel Schmid \mailto{gabriel.schmid@edu.unige.it}}
\affil[5]{Department of Mathematics, University of Genoa, 16146 Genoa, Italy}
\affil[6]{INFN-Section Genoa and INdAM-GNFM, 16146 Genoa, Italy}

\date{}

\begin{document}

\maketitle
 
\begin{abstract}
    We develop a geometric framework for Weyl quantization on pseudo-Riemannian manifolds, in which pseudodifferential operators act on sections of vector bundles equipped with pseudo-Hermitian metrics and compatible connections. We construct the associated star product and compute its semiclassical expansion up to third order in the semiclassical parameter. A central feature of our approach is a correspondence, modulo smoothing remainders, between formally self-adjoint symbols and formally self-adjoint operators, extending known results from flat space to curved geometries. In addition, we analyze the Moyal equation satisfied by the Wigner function in this setting and provide explicit computations of Weyl symbols for several physically significant operators, including the Dirac, Maxwell, linearized Yang-Mills, and linearized Einstein operators. Our results lay the foundation for future developments in quantum field theory on curved spacetimes, semiclassical analysis, and chiral kinetic theory.
\end{abstract}

\paragraph*{\textbf{Keywords:}} Pseudodifferential calculus, Weyl quantization, star product, Moyal equation.
\paragraph*{\textbf{MSC 2020:}} Primary: 58J40; Secondary: 35S05, 53D55, 81Q20.

\section{Introduction}

In physics, quantization represents a systematic procedure that promotes a classical mechanical system described by phase space quantities to a quantum mechanical system described by operators acting on Hilbert spaces. From the mathematical point of view, this can be understood in terms of pseudodifferential calculus, which relates operators to their symbols.

Hermann Weyl first introduced what is now known as \emph{Weyl quantization} in his work on group theory and quantum mechanics \cite{Weyl,WeylBook}, proposing a map from symmetrized classical observables in phase space to quantum operators. Around the same time, Wigner \cite{Wigner} defined a quasiprobability distribution function on phase space, now known as the \textit{Wigner function}, in his analysis of quantum corrections to classical statistical mechanics. Although not originally framed as such, the Wigner transform can be viewed, formally and up to normalization conventions, as the inverse of Weyl quantization. The connection between Weyl's quantization and Wigner's transformation was only fully recognized in the late 1940s through the work of Groenewold \cite{Groenewold} and Moyal \cite{Moyal_1949}. This quantization scheme has since been referred to by various combinations of the aforementioned authors' names; in this work, we will refer to it simply as \emph{Weyl quantization}.

Beyond its significance as a quantization scheme in theoretical physics, Weyl quantization is also of mathematical interest within the framework of \emph{pseudodifferential calculus}. In this context, Weyl quantization is a mapping that assigns to a given \textit{symbol} $a(x,p)\in C^{\infty}(T^{\ast}\mathbb{R}^{d})$, assumed to behave as a homogeneous function of $p$ in the asymptotic limit $\vert x\vert, \vert p\vert\to\infty$, a linear and continuous operator $\hat{A}\:C^{\infty}_{\mathrm{c}}(\mathbb{R}^{d})\to C^{\infty}(\mathbb{R}^{d})$ defined by
\begin{align}\label{eq:Weyl}
    \hat{A}\psi(x):=\frac{1}{(2\pi)^d}\int_{\mathbb{R}^{d}}\int_{\mathbb{R}^{d}}a\bigg(\frac{x+y}{2},p\bigg)\psi(y)e^{i p \cdot (x-y)}\,d^dy\,d^dp\qquad\forall\psi\in C_{\mathrm{c}}^{\infty}(\mathbb{R}^{d}) .
\end{align}
One of the essential features of this quantization scheme is the fact that real symbols induce (formally) self-adjoint operators and that the quantization is covariant under linear symplectic transformations of phase space $T^{\ast}\mathbb{R}^{d}\cong\mathbb{R}^{d}\times\mathbb{R}^{d}$.

Although Weyl quantization is a natural choice because of its symmetry and self-adjointness properties, it is only one among a broad family of quantization schemes. The theory of pseudodifferential operators and, more generally, \textit{Fourier integral operators}, was developed in the 1960s by Bokobza-Unterberger \cite{BokobzaUnterberger}, Kohn-Nirenberg \cite{KohnNirenberg}, Hörmander \cite{Hormander,Hormander2,HormanderFourier}, and others, generalizing earlier works on \textit{singular integral operators} (see, e.g.,~\cite{Singular,Mikhlin,Calderon1,Calderon2}), as a tool to study (elliptic) partial differential equations. Since then, pseudodifferential calculus has become an indispensable tool for studying linear PDEs. For instance, even the simplest operations one might want to perform on elliptic differential operators (e.g., taking the inverse, square root, etc.) typically yield operators outside the class of differential operators. Pseudodifferential operators provide a generalization of differential operators, and under some reasonable assumptions on the class of symbols, the operations mentioned above can be performed within this class. Furthermore, the correspondence between operators and their symbols allows many computations and constructions to be carried out purely at the algebraic level of the symbols. We refer to \cite{Shubin,Sjorstrand,HormanderIII,Hintz} for a detailed discussion of that subject. 

In the context of pseudodifferential calculus, one often considers a quantization of the type 
\begin{align}\label{eq:KohnNirenberg}
    \hat{A}\psi(x):=\frac{1}{(2\pi)^d}\int_{\mathbb{R}^{d}}\int_{\mathbb{R}^{d}}a(x,p)\psi(y)e^{i p \cdot (x-y)}\,d^dy\,d^dp\qquad\forall\psi\in C_{\mathrm{c}}^{\infty}(\mathbb{R}^{d}),
\end{align}
which we shall henceforth refer to as \textit{Kohn-Nirenberg quantization}. Although less symmetric, one advantage of the Kohn-Nirenberg quantization scheme is that its symbols naturally generalize the \emph{total symbols} of linear differential operators, which are linear combinations of symbols of the form $a(x,p)=a(x)p^n,\,n\in\mathbb{N}$. However, both the Weyl \eqref{eq:Weyl} and the Kohn-Nirenberg quantization \eqref{eq:KohnNirenberg} are just two possibilities among a class of \textit{$\tau$-quantizations} with symbols of the form $a(z_{\tau}(x,y),p)$ with $z_\tau(x,y) := x + \tau(y-x)$ depending on a given parameter $\tau\in [0,1]$. 

From the point of view of pseudodifferential calculus, the Weyl quantization on $\mathbb{R}^{d}$ for a general class of symbols has been extensively studied by Hörmander \cite{HormanderWeyl} and Dencker \cite{DenckerWeyl}, following earlier works of Grossmann-Loupias-Stein \cite{Grossmann}, Berezin-Shubin \cite{BerezinShubin}, and Voros \cite{Voros}; see also the expositions in \cite[Sec.~18.5]{HormanderIII} and \cite[Sec.~23]{Shubin}. For a mathematical discussion with an eye towards applications in quantum field theory, see \cite{DerezinskiGerard}. 

On Euclidean space, the quantization map establishes a canonical linear isomorphism between symbol classes and the corresponding pseudodifferential operators. Moreover, there are explicit formulas for both constructing an operator from its symbol and for recovering the symbol of a given operator. A central feature of this calculus is its covariance under coordinate changes. This local covariance, in turn, allows one to define symbol classes on manifolds with values in vector bundles. Furthermore, one can define pseudodifferential operators on manifolds acting on the sections of vector bundles by patching together local constructions using a partition of unity (see \cite[Chap.~5]{Hintz} and \cite[Sec.~4]{Shubin}). While sufficient for many applications, this approach depends, at least up to a smoothing operator, on non-canonical choices such as coordinate systems and trivializations, and thus lacks a fully intrinsic character. In the standard locally patched calculus, only the \emph{principal symbol} of a given pseudodifferential operator behaves tensorially and is globally well defined. Consequently, there is strong motivation to develop a coordinate-independent, geometrically natural theory of symbols and quantization on manifolds.

Several such approaches have been explored in the literature. Early work concentrated on the Kohn-Nirenberg quantization for scalar-valued operators, notably by Bokobza-Haggiag \cite{Bokobza1,Bokobza2}, Widom \cite{Widow1,Widow2}, and Drager \cite{Drager}. More recent developments have extended these ideas to vector bundle-valued settings, as explored by Pflaum \cite{PflaumNormal} and Sharafutdinov \cite{MR2191866,MR2186590}. A distinct line of investigation, encompassing both the scalar and vector bundle cases, as well as the broader family of $\tau$-quantizations, has been pursued by Levy \cite{Levy}. Rather than relying on connections, Levy regards the exponential map as a central object. More precisely, he considers a restricted class of manifolds, namely those admitting a \textit{(global) linearization}, a concept originally defined in \cite{Bokobza1}. These are non-compact manifolds endowed with a smooth map $\mathrm{exp} \colon TM \to M$ that mimics the exponential map of pseudo-Riemannian geometry. Examples of such manifolds are geodesically complete pseudo-Riemannian manifolds, which are \textit{globally geodesically convex}, meaning that every pair of points is connected by a unique geodesic segment. An invariant $\tau$-quantization on closed Riemannian manifolds for operators acting on $1$-forms has been developed in \cite{CapoferriCurl}. Further advancements in Weyl and $\tau$-quantizations within the scalar framework have been contributed by Underhill \cite{Underhill}, Unterberger \cite{Unterberger}, Liu-Qian \cite{LiuQian}, Pflaum \cite{PflaumWeyl}, Safarov \cite{Safarov0,Safarov,SafarovBook}, Fulling \cite{fulling} (see also the earlier works~\cite{Fulling0,Kennedy}), and, more recently, by Dereziński-Latosiński-Siemssen \cite{dls_weyl}.
One of the main characteristics of the latter two approaches is the inclusion of some power of the \textit{Van Vleck-Morette determinant} in the quantization, which has been argued to be preferable from a geometric perspective. The aim of this paper is to extend the work of Fulling and Dereziński-Latosiński-Siemssen to the setting of vector bundles. In particular, this generalization enables the treatment of a wide range of physically significant operators, including the Maxwell and linearized Yang-Mills operators, the Dirac operator, and the linearized Einstein operator, which makes it of direct interest for physical applications.

Another line of research concerns the definition of \emph{global Fourier integral operators}, generalizing the pseudodifferential operator constructions mentioned above, especially in the context of obtaining a global construction of propagators for hyperbolic equations. In particular, Laptev-Safarov-Vassiliev \cite{Laptev} (see also~\cite{SafarovVassiliev}) demonstrated that the propagator for a broad class of hyperbolic operators on compact manifolds can be represented as a \textit{global} Fourier integral operator with a complex-valued phase function (as in \cite{Melin}). Building upon this work and some applications discussed in \cite{Jakobson}, a global symbol calculus for Fourier integral operators on closed manifolds has been developed by Safarov \cite{Safarov2}. Furthermore, a geometric approach to global parametrices for the scalar wave operator on closed Riemannian manifolds, including a global invariant definition of the full symbol of the wave propagator and an algorithm for explicitly calculating its homogeneous components, was recently proposed by Capoferri-Levitin-Vassiliev \cite{Capoferri1}. This result has also been generalized to the Lorentzian setting by Capoferri-Dappiaggi-Drago \cite{Capoferri2} for the wave operator on globally hyperbolic Lorentzian manifolds with compact Cauchy hypersurfaces. Generalizations to the propagator for the Dirac operator have been obtained by Capoferri-Vassiliev \cite{Capoferri3} in the Riemannian setting and by Capoferri-Murro \cite{CapoferriMurro} in the Lorentzian setting.

The construction of an invariant and global pseudodifferential calculus has become an increasingly active area of research in recent years, as reflected by the preceding discussion. Beyond the geometric approaches mentioned so far, other significant frameworks have been developed for manifolds with additional algebraic structures or specified singular structures. Examples include the study of intrinsic pseudodifferential operators on \emph{manifolds with a Lie structure at infinity}, e.g.~\cite{Ammann,Nistor}, as well as the invariant construction of global pseudodifferential operators on \emph{Lie groups}, with early work by Strichartz \cite{Strichartz}, see e.g.~the monograph by Fischer-Ruzhansky \cite{FischerRuzhansky} on this topic and references therein. In the context of the present article, it is particularly noteworthy to mention the recent work on Weyl quantization on a large class of locally compact Lie groups developed in \cite{Mantoiu}.

Parallel to the development of pseudodifferential calculus and Weyl quantization, a different strand of quantization theory emerged in the form of deformation quantization~\cite{rieffel1993deformation}. An overview of the literature can be found in \cite{sternheimer1998deformation}, or in the more recent PhD thesis \cite{weiss2010deformation}, and applications to quantum field theory are discussed in \cite{buchholz2011warped}.
Instead of assigning operators to symbols, deformation quantization aims to endow the space of smooth functions on a Poisson manifold, such as the cotangent bundle, with a noncommutative star product that deforms the usual pointwise multiplication in a way that reflects the underlying Poisson structure. It was shown in \cite{Kontsevich2003} that any finite-dimensional Poisson manifold admits a deformation quantization. The case of deformation quantization of Hermitian vector bundles was treated in~\cite{bursztyn2000deformation}. 

It is well known~\cite{Landsman1999,fedosov1996deformation} that for the standard Poisson structure on the cotangent bundle of flat $\mathbb{R}^n$ deformation quantization and Weyl quantization agree. In \cite{fedosov2001pseudo}, the connection between deformation quantization and pseudodifferential operators was explored. Recently in~\cite{much2024deformation}, a deformation quantization for non-flat spacetimes was developed. This poses the question of whether the equivalence of these different quantization schemes holds more generally. The structure of the expansion of the star product in \cref{Prp:StarProdExp} certainly suggests agreement. Furthermore, a change in the parameter $\tau$ in our quantization scheme, as discussed in \cref{Prp:tauchange}, appears to correspond to a gauge transformation~\cite{PINZUL2008284} in the context of deformation quantization. Given Proposition 3.5 in~\cite{dls_weyl}, one would expect the balanced quantization with $\tau=\tfrac{1}{2}$ to correspond to the unique gauge choice~\cite{felder2000deformation} where the trace of the commutator vanishes to all orders. However, these considerations are beyond the scope of this paper and will be the subject of future work.

Our present work is motivated by several physical applications. For example, it is well known that Weyl quantization and Wigner functions play an important role in quantum mechanics \cite{peres1995quantum,schroeck2013quantum,doi:10.1142/5287,10.1119/1.2957889}, semiclassical analysis \cite{teufel2003adiabatic,guillemin2013semi,zworski2022semiclassical}, kinetic theory \cite{groot1980relativistic,stephanov2012chiral,HIDAKA2022103989}, optics \cite{Schleich2001,torre2005}, and plasma physics \cite{tracy2014,10.1063/1.5095076}. However, the majority of these references are concerned with physics in flat spacetime, and our results enable the extension of these methods to the setting of manifolds and vector bundles. We mention here a few concrete examples. It is well known that the use of the Wigner function can be convenient for the description of high-frequency/semiclassical wave propagation, as shown, for example, in the non-relativistic setting in \cite{pulvirenti2006}. This offers a phase space alternative to geometric optics/WKB-type approximations, which break down at caustics. Furthermore, the use of the Wigner function provides a connection between high-frequency/semiclassical wave propagation and the Vlasov equation, in direct connection with kinetic theory \cite{groot1980relativistic}. In this context, our results could be applied to the study of wave propagation on manifolds and could help clarify connections with chiral kinetic theory in curved spacetime \cite{liu2018chiral,PhysRevD.105.096019,KAMADA2022104016,Liu_2020,HIDAKA2022103989}. Another possible application is the study of gravitational spin Hall effects \cite{GSHE_reviewCQG}. This concerns the spin-dependent propagation of the energy centroids of high-frequency/semiclassical localized wave packets that carry intrinsic angular momentum or spin, such as electromagnetic \cite{Oancea_2020,O2,PhysRevD.109.064020,PhysRevD.102.084013,PhysRevD.110.064020,SHE_QM1}, Dirac \cite{rudiger,audretsch,GSHE_Dirac,SHE_Dirac}, or linearized gravitational \cite{SHE_GW,GSHE_GW,GSHE_lensing,GSHE_lensing2,Frolov_2024} wave packets. In these references, the derivation of the gravitational spin Hall effects relies on geometric
optics/WKB-type approximations, which have known drawbacks. Using the Weyl calculus developed in the present paper, the theoretical foundations of these effects could be strengthened by deriving Egorov-type theorems for the dynamics of the energy centroids, similar to the cases presented in \cite{teufel2003adiabatic,Teufel2013,DeNittis2017}.

\subsection*{Structure of the paper}
In \cref{Sec:Preliminaries}, we review the necessary geometric and analytic preliminaries on pseudo-Riemannian geometry, the geometry of vector bundles, and symbol classes for pseudodifferential operators acting on sections of vector bundles.

\cref{Sec:WeylQuant} introduces the central concept of this work: the Weyl quantization of bundle-valued symbols. We then extend this construction to the more general $\tau$-quantization framework and analyze the asymptotic relations between symbols corresponding to different values of $\tau$. In the subsequent parts, we discuss the star product of the Weyl quantization, that is, the associative product at the level of symbols corresponding to the composition of the relevant Weyl-quantized operators. Furthermore, we derive the asymptotic expansion up to the third order in the semiclassical parameter. Some technical details of the proof of this expansion are provided in \cref{Appendix:StarProd}.

With the framework in place, we discuss some properties of the Weyl quantization in \cref{Sec:PropWeylQuant}. In particular, we establish, modulo smoothing remainders, the equivalence between formal self-adjoint Weyl symbols and formal self-adjoint linear operators, generalizing earlier work on $\mathbb{R}^{d}$ and on manifolds with scalar-valued operators. We also examine the \textit{Moyal equation}, which governs the evolution of the Wigner function in this setting.

Finally, in \cref{Sec:Examples}, we derive the Weyl symbols corresponding to a general class of second-order linear differential operators that act on sections of vector bundles. As an explicit application, we compute the Weyl symbols of various physically significant operators, such as the Maxwell and linearized Yang-Mills operators, the Dirac operator, and the linearized Einstein operator.

\subsection*{Acknowledgments} We are grateful for the financial support and hospitality of the Erwin Schrödinger International Institute for Mathematics and Physics, where the main part of this work was carried out during the Research in Teams program \textit{Weyl Pseudodifferential Calculus with Applications}. We also want to express our gratitude to Jérémie Joudioux for many helpful discussions and earlier work on the project, to Daniel Siemssen for sharing with us parts of the Mathematica code that were used in \cite{dls_weyl}, to Albert Much for insightful discussions on topics related to this project, and to Matteo Capoferri and Simone Murro for valuable comments and suggestions on the manuscript. The work of L.A. was funded in part by NSFC under grant W2431012. This research was funded in whole or in part by the Austrian Science Fund (FWF) \href{https://doi.org/10.55776/PIN9589124}{10.55776/PIN9589124}. This work has been produced within the activities of the INdAM-GNFM. G.S. gratefully acknowledges support from a PhD scholarship of the University of Genoa and partial support from the MIUR Excellence Department Project 2023–2027, awarded to the Department of Mathematics, University of Genoa.

\subsection*{Data availability and conflicts of interest}
The manuscript does not have associated data. The authors have no conflicts of interest to declare that are relevant to the content of this article.

\section{Geometric and analytic preliminaries}
\label{Sec:Preliminaries}
In this section, we recall the necessary geometric and analytic preliminaries on differential geometry, pseudo-Riemannian geometry, vector bundles with connections, and symbol classes for pseudodifferential calculus. This will help us fix the notation and conventions adopted in this article.

\subsection{Geometry of pseudo-Riemannian manifolds}
Throughout this article, let $(M,g)$ be a real, oriented, smooth pseudo-Riemannian manifold of finite dimension $d\in\mathbb{N}$, equipped with its Levi-Civita connection $\nabla$. We denote by $C^\infty(M)$ the ring of smooth complex-valued scalar functions on $M$. For a given vector bundle $E\xrightarrow{} M$ over $M$, we denote by $\Gamma(E)$ the $C^{\infty}(M)$-module of smooth sections and by $\Gamma_{\mathrm{c}}(E)$ the subspace of sections with compact support. Furthermore, we use the notation
\begin{align}
    \mathfrak{X}(M):=\Gamma(TM),\qquad \Omega^k(M):=\Gamma\bigg({\bigwedge}^{k}T^{\ast}M\bigg),\qquad\Omega^{k}(M,E):=\Gamma\bigg(E\otimes{\bigwedge}^{k}T^{\ast}M\bigg),
\end{align}
for the $C^{\infty}(M)$-modules of smooth vector fields, differential $k$-forms, and $E$-valued $k$-forms on $M$, respectively. The bundle projection of the cotangent bundle is denoted by $\pi\colon T^{\ast}M\to M$. For the Riemann and Ricci curvature tensors of $(M,g)$, we adopt the sign and index conventions
\begin{align} \label{eq:commutator_covd}
   \tensor{R}{_\alpha_\beta}:=\tensor{R}{^\lambda_\alpha_\lambda_\beta},\qquad  (\nabla_{\mu}\nabla_{\nu}-\nabla_{\nu}\nabla_{\mu})v^{\rho}=:\tensor{R}{^\rho_\sigma_\mu_\nu}v^{\sigma}\qquad\qquad \forall v\in\mathfrak{X}(M),
\end{align}
which are the same conventions used, for example, in the textbooks of Misner-Thorne-Wheeler \cite{MisnerThorneWheeler} and Wald \cite{Wald}.

Next, we review the definitions of \emph{Synge's world function} and the \emph{Van Vleck-Morette determinant}. For more details, refer to \cite{Poisson,visser}. A subset $C\subset M$ is called \textit{(geodesically) convex} if it is a normal neighborhood of all its points or, in other words, if there exists a unique geodesic segment contained entirely in $C$ connecting every pair of points $x,y\in C$. On such a set, we use the notation
\begin{align}
    y-x:=\mathrm{exp}_{x}^{-1}(y)\in T_{x}M ,
\end{align}
which is the tangent vector at $x$ of the unique geodesic segment from $x$ to $y$ that is contained entirely in $C$.

\begin{definition}[Synge's world function]
Let $(M,g)$ be a smooth pseudo-Riemannian manifold and $C$ a geodesically convex open set. Then, the map $\sigma_{g}:C\times C\to\mathbb{R}$ defined by
\begin{align}
	\sigma_{g}(x,y):=\frac{1}{2}g_{x}(y-x,y-x),
\end{align}
for all $x,y\in C$ is called {\normalfont{Synge's world function}}.
\end{definition}

The unique geodesic segment in $C$ connecting $x\in C$ with $y\in C$ is given by $\gamma(t)=\mathrm{exp}_{x}[t(y-x)]$, and since this defines a smooth map on $[0,1]\times C\times C$, we conclude that $\sigma_{g}$ defines a smooth function on $C\times C$. Furthermore, if $(M,g)$ is a Riemannian manifold or a Lorentzian manifold with signature $(-,+,\dots,+)$, then $\sigma_{g}$ is related to the signed squared geodesic length via
\begin{align}
    \sigma_{g}(x,y) = \frac{1}{2} s_{x,y} \bigg[ \int_{0}^{1} \sqrt{\vert g_{\gamma(t)}(\dot{\gamma}(t),\dot{\gamma}(t))\vert}\,dt \bigg]^{2} ,
\end{align}
where $s_{x,y}=-1$ if $\gamma$ is timelike and $s_{x,y}=1$ otherwise. In other words, $\sigma_{g}$ is half the signed squared geodesic length from $x$ to $y$.

If $(M,g)$ is \textit{globally} geodesically convex, which means that any pair of points of $M$ can be connected by a unique geodesic segment, then $\sigma_{g}$ is well defined globally on $M\times M$. This is, for example, the case for $M=\mathbb{R}^{d}$ with its flat metric or, more generally, for any complete, simply connected Riemannian manifold with nonpositive sectional curvature, as a consequence of the Cartan-Hadamard theorem \cite[Thm.~3.8]{Lang}. In general, the world function $\sigma_{g}$ depends on the choice of convex set $C$ and therefore need not be globally defined. However, one can always find an open neighborhood $\mathcal{N}\subset M\times M$ of the diagonal $\{(x,x)\in M\times M\mid x\in M\}$ on which $\sigma_{g}$ is well defined. Furthermore, for any other such neighborhood $\mathcal{N}^{\prime}$, there exists a neighborhood $\mathcal{N}^{\prime\prime}\subset\mathcal{N}\cap\mathcal{N}^{\prime}$ of the diagonal on which the corresponding world functions coincide (for details, see \cite{Moretti}). 

\begin{definition}[Van Vleck-Morette Determinant]
    Let $(M,g)$ be a smooth pseudo-Riemannian manifold, and let $\mathcal{N}$ be an open neighborhood of the diagonal on which $\sigma_{g}$ is well defined. Then, the {\normalfont{Van Vleck-Morette determinant}} is defined by
    \begin{align}
        \Delta(x,y):= \epsilon_g \frac{\mathrm{det} \left[ -
        \partial_x \partial_y \sigma_{g}(x,y) \right]}{\sqrt{\vert g(x)\vert}\sqrt{\vert g(y)\vert}},
    \end{align}
    for all $(x,y)\in\mathcal{N}$, where $g(x):=\mathrm{det}((g_{x})_{\mu\nu})$ denotes the metric determinant at $x$ and
$\epsilon_g:=\operatorname{sgn}\det((g_x)_{\mu\nu}) = (-1)^q$, with $q$ representing the number of negative eigenvalues of the metric.
\end{definition}

The Van Vleck-Morette determinant serves as the Jacobian determinant of the following coordinate transformation $TM\to M\times M$. On a sufficiently small neighborhood of the zero section, the local change of variables $(z,v)\mapsto (x,y):=(\mathrm{exp}_{z}(-v/2),\mathrm{exp}_{z}(v/2))$ satisfies
\begin{align}
    d\mu_M(x)d\mu_M(y) = \Delta(x,y)^{-1}\,d\mu_{TM}(z,v),
\end{align}
or equivalently
\begin{align}
    d\mu_{TM}(z,v)=\Delta(x,y)\,d\mu_M(x)d\mu_M(y),
\end{align}
where \(d\mu_{TM}(z,v):=d\mu_M(z)d\mu_{T_zM}(v)\).

\subsection{Geometry of vector bundles}

To describe fields, we consider smooth complex vector bundles
\begin{align}
    E\xrightarrow{\pi_E} M,\qquad F\xrightarrow{\pi_F} M .
\end{align}
We write $\operatorname{rank}_{\mathbb C}(E)=n$ and
$\operatorname{rank}_{\mathbb C}(F)=m$, with $m, n \in \mathbb{N}$. When needed, we also use additional
bundles, such as $G\xrightarrow{\pi_G}M$, equipped with the same type of
structures described below. 

We assume that all bundles under consideration are equipped with (similarly for $G$, although not explicitly stated):
\begin{itemize}
    \item[(i)]\textit{Pseudo-Hermitian fiber metrics}, that is, non-degenerate Hermitian sesquilinear forms
    \begin{align}
        \langle\cdot,\cdot\rangle_{E_{x}}\colon E_{x}\times E_{x}\to\mathbb{C},\qquad \langle\cdot,\cdot\rangle_{F_{x}}\colon F_{x}\times F_{x}\to\mathbb{C}, \qquad \forall x\in M ,
    \end{align}
    depending smoothly on $x$, with the convention that they are $\mathbb{C}$-linear in the second and antilinear in the first argument. On the level of sections, these bundle metrics induce non-degenerate Hermitian sesquilinear pairings, which we denote by
    \begin{subequations}
        \begin{align}
            (\Psi_{1},\Psi_{2})_{E}&:=\int_{M}\langle\Psi_{1}(x),\Psi_{2}(x)\rangle_{E_{x}}\,d\mu_{M}(x),\\
            (\Phi_{1},\Phi_{2})_{F}&:=\int_{M}\langle\Phi_{1}(x),\Phi_{2}(x)\rangle_{F_{x}}\,d\mu_{M}(x),
        \end{align}
     \end{subequations}
    for all $\Psi_1, \Psi_2 \in \Gamma(E)$ and $\Phi_1, \Phi_2 \in \Gamma(F)$ with
$\mathrm{supp}(\Psi_1) \cap \mathrm{supp}(\Psi_2)$ and 
$\mathrm{supp}(\Phi_1) \cap \mathrm{supp}(\Phi_2)$ compact, where $d\mu_M$ denotes the volume measure associated with the pseudo-Riemannian metric $g$ on $M$. We denote by $E^{\ast}$ and $F^{\ast}$ the corresponding \textit{dual bundles}, that is, the bundles whose fibers at $x\in M$ consist of the linear functionals of $E_{x}$ and $F_{x}$, respectively:
    \begin{align}
        E^{\ast}_{x}:=\{\varphi\:E_{x}\to\mathbb{C}\mid \varphi(\lambda v)=\lambda \varphi(v), \forall v \in E_x, \lambda \in \mathbb{C}\},
    \end{align} 
    and similarly for $F^{\ast}$. The pseudo-Hermitian metrics on $E$ and $F$ allow us to identify $\Gamma(E)\cong\Gamma(E^{\ast})$ and $\Gamma(F)\cong\Gamma(F^{\ast})$ ($\mathbb{C}$-antilinearly) via $\Psi\mapsto\langle \Psi, \cdot \rangle_{E}$ and $\Phi\mapsto\langle \Phi, \cdot \rangle_{F}$.
    \item[(ii)]\textit{Connections} that are compatible with the fiber metrics $\langle\cdot,\cdot\rangle_{E}$ and $\langle\cdot,\cdot\rangle_{F}$, respectively, denoted by
    \begin{align}
        \nabla^{E}\colon\Gamma(E)\to\Gamma(E\otimes T^{\ast}M),\qquad\nabla^{F}\colon\Gamma(F)\to\Gamma(F\otimes T^{\ast}M).
    \end{align}
    Compatibility means that, for all vector fields $X \in \mathfrak{X}(M)$ and sections $\Psi_1, \Psi_2 \in \Gamma(E)$,
\begin{align}
    X\langle \Psi_1,\Psi_2\rangle_E = \langle \nabla^E_X\Psi_1,\Psi_2\rangle_E +     \langle \Psi_1,\nabla^E_X\Psi_2\rangle_E,
\end{align}
and analogously for $F$ and any additional bundle under consideration.
\end{itemize}
The components of the curvature $2$-forms $E\in\Omega^{2}(M,\mathrm{End}(E))$ and $F\in\Omega^{2}(M,\mathrm{End}(F))$ are defined as
\begin{subequations}
\begin{align}
    (\nabla^{E}_{\mu}\nabla^{E}_{\nu}-\nabla^{E}_{\nu}\nabla^{E}_{\mu})\Psi^{A}=:\tensor{E}{^A_B_\mu_\nu}\Psi^{B}\qquad\qquad \forall \Psi\in\Gamma(E), \\
    (\nabla^{F}_{\mu}\nabla^{F}_{\nu}-\nabla^{F}_{\nu}\nabla^{F}_{\mu})\Phi^{C}=:\tensor{F}{^C_D_\mu_\nu}\Phi^{D}\qquad\qquad \forall \Phi\in\Gamma(F).
\end{align}    
\end{subequations}

For $x\in M$ and $v\in D_{x}M\subset T_{x}M$, we define the \textit{exponential map} by $\mathrm{exp}_{x}(v):=\gamma_{x,v}(1)$, as usual, where $\gamma_{x,v}\colon I \subseteq \mathbb{R} \to M$ denotes the unique geodesic with initial conditions $\gamma_{x,v}(0)=x$ and $\dot{\gamma}_{x,v}(0)=v$. This defines a map 
\begin{align}
    \mathrm{exp}_{x}\colon D_{x}M\to M,
\end{align}
where $D_{x}M\subset T_{x}M$ denotes the subset of vectors $v\in T_{x}M$ for which $\gamma_{x,v}$ is at least defined on $[0,1]$. Following~\cite[Sec.~2]{dls_weyl}, we will use the suggestive notation
\begin{align}
    x+v:=\mathrm{exp}_{x}(v)\in M
\end{align}
for all $x\in M$ and $v\in D_{x}M$. The \textit{parallel transport} operator of the bundle $(E,\nabla^{E})$ from $x\in M$ to $x+u$ for $u\in D_{x}M$ along the geodesic $t\mapsto \mathrm{exp}_{x}(tu)$ is denoted by
    \begin{align}
        \J{x+u}{x}\colon E_{x}\to E_{x+u}.
    \end{align}
Since it will typically be clear from the context, we use the same symbol
for the parallel transport operators associated with $(TM,\nabla)$,
$(E,\nabla^E)$, $(F,\nabla^F)$, and any additional bundles under consideration,
as well as for the induced transports on dual bundles and pullbacks. If
$a\in\Gamma(E^*\otimes F)\simeq \Gamma(\operatorname{Hom}(E,F))$, then the value
$a(z+v)$ transported to the fiber over $z$ will be denoted by
\begin{equation}
    \J{z}{z+v}\,a(z+v)\,\J{z+v}{z}
    \in E_z^*\otimes F_z \simeq \operatorname{Hom}(E_z,F_z),
\end{equation}
whenever $z+v$ is defined. Here, the left factor transports the $F$-fiber from
$z+v$ to $z$, while the right factor transports the input vector from $E_z$
to $E_{z+v}$; equivalently, it induces the dual pullback
$E_{z+v}^*\to E_z^*$.

\subsection{Horizontal and vertical covariant derivatives}
\label{Subsec:HorVertCD}
For a given vector bundle $E\to M$ equipped with a connection $\nabla^{E}$, we introduce a notion of \textit{horizontal} and \textit{vertical derivatives} for sections of the pull-back bundle $\pi^{\ast} E\to T^{\ast}M$, following the definitions considered in the articles \cite[Sec.~2.4]{dls_weyl}, \cite{MR2191866,MR2186590}, and the monograph \cite[Sec.~3.5]{sharafutdinov2012integral}. The same construction and notation will be used for all the other bundles, including induced bundles, tensor products, duals, and pullbacks.

The Levi-Civita connection on $M$ induces a linear connection on $T^*M$. Equivalently, the corresponding \textit{Ehresmann connection} is a vector subbundle $H(T^{\ast}M)$ such that
\begin{align}\label{eq:decompCoTan}
    T(T^{\ast}M)=V(T^{\ast}M)\oplus H(T^{\ast}M) ,
\end{align}
where $V(T^{\ast}M):=\mathrm{ker}(d\pi)$ denotes the \emph{vertical subbundle}, i.e., the bundle with fibers $V_{(x,p)}(T^{\ast}M)=T_{(x,p)}(T_{x}^{\ast}M)$. For a vector field $X\in\mathfrak{X}(M)$, we denote by $X^{h}\in\Gamma(H(T^{\ast}M))$ its horizontal lift and by $X^{v}\in\Gamma(V(T^{\ast}M))$ the vertical lift of the covector $g(X,\cdot)$ \cite{1967185}. The horizontal and vertical lifts allow us to define the horizontal and vertical covariant derivatives that act on sections of the pullback bundle $\pi^{\ast}E\to T^{\ast}M$. 

\begin{definition}[Horizontal and vertical derivative]
    Let $E\to M$ be a smooth vector bundle equipped with a connection $\nabla^{E}$. The \textit{horizontal and vertical covariant derivatives} induced by $\nabla^{E}$ are the operators
    \begin{subequations}
        \begin{align}
           & \overset{\mathrm{h}}{\nabla}\colon\Gamma(\pi^{\ast}E)\to \Gamma(\pi^{\ast}(E\otimes T^{\ast}M)),  &\overset{\mathrm{h}}{\nabla}_{X}\Psi:=(\nabla^{\pi^{\ast}E})_{X^{h}}\Psi,\\
            &\overset{\mathrm{v}}{\nabla}\colon\Gamma(\pi^{\ast}E)\to \Gamma(\pi^{\ast}(E\otimes TM)),  &\overset{\mathrm{v}}{\nabla}_{X}\Psi:=(\nabla^{\pi^{\ast}E})_{X^{v}}\Psi,
        \end{align}
    \end{subequations}
    where $X\in\mathfrak{X}(M)$ and $\nabla^{\pi^{\ast}E}:=\pi^{\ast}\nabla^{E}$  denotes the pull-back connection on the bundle $\pi^{\ast}E\to T^{\ast}M$. 
\end{definition}

Let us now turn to a local description of these operators in coordinates. We choose a smooth local frame $e_{A}$ of $E$, which provides a local frame $e_{A}^{\prime}:=\tensor{e}{_A}\circ\pi$ for the pull-back bundle $\pi^{\ast}E$. Furthermore, we choose local coordinates $x = (x^{1},\dots,x^{d})$ on $M$ and denote the corresponding induced coordinates on the manifold $T^{\ast}M$ by $(x,p) = (x^{1},\dots,x^{d};p_{1},\dots,p_{d})$. Now, any section of $\pi^{\ast}E$ can be written locally as $\Psi(x,p)=\Psi^{A}(x,p)e_{A}^{\prime}(x,p)$ for $(x,p)\in T^{\ast}M$ and component functions $\Psi^{A}\in C^{\infty}(T^{\ast}M)$. Given a vector field $X\in \mathfrak{X}(M)$, locally written as $X=X^{\mu}\partial_{x^{\mu}}$, the corresponding horizontal and vertical lifts are given by \cite{1967185} 
\begin{align}
    X^{h}=X^{\mu}(\partial_{x^{\mu}}+p_{\sigma}\Gamma^{\sigma}_{\mu \nu}\partial_{p_{\nu}}) \qquad\text{and}\qquad X^{v}=X_{\mu}\partial_{p_{\mu}},
\end{align}
where $X_\mu = g_{\mu \nu} X^\nu$ and $\Gamma_{\mu \nu}^{\sigma}$ denote the Christoffel symbols of $\nabla$. It follows that
\begin{subequations}
\begin{align}
    \hnabla_{X}\Psi &= (\pi^{\ast}\nabla^{E})_{X^{h}}\Psi= d\Psi^{A}(X^{h})e_{A}^{\prime}+\Psi^{A}(\pi^{\ast}\nabla^{E})_{X^{h}} e_A^{\prime} \nonumber\\
    &= X^{\mu} \left(\partial_{x^{\mu}}\Psi^{A}+p_{\sigma}\Gamma^{\sigma}_{\mu \nu}\partial_{p_{\nu}}\Psi^{A}+\omega\indices{^A_{B \mu}}\Psi^{B} \right) e_{A}^{\prime}, \\
    \vnabla_{X}\Psi &= (\pi^{\ast}\nabla^{E})_{X^{v}}\Psi=d\Psi^{A}(X^{v})e_{A}^{\prime}+\Psi^{A}(\pi^{\ast}\nabla^{E})_{X^{v}} e_{A}^{\prime} \nonumber\\
    &= X_{\mu} \left( \partial_{p_{\mu}}\Psi^{A} \right) e_{A}^{\prime} ,
\end{align}
\end{subequations}
where $\omega\indices{^A_{B \mu}}$ denotes the connection coefficients of $\nabla^{E}$, i.e.~$\nabla^{E}_{\partial_{x^\mu}}e_{B}=\omega\indices{^A_{B \mu}}e_{A}$, and we used that the pullback frame $e_{A}^{\prime}(x,p)$ is constant in the vertical directions. In summary, the coordinate expressions of the horizontal and vertical covariant derivatives are given by
\begin{subequations}
\begin{align}
\begin{split}
    \overset{\mathrm{h}}{\nabla}_{\mu}\Psi^{A}&=\nabla^{E}_{\mu}\Psi^{A}+\Gamma_{\mu \nu}^{\sigma}p_{\sigma}\partial_{p_{\nu}}\Psi^{A}=\partial_{x^{\mu}}\Psi^{A}+\omega\indices{^A_{B \mu}}\Psi^{B}+\Gamma_{\mu \nu}^{\sigma}p_{\sigma}\partial_{p_{\nu}}\Psi^{A},
    \end{split}\\
    \begin{split}
    \overset{\mathrm{v}}{\nabla}{}^\mu \Psi^{A}&=\partial_{p_{\mu}}\Psi^{A} .
    \end{split}
\end{align}
\end{subequations}
When these operators are extended to pullbacks of tensor products using the
induced connections on all tensor indices, their commutators satisfy
\begin{subequations}
\begin{align}
    [\hnabla_\mu,\vnabla^\alpha]\Psi &=0,\qquad \qquad \quad
    [\vnabla^\alpha,\vnabla^\beta]\Psi=0,\\
    (\hnabla_\mu\hnabla_\nu-\hnabla_\nu\hnabla_\mu)\Psi^A
    &=
    R\indices{^\alpha_{\beta\mu\nu}}p_\alpha \vnabla^\beta\Psi^A
    +E\indices{^A_{B\mu\nu}}\Psi^B ,
\end{align}
\end{subequations}
where $E\in\Omega^{2}(M,\mathrm{End}(E))$ denotes the curvature tensor of $(E,\nabla^{E})$.  Furthermore, for the tautological covector $p=p_\mu dx^\mu\in\pi^*(T^*M)$, one has
\begin{align}
    \hnabla_\alpha p_\mu=0,\qquad
    \hnabla_\alpha p^\mu=0,\qquad
    \vnabla^\alpha p_\mu=\delta^\alpha_\mu,\qquad
    \vnabla^\alpha p^\mu=g^{\alpha\mu}.
\end{align}

The following lemma is due to Sharafutdinov.

\begin{lemma}[\protect{\cite[Lemma~7.5]{MR2191866}}] \label{lemma:derivative}
    Let $E$ be a vector bundle over $M$ with connection $\nabla^{E}$. Then, for any $k \in \mathbb{N}$ and for any multi-index $\boldsymbol{\alpha} \in \mathbb{N}^{k}$, it holds that
    \begin{align}
      \hnabla^{(\boldsymbol{\alpha})} u(x,p) = \partial_{v}^{\boldsymbol{\alpha}} \bigg[ \J{x}{x+v}u(x+v,\J{x+v}{x}p) \bigg] \bigg\vert_{v=0}
    \end{align}
    for all $u\in \Gamma(\pi^{\ast}E)$, where $\hnabla^{(\boldsymbol{\alpha})}$ denotes the symmetrized horizontal derivative
    \begin{align}
        \hnabla^{(\boldsymbol{\alpha})}:=\frac{1}{k!}\sum_{\sigma\in \mathfrak{S}_{k}} \hnabla_{\sigma(\boldsymbol{\alpha}(1))} \dots \hnabla_{\sigma(\boldsymbol{\alpha}(k))},
    \end{align} 
    with $\mathfrak{S}_{k}$ being the symmetric group of order $k$.
\end{lemma}

\subsection{Symbol classes and (semiclassical) pseudodifferential operators}

In this section, we briefly review some well-known symbol classes for (semiclassical) pseudodifferential operators taking values in arbitrary vector bundles. First of all, let us recall the standard notion of (uniform) symbols introduced by Kohn-Nirenberg \cite{KohnNirenberg} and Hörmander \cite{Hormander}; see also \cite{Sjorstrand,HormanderI,Shubin} for textbook accounts.

\begin{definition}[Symbol classes]
    Let $k\in\mathbb{R}$ and $U\subset\mathbb{R}^{d}$ be an open set. We denote by $\mathcal{S}^{k}(U,\mathbb{C}^{m \times n})$ the set of all smooth functions $a\in C^{\infty}(U\times\mathbb{R}^{d},\mathbb{C}^{m \times n})$ with the property that for all compact subsets $K\subset U$ and any multi-indices $\boldsymbol{\alpha}, \boldsymbol{\beta} \in \mathbb{N}^{d}$, there is a constant $C_{K,\boldsymbol{\alpha}, \boldsymbol{\beta}} > 0$ such that
    \begin{align}
        \vert \partial_{x}^{\boldsymbol{\alpha}} \partial_{p}^{\boldsymbol{\beta}} a(x,p)\vert\leq C_{K, \boldsymbol{\alpha}, \boldsymbol{\beta}} (1+\vert p\vert)^{k-\vert\boldsymbol{\beta}\vert} \qquad \forall \, (x,p)\in K\times\mathbb{R}^{d} ,
    \end{align}
    where $\vert\cdot\vert$ denotes an arbitrary norm on $\mathbb{C}^{m \times n}$. Then, a {\normalfont{symbol of order $k\in\mathbb{R}$}} is a section $a\in\Gamma(T^{\ast}M,\pi^{\ast}(E^{\ast}\otimes F))$, such that in any local chart $(U,\varphi)$ with the induced cotangent coordinates and local trivializations of $E$, and $F$ (and hence also of $\pi^{\ast}E$ and $\pi^{\ast}F$), its local representative $a_{U}\in C^{\infty}(\varphi(U)\times\mathbb{R}^{d},\mathbb{C}^{m \times n})$ is such that $a_{U}\in\mathcal{S}^{k}(\varphi(U),\mathbb{C}^{m \times n})$. The space of all symbols of order $k$ is denoted by
    \begin{align}
        \mathcal{S}^{k}(M;E,F),
    \end{align}
    or simply by $\mathcal{S}^{k}$ when it is clear from the context which manifold and bundles are used. Furthermore, we introduce the space of \textit{smoothing symbols}, which is defined as
    \begin{align}
        \mathcal{S}^{-\infty}(M;E,F):= \bigcap_{k\in\mathbb{R}}\mathcal{S}^{k}(M;E,F).
    \end{align}
\end{definition}

In other words, a symbol $a$ of order $k$ has at most polynomial growth of order $k$ in the cotangent variable $p$, uniformly with respect to $x$. One can easily check that the symbols transform covariantly under the change of coordinates and local trivializations, which shows that the above definition is purely geometric; see \cite[Lemma~5.33]{Hintz}.

One may also consider slightly more general classes of symbols by replacing $(1+\vert p\vert)^{k-\vert\boldsymbol{\beta}\vert}$ in the defining estimate above with $(1+\vert p\vert)^{k-\rho\vert\boldsymbol{\beta}\vert+\delta\vert\boldsymbol{\alpha}\vert}$ for given parameters $\rho,\delta\in [0,1]$ and $1-\rho\leq\delta<\rho$, leading to symbol classes usually denoted by $\mathcal{S}^{k}_{\rho,\delta}$ (see e.g.~\cite{HormanderFourier,Sjorstrand}). In other words, one also allows for a decrease in the order gained by differentiation in the base variable. Although symbols in this general class have many applications, e.g.~in scattering theory, we shall restrict our attention to symbols of the class $(\rho,\delta)=(1,0)$ throughout this text.

\begin{remark}\label{Rem:TopSym} The space $\mathcal{S}^{k}(U,\mathbb{C}^{m \times n})$ carries a natural structure of a Fréchet space whose topology is induced by the family of seminorms indexed by compact subsets $K \subset U$ and multi-indices $\alpha,\beta\in\mathbb{N}^d$, defined via the optimal constants in the defining estimates:
\begin{align} p_{K,\boldsymbol{\alpha},\boldsymbol{\beta}}(a):=\sup_{(x,p)\in K\times\mathbb{R}^{d}}\frac{\vert \partial_{x}^{\boldsymbol{\alpha}} \partial_{p}^{\boldsymbol{\beta}} a(x,p)\vert}{ (1+\vert p\vert)^{k-\vert\boldsymbol{\beta}\vert}}.
\end{align} 
Using local charts and trivializations, we equip $\mathcal{S}^{k}(M;E,F)$ with the corresponding Fréchet topology. More explicitly, we choose a countable, locally finite atlas of $M$ that trivializes the bundles $E$ and $F$, together with compact exhaustions in each chart, and we use the above seminorms for the local representatives $a_U$. Different choices of atlas, trivializations, and compact exhaustions give equivalent topologies. \end{remark}

Let us now briefly recall some basic concepts of pseudodifferential calculus. On an open subset $U\subset\mathbb{R}^{d}$, a \textit{pseudodifferential operator} of order $k\in\mathbb{R}$ is a linear and continuous operator of the form
\begin{align}
    \hat{A}\colon C^{\infty}_{\mathrm{c}}(U, \mathbb{C}^n)\to C^{\infty}(U,\mathbb{C}^m),\qquad \hat{A}\psi(x)= \frac{1}{(2\pi)^d}\int_{\mathbb{R}^{d}}\int_{U}e^{ip \cdot (x-y)}a(x,y,p)\psi(y)\,d^dy\,d^dp ,
\end{align}
for a given smooth $a(x, y, p)$ satisfying the symbol estimates in the variable
$p$, uniformly for $(x,y)$ in compact subsets of $U\times U$, and where the integral must be understood in the sense of \textit{oscillatory integrals}. The space of all such operators is denoted by $\Psi^{k}(U, \mathbb{C}^{m \times n})$. Pseudodifferential operators with symbols in $\mathcal{S}^{-\infty}$ are in one-to-one correspondence with operators whose Schwartz kernels are smooth. Equivalently, smoothing operators are exactly those operators that continuously extend to operators of the form $\hat{A}\colon\mathcal{E}^{\prime}(U, \mathbb{C}^{n})\to C^{\infty}(U, \mathbb{C}^{m})$, where $\mathcal{E}^{\prime}(U, \mathbb{C}^{n})$ denotes the space of compactly supported distributions on $U$, which is the origin of the name ``smoothing''. We denote the space of all smoothing operators by $\Psi^{-\infty}(U, \mathbb{C}^{m \times n})$, or simply $\Psi^{-\infty}$. Note that in this general definition, we allow a symbol to depend on both the $x$ and $y$ variables. However, it is well known that for every pseudodifferential operator, there exists a symbol $\sigma_{\hat{A}}\in \mathcal{S}^{k}(U, \mathbb{C}^{m \times n})$, called the \textit{complete symbol}, that is unique modulo $\mathcal{S}^{-\infty}$, such that $\hat{A}$ is the \textit{Kohn-Nirenberg quantization} of $\sigma_{\hat{A}}$:
\begin{align}
    \hat{A}\psi(x)= \frac{1}{(2\pi)^d} \int_{\mathbb{R}^{d}}\int_{U}e^{ip \cdot (x-y)}\sigma_{\hat{A}}(x,p)\psi(y)\,d^dy\,d^dp .
\end{align}
More details can be found in \cite{Shubin,Sjorstrand}. 

If $\varphi\colon U\to V$ is a diffeomorphism between two open subsets $U,V\subset\mathbb{R}^{d}$, then $\hat{A}_{\varphi}:=\hat{A}(\bullet\circ\varphi)\circ\varphi^{-1}$ defines a pseudodifferential operator of the same order on $V$; see \cite[Thm.~4.1]{Shubin} or \cite[Thm.~3.9]{Sjorstrand}. In particular, this allows one to define pseudodifferential operators on manifolds: A linear and continuous operator of the form $\hat{A}\colon\Gamma_{\mathrm{c}}(E)\to\Gamma(F)$ is called a \textit{pseudodifferential operator} of order $k\in\mathbb{R}$ if its Schwartz kernel is smooth outside the diagonal $\delta:=\{(x,x)\mid x\in M\}$ and if, in any local coordinate chart $\varphi\:U\to\mathbb{R}^{d}$ with $U\subset M$ open, which also trivializes the corresponding vector bundles, its local representative $\hat{A}_{\varphi}:=\hat{A}(\bullet\circ\varphi)\circ\varphi^{-1}$ has the form
\begin{align}
    \hat{A}_{\varphi}\psi(x)= \frac{1}{(2\pi)^d} \int_{\mathbb{R}^{d}}\int_{\varphi(U)}e^{ip \cdot (x-y)}a_{U}(x,y,p)\psi(y)\,d^dy\,d^dp
\end{align}
for some smooth $a_{U}(x, y, p)$ satisfying symbol estimates in $p$, uniformly
for $(x,y)$ in compact subsets of $\varphi(U)\times\varphi(U)$. We denote the space of all pseudodifferential operators of order $k$ on $M$ by $\Psi^{k}(M;E,F)$, or simply by $\Psi^{k}$ if there is no risk of confusion. Analogously, we denote the space of smoothing operators by $\Psi^{-\infty}(M;E,F)$ or simply by $\Psi^{-\infty}$. As before, the operators in $\Psi^{-\infty}(M;E,F)$ are exactly those linear and continuous operators $\hat{A}\colon\Gamma_{\mathrm{c}}(E)\to\Gamma(F)$ with a \emph{smooth} Schwartz kernel. We refer to \cite[Sec.~4]{Shubin} and \cite[Sec.~5]{Hintz} for more details on pseudodifferential operators on manifolds.

Since the symbol classes are globally well defined and behave covariantly under the change of coordinates, one can also define a \textit{quantization map} $\mathcal{S}^{k}(M;E,F)\to \Psi^{k}(M;E,F)$ by choosing a partition of unity and gluing together local constructions. A detailed discussion can be found in \cite[Sec.~5.6]{Hintz}.

\begin{example}[Linear Differential Operators] \label{example:diff_op}
    Let $\hat{A}\colon\Gamma(E)\to\Gamma(F)$ be a linear differential operator of order $k \in \mathbb{N}$, which is an operator that in a local chart $U\subset M$ and local trivializations can be written as a linear differential operator $\hat{A}\colon C^{\infty}(U,\mathbb{C}^{n})\to C^{\infty}(U,\mathbb{C}^{m})$ of the form 
    \begin{align}
        \hat{A}=\sum_{\vert \boldsymbol{\alpha} \vert \leq k} f^{\boldsymbol{\alpha}}(x) \frac{1}{i^{\vert\boldsymbol{\alpha}\vert}}\frac{\partial}{\partial x^{\boldsymbol{\alpha}}},
    \end{align}
    for smooth coefficient functions $f^{\boldsymbol{\alpha}} \in C^{\infty}(U,\mathbb{C}^{m \times n})$. We denote the space of these linear differential operators of order at most $k$ by $\mathrm{Diff}^{k}(M;E,F)$.
    By the Theorem of Peetre (see e.g.~\cite[Thm.~6.2]{Kahn}), these are exactly the operators that are \emph{local}, that is, $\mathrm{supp}(\hat{A}\Psi)\subset\mathrm{supp}(\Psi)$ for all $\Psi\in\Gamma(E)$. In each local chart and trivialization, the total symbol of $\hat A$ is the polynomial
    \begin{align}
        a(x,p):=a[\hat{A}](x,p):=\sum_{\vert \boldsymbol{\alpha} \vert \leq k} f^{\boldsymbol{\alpha}}(x) p_{\boldsymbol{\alpha}},
    \end{align}
    which belongs to the local symbol class $\mathcal S^k(U,\mathbb C^{m\times n})$.
\end{example}

Next, let us introduce a more general class of symbols that is particularly well suited for the purpose of \emph{semiclassical analysis}, namely symbols $a_{h}(x,p)$ that are allowed to depend additionally on an external \emph{semiclassical parameter} $h$. For details, we refer to the extensive discussion in \cite[App.~E]{Dytalov}, \cite[Sec.~3.5 and 5.15]{Hintz}, and \cite{zworski2022semiclassical}.

\begin{definition}[Semiclassical symbols]
Fix $h_0\in(0,1]$. For $k\in\mathbb R$, we denote by $\mathcal S_h^k(M;E,F)$ the space of all families $a=(a_h)_{h \in (0,h_0]}$, with $a_h\in\mathcal S^k(M;E,F)$ for each $h \in (0,h_0]$, such that $(a_h)_{h \in (0,h_0]}$ is bounded in the Fréchet topology of $\mathcal S^k(M;E,F)$ introduced in \cref{Rem:TopSym}. We call any $a\in\mathcal{S}_{h}^{k}(M;E,F)$ \emph{a semiclassical symbol of order $k$}.

\noindent For $l\in\mathbb R$, we set
\begin{equation}
    h^{-l}\mathcal S_h^k(M;E,F) := \{ a=(a_h)_{h \in (0,h_0]} \mid (h^l a_h)_{0<h\le h_0} \in\mathcal S_h^k(M;E,F)\}.
\end{equation}
Equivalently, in local coordinates and trivializations, for every compact
$K\subset U$ and all multi-indices $\boldsymbol\alpha,\boldsymbol\beta$, there is
a constant $C_{K,\boldsymbol\alpha,\boldsymbol\beta}>0$, independent of $h$, such that
\begin{equation}
    \left| \partial_x^{\boldsymbol\alpha}\partial_p^{\boldsymbol\beta}a_h(x,p) \right| \le C_{K,\boldsymbol\alpha,\boldsymbol\beta} h^{-l}(1+|p|)^{k-|\boldsymbol\beta|} \qquad \forall \, (x,p)\in K\times\mathbb{R}^{d},\,\, h \in (0, h_0].
\end{equation}

\noindent Finally, we define the set of semiclassical smoothing symbols as
\begin{equation}
    h^\infty\mathcal S_h^{-\infty}(M;E,F)
    :=
    \bigcap_{k,l\in\mathbb R}h^{-l}\mathcal S_h^k(M;E,F).
\end{equation}
Thus $a\in h^\infty\mathcal S_h^{-\infty}$ when all local
$\mathcal S^k$-seminorms of $a_h$, for every $k$, are
$\mathcal O(h^N)$ for every $N\in\mathbb N$.
\end{definition}

On an open subset $U\subset\mathbb{R}^{d}$, the \emph{semiclassical Kohn-Nirenberg quantization} of a symbol $a \in h^{-l}\mathcal{S}_h^{k}(U, \mathbb{C}^{m \times n})$ is the family of operators
\begin{equation}
   \hat{A}_h\colon C_c^{\infty}(U,\mathbb{C}^{n})\to C^{\infty}(U,\mathbb{C}^{m}), \quad \hat{A}_{h}\psi(x)=\frac{1}{(2\pi h)^d}\int_{\mathbb{R}^{d}}\int_{U}e^{\frac{i}{h}p \cdot (x-y)}a_{h}(x,p)\psi(y)\,d^dy\,d^dp.
\end{equation}
As in the non-semiclassical setting discussed above, one can also start with $a_h(x,y,p)$, and the complete symbol $a_h(x,p)$ is unique modulo $h^{\infty}\mathcal{S}_h^{-\infty}$. It is straightforward to verify that this definition behaves covariantly under the change of coordinates. Hence, there is a well-defined notion of semiclassical pseudodifferential operators on manifolds. We denote the space of all such operators by $h^{-l}\Psi_h^{k}(M;E,F)$ or simply $h^{-l}\Psi_h^{k}$ if the manifold and bundles are clear in the context. The space of semiclassical smoothing operators is denoted by $h^{\infty}\Psi_{h}^{-\infty}(M;E,F)$ and consists of families whose Schwartz kernels are smooth and whose $C^{\infty}$-seminorms on compact subsets of $M \times M$ are $\mathcal{O}(h^{N})$ for every $N\in\mathbb{N}$. When $h_0=1$, the local formula at $h=1$ reduces to the ordinary Kohn-Nirenberg quantization.

\begin{example}[Semiclassical Differential Operators]\label{example:semi_diff_op}
Let $k \in \mathbb{N}$ and fix $h_0\in(0,1]$. A \emph{semiclassical differential operator} of order $k$ from $E$ to $F$ is a linear and local family of operators $(\hat{A}_{h})_{h \in (0,h_0]}\colon\Gamma(E)\to\Gamma(F)$. In a local chart $(U,\varphi)$ and local trivializations of the bundles $E$ and $F$, this can be written in the form
\begin{align}
        \hat{A}_{h}=\sum_{\vert \boldsymbol{\alpha} \vert \leq k} f^{\boldsymbol{\alpha}}_{h}(x) \frac{h^{\vert\boldsymbol{\alpha}\vert}}{i^{\vert\boldsymbol{\alpha}\vert}}\frac{\partial}{\partial x^{\boldsymbol{\alpha}}},
\end{align}
    for a $h$-dependent family of functions $f_{h}^{\boldsymbol{\alpha}} \in C^{\infty}(U,\mathbb{C}^{m \times n})$ satisfying
    \begin{equation}
        \sup_{h\in(0, h_0]} \sup_{x \in K} | \partial_x^{\boldsymbol{\beta}} f_h^{\boldsymbol{\alpha}}(x) | < \infty,
    \end{equation}
    for every compact set $K \subset U$, every multi-index $\boldsymbol{\beta}$, and every $|\boldsymbol{\alpha}| \leq k$.
    We denote the space of such operators by $\mathrm{Diff}_{h}^{k}(M;E,F)$. More generally, $h^{-l}\mathrm{Diff}_{h}^{k}(M;E,F)$ consists of those families for which $(h^{l}\hat{A}_{h})_{h\in(0,h_{0}]}$ is bounded in the coefficient seminorms appearing above.
    
    In each local chart and trivialization, the total symbol of $\hat{A}_h$ is
    \begin{align}
        a_{h}(x,p):=\sum_{\vert \boldsymbol{\alpha} \vert \leq k} f_{h}^{\boldsymbol{\alpha}}(x) p_{\boldsymbol{\alpha}},
    \end{align}
    which is an element of $\mathcal{S}_{h}^{k}(U,\mathbb{C}^{m \times n})$. If instead we take $\hat{A}_h \in h^{-l}\mathrm{Diff}_{h}^{k}(M;E,F)$, then its total semiclassical symbol belongs to $h^{-l}\mathcal{S}_{h}^{k}(M;E,F)$. In other words, if $(h^{l}a_{h})_{h\in (0,h_{0}]}$ is a bounded family in $\mathcal{S}^{k}(M;E,F)$ with respect to its Fréchet topology, i.e.
    \begin{align}
        \sup_{h\in (0,h_0]}\sup_{(x,p)\in K\times \mathbb{R}^d}h^{l}\frac{\vert \partial_{x}^{\boldsymbol{\alpha}} \partial_{p}^{\boldsymbol{\beta}} a_{h}(x,p)\vert}{ (1+\vert p\vert)^{k-\vert\boldsymbol{\beta}\vert}}<\infty,
    \end{align}
    for all $\boldsymbol{\alpha},\boldsymbol{\beta}\in\mathbb{N}^{d}$ and compact $K\subset U$, we call $\hat{A}_{h} \in h^{-l}\mathrm{Diff}_{h}^{k}(M;E,F)$ a \emph{semiclassical differential operator of order $k$ and semiclassical order $l$.}

    In particular, every differential operator $\hat{A} \in \mathrm{Diff}^{k}(M;E,F)$ as in \cref{example:diff_op} can be regarded as a semiclassical differential operator of semiclassical order $k$. Locally, we can write 
    \begin{equation}
        \hat{A}=\sum_{\vert \boldsymbol{\alpha} \vert \leq k} f^{\boldsymbol{\alpha}}(x) \frac{1}{i^{\vert\boldsymbol{\alpha}\vert}}\frac{\partial}{\partial x^{\boldsymbol{\alpha}}} = h^{-k} \sum_{\vert \boldsymbol{\alpha} \vert \leq k} h^{k-{\vert\boldsymbol{\alpha}\vert}} f^{\boldsymbol{\alpha}}(x) \frac{h^{\vert\boldsymbol{\alpha}\vert}}{i^{\vert\boldsymbol{\alpha}\vert}}\frac{\partial}{\partial x^{\boldsymbol{\alpha}}}.
    \end{equation}
    Thus, $h^k \hat{A} \in \mathrm{Diff}_{h}^{k}(M;E,F)$, or equivalently $\hat{A} \in h^{-k} \mathrm{Diff}_{h}^{k}(M;E,F)$. The corresponding semiclassical total symbols are
    \begin{subequations}
    \begin{align}
        a_h[h^k \hat{A}](x,p) &= \sum_{\vert \boldsymbol{\alpha} \vert \leq k} h^{k-{\vert\boldsymbol{\alpha}\vert}} f^{\boldsymbol{\alpha}}(x) p_{\boldsymbol{\alpha}} \in \mathcal{S}_{h}^{k}(M;E,F), \\
        a_h[\hat{A}](x,p) &= \sum_{\vert \boldsymbol{\alpha} \vert \leq k} h^{-{\vert\boldsymbol{\alpha}\vert}} f^{\boldsymbol{\alpha}}(x) p_{\boldsymbol{\alpha}} \in h^{-k} \mathcal{S}_{h}^{k}(M;E,F).
    \end{align}
    \end{subequations}
\end{example}

Finally, we recall the notion of an \emph{asymptotic symbol expansion}, which allows us to represent certain semiclassical symbols as formal power series in the semiclassical parameter $h$. For a given semiclassical symbol $a:=(a_{h})_{h\in (0,h_0]}\in\mathcal{S}_{h}^{k}(M;E,F)$, we write 
\begin{align}\label{eq:AsymSumDef1}
    a\sim\sum_{i=0}^{\infty}h^{i}a_{i}
\end{align}
for a sequence of $h$-independent symbols $a_{i}\in\mathcal{S}^{k-i}(M;E,F)$, if for all $N\in\mathbb{N}$ it holds that 
\begin{align}\label{eq:AsymSumDef2}
    a-\sum_{i=0}^{N-1}h^{i}a_{i}\in h^{N}\mathcal{S}_{h}^{k-N}(M;E,F) .
\end{align}
The formal series $\sum_{i=0}^{\infty}h^{i}a_{i}$ is called the \emph{asymptotic sum} of $a$. It is well known that such an asymptotic expansion, if it exists, characterizes the symbol $a$ uniquely in the following sense: given any sequence of symbols $(a_{i})_{i\in\mathbb{N}}$ with $a_{i}\in\mathcal{S}^{k-i}(M;E,F)$, there exists a symbol $a\in\mathcal{S}_{h}^{k}(M;E,F)$, unique modulo $h^{\infty}\mathcal{S}_h^{-\infty}(M;E,F)$, such that $a\sim\sum_{i=0}^{\infty}h^{i}a_{i}$, i.e.~such that~\cref{eq:AsymSumDef2} is satisfied. More generally, for $a\in h^{-l}\mathcal S_h^k(M;E,F)$ and
$h$-independent symbols $a_i\in\mathcal S^{k-i}(M;E,F)$, we write 
\begin{equation}
    a\sim h^{-l}\sum_{i=0}^\infty h^i a_i
\end{equation}
if for all $N\in\mathbb{N}$, it holds that 
\begin{equation}
    a-h^{-l}\sum_{i=0}^{N-1}h^i a_i
\in h^{N-l}\mathcal S_h^{k-N}(M;E,F).
\end{equation}

An important subset of semiclassical symbols is the \emph{polyhomogeneous} semiclassical symbols (e.g.~\cite{Dytalov,HormanderIII}). These admit an asymptotic expansion as in \cref{eq:AsymSumDef1}, with coefficients $a_{i}\in\mathcal{S}^{k-i}(M;E,F)$ that are \emph{positively homogeneous of degree $k-i$} in the cotangent variable:
\begin{align}
    a_{i}(x,\lambda p)=\lambda^{k-i}a_{i}(x,p),
\end{align}
for all $\lambda>0$, whenever $|p|\ge 1$ and $|\lambda p|\ge 1$. In the following, we denote the subspace of polyhomogeneous semiclassical symbols by 
\begin{align}
    \mathcal{S}_{h,\mathrm{poly}}^{k}(M;E,F)\subset\mathcal{S}^{k}_{h}(M;E,F) .
\end{align}
In modern literature, polyhomogeneous symbols are also frequently referred to as ``classical'' (e.g.~\cite{Hintz,Shubin,Sjorstrand}). In this context, the term ``classical'' is unrelated to the semiclassical parameter $h$.

\section{Weyl quantization on vector bundles}
\label{Sec:WeylQuant}
In this section, we introduce the Weyl quantization for symbols associated with operators acting on sections of vector bundles. To define the quantization, we start from a vector bundle generalization of the well-known Wigner function from quantum mechanics. We provide an explicit relation between Weyl symbols and Schwartz kernels, and we briefly discuss alternative quantization schemes beyond the midpoint (Weyl) prescription. Finally, we define the star product, which is the associative product of symbols corresponding to the composition of operators, and establish its asymptotic expansion up to the third order in the semiclassical parameter.

\subsection{Definition of the Weyl quantization}

Throughout this section, $h \in (0, h_0]$ denotes the semiclassical parameter. Let $\mathcal{N}$ be an open neighborhood of the diagonal $\{(x,x)\mid x\in M\}\subset M\times M$ on which the geodesic midpoint construction is well defined, and $\chi\in C^{\infty}(M\times M)$ a cut-off function supported in $\mathcal{N}$ such that $\chi=1$ on another geodesically convex neighborhood of the diagonal whose closure is contained in $\mathcal{N}$. We introduce the Weyl quantization by starting with the definition of the Wigner function. 

\begin{definition}[Wigner function] \label{Def:Wigner}
Given two sections $\Psi \in \Gamma(F)$, $\Phi \in \Gamma(E)$, we define the $h$-dependent {\normalfont Wigner function} $W_h[\Psi,\Phi] \in \Gamma(T^*M, \pi^*(F^* \otimes E))$ for all $(z,p)\in T^{\ast}M$ as
\begin{align} \label{eq:W_def}
    W_h[\Psi, \Phi](z, p) &=
    \int_{T_{z} M } \J{z}{z-\frac{u}{2}} \Psi^* \left( z - \frac{u}{2} \right) \otimes  \J{z}{z+\frac{u}{2}} \Phi \left(z+\frac{u}{2}\right) e^{-\frac{i p \cdot u}{h}} \nonumber\\
    &\qquad\qquad\qquad \chi\left(z-\frac{u}{2}, z+\frac{u}{2}\right)\Delta^{-\gamma}\left(z-\frac{u}{2}, z+\frac{u}{2}\right) \dfrac{ d\mu_{T_zM} (u) }{(2\pi h)^d} .
\end{align}
Here, $\gamma \in \mathbb{R}$ is a parameter that determines the power of the Van Vleck-Morette determinant, and $\Psi^{\ast}:=\langle \Psi, \cdot \rangle_{F}\in\Gamma(F^{\ast})$ denotes the dual section.
\end{definition}

We note that $W_h[\Psi,\Phi]$ is $\mathbb{C}$-linear in $\Phi$, while it is $\mathbb{C}$-antilinear in $\Psi$, due to the antilinear identification $\Psi \mapsto \Psi^{\ast} := \langle \Psi, \cdot \rangle_{F}$. In the scalar case, a similar definition without the Van Vleck-Morette determinant ($\gamma=0$) has appeared in \cite{PflaumWeyl, LiuQian}. The role of the parameter $\gamma$ has been discussed in \cite{fulling}, and some particular choices can be argued to be preferred. In particular, $\gamma = \frac{1}{2}$ is used in \cite{dls_weyl}. We will keep $\gamma$ arbitrary and discuss particular choices later. 

\begin{remark}
    Note that, compared to \cite{fulling}, we use a different convention for the $\gamma$ factor. Our $\gamma$ factor is related to the one used in \cite{fulling}, which we denote by $\gamma_F$, by the relation $\gamma + \gamma_F = 1$. In other words, based on the notation of \cite{fulling}, we work with $\gamma = \gamma_{\textrm{Wigner}}$, while Fulling uses $\gamma_F = \gamma_{\textrm{Weyl}}$. 
\end{remark}

\begin{remark}
    The definition of $W_h[\Psi,\Phi]$ should be understood in a \textit{weak sense}. One can make sense of the right-hand side of \cref{eq:W_def} by choosing local trivializations of the bundles $E$ and $F$ to obtain well-defined vector-valued integrals. These objects depend on the choice of trivializations. However, expressions such as
    \begin{align}
        \int_{T^*M} a_h \cdot W_h [\Psi,\Phi]\, d\mu_{T^*M},
    \end{align}
    for $a_h \in \Gamma(T^*M, \pi^*(F \otimes E^*))$, where $a_h \cdot W_h[\Psi,\Phi]$ denotes the duality pairing between the bundles $\pi^*(F \otimes E^*)$ and $\pi^*(F^* \otimes E)$, are clearly independent of such choices, hence warranting the Wigner function as a well-defined object.
\end{remark}

\begin{definition}[Weyl quantization]\label{Def:WeylQuant}
Let $a=(a_h)_{h \in (0, h_0]}\in h^{-l}\mathcal S_h^k(M;E,F)$. The Weyl quantization of $a$ is the family of linear operators $\hat A=(\hat A_h)_{h \in (0, h_0]}$, with $\hat A_h\colon\Gamma_{\mathrm c}(E)\to\Gamma(F)$, defined weakly by
\begin{equation}
    (\Psi,\hat A_h\Phi)_F = \int_M \langle \Psi,\hat A_h\Phi\rangle_F\,d\mu_M = \int_{T^*M} a_h\cdot W_h[\Psi,\Phi]\,d\mu_{T^*M},
\end{equation}
for all $h \in (0, h_0]$ and for all test sections $\Psi\in\Gamma_{\mathrm c}(F)$ and $\Phi\in\Gamma_{\mathrm c}(E)$.  Conversely, let $\hat{A}=(\hat A_h)_{h \in (0, h_0]}$ be a family of differential operators $\hat A_h\colon\Gamma_{\mathrm c}(E)\to\Gamma(F)$. A symbol $a=(a_h)_{h \in (0, h_0]}\in h^{-l}\mathcal S_h^k(M;E,F)$ is called a Weyl symbol of $\hat{A}$ if the above identity holds for all $h \in (0, h_0]$ and for all test sections $\Psi\in\Gamma_{\mathrm c}(F)$ and $\Phi\in\Gamma_{\mathrm c}(E)$.
\end{definition}

If we insert the definition of the Wigner function into the above equation and perform the coordinate transformation $(z,u)\mapsto (x,y):=\left( z - \frac{u}{2}, z + \frac{u}{2} \right)$, we obtain
\begin{align}
    (\Psi, \hat{A}_h \Phi )_{F} &= \int_{T^*M} a_h(z,p) \cdot W_h [\Psi,\Phi](z,p)\, d\mu_{T^*M}(z,p) \nonumber \\
    &= \int_{M \times M} \langle\Psi(x), A_h(x, y) \Phi(y)\rangle_{F}\, d\mu_{M \times M}(x,y)\label{eq:Schwartz},
\end{align}
where in the second line we used the Schwartz kernel theorem for bundle-valued operators on manifolds \cite[Thm.~1.5.1]{Tarkhanov}, which asserts that any linear operator $\hat{A}_h\colon\Gamma_{\mathrm{c}}(E)\to\Gamma(F)$ is continuous with respect to the relevant LF-space and Fréchet topologies if and only if it is induced by a \textit{(Schwartz) kernel} $A_h(\cdot,\cdot)\in\mathcal{D}^{\prime}(M\times M,F\boxtimes E^{\ast})$, i.e.~$A_h(x,y)\colon E_{y}\to F_{x}$ for all $(x,y)\in M\times M$, as in~\cref{eq:Schwartz}, where
\begin{align}
    F\boxtimes E^{\ast}:=(\mathrm{pr}_{1}^{\ast}F)\otimes_{M\times M}(\mathrm{pr}_{2}^{\ast}E^{\ast})
\end{align}
denotes the outer tensor product of vector bundles ($\mathrm{pr}_{1,2}$ are the projectors to the factors of $M \times M$). This equation allows us to express the kernel in terms of the symbol, and we find 
\begin{equation} \label{eq:symbol2kernel}
    A_h(x,y) = \Delta^{1 - \gamma}(x,y)\chi(x,y) \int_{T^*_z M} \J{x}{z} a_h(z,p) \J{z}{y} e^{-\frac{i p \cdot u}{h}} \frac{d \mu_{T^*_z M}(p)}{(2 \pi h)^d} ,
\end{equation}
which has to be understood in the sense of oscillatory integrals. In this equation, we have $z = x + \frac{1}{2}(y-x)$ and $u = \J{z}{x}(y-x)$. By taking the inverse Fourier transform, we obtain a Weyl symbol
\begin{equation} \label{eq:kernel2symbol}
    a_h(z,p) = \int_{T_z M} \J{z}{x} A_h(x, y) \J{y}{z} e^{\frac{i p \cdot u}{h}} \Delta^{\gamma - 1}(x,y)\, d\mu_{T_z M}(u) \qquad \mathrm{mod} \quad h^\infty \mathcal{S}_h^{-\infty}(M; E, F) ,
\end{equation}
where $x = z - \frac{u}{2}$ and $y = z + \frac{u}{2}$.

At this point, it is important to emphasize that the Weyl quantization introduced in \cref{Def:WeylQuant} depends on the choice of cut-off function $\chi$ near the diagonal. However, different choices of cut-off lead to quantizations that differ by an element of $h^\infty\Psi_h^{-\infty}(M;E,F)$.

\begin{proposition} \label{Prop:3.1}
    For any semiclassical symbol $a \in h^{-l} \mathcal{S}_h^{k}(M;E,F)$, its Weyl quantization $\hat A=(\hat A_h)_{h \in (0, h_0]}$, with
$\hat A_h\colon\Gamma_{\mathrm c}(E)\to\Gamma(F)$ given in \cref{Def:WeylQuant}, has the following two properties:
    \begin{itemize}
        \item[(i)] $\hat A\in h^{-l}\Psi_h^k(M;E,F)$.
        \item[(ii)] $\hat A$ is independent of the cut-off $\chi$ up to an element of $h^\infty\Psi_h^{-\infty}(M;E,F)$.
    \end{itemize}
    In other words, the Weyl quantization defines a linear map
    \begin{align}
        h^{-l} \mathcal{S}_h^{k}(M;E,F) \to h^{-l} \Psi_h^{k}(M;E,F)/ h^{\infty} \Psi_h^{-\infty}(M;E,F),\qquad a\mapsto \hat{A}.
    \end{align}
\end{proposition}

\begin{proof}
First, note that it is immediate that the singular support of the Schwartz kernel
\begin{equation}
    A_h(\cdot,\cdot)\in\mathcal D'(M\times M,F\boxtimes E^*)
\end{equation}
of $\hat A_h$, given by \cref{eq:symbol2kernel}, is contained in the diagonal $\delta:=\{(x,x)\in M\times M\mid x\in M\}$. Indeed, the phase function $\varphi(p,u):=p\cdot u$, with $u=\J{z}{x}(y-x)$, is stationary in $p$ if and only if $u=0$, and this is equivalent to $x=y$. Moreover, in local coordinates and trivializations, the kernel in \cref{eq:symbol2kernel} has the standard local form of a semiclassical pseudodifferential kernel with symbol in $h^{-l}\mathcal S_h^k$. Hence, $\hat A\in h^{-l}\Psi_h^k(M;E,F)$.

Now, let $\chi'\in C^\infty(M\times M)$ be another admissible cut-off function. The difference of the two kernels is obtained from \cref{eq:symbol2kernel} by replacing $\chi$ with $\chi-\chi'$. Since $\chi-\chi'$ vanishes in a neighborhood of the diagonal, the phase $\varphi(p,u)=p\cdot u$ has no stationary point on $\mathrm{supp}(\chi-\chi^{\prime})$. Equivalently, on compact subsets of $M\times M$ intersected with this support, one has $|u|\ge c>0$ for some constant $c>0$. Therefore, considering the regularization operator
\begin{equation}
    L:=\frac{i h}{|u|^2}u^\alpha\partial_{p_\alpha},
\end{equation}
defined on sections supported on $\mathrm{supp}(\chi-\chi^{\prime})$, which satisfies
\begin{equation}
    L e^{-\frac{i}{h}p\cdot u}=e^{-\frac{i}{h}p\cdot u}\qquad\text{on}\qquad \mathrm{supp}(\chi-\chi^{\prime}),
\end{equation}
repeated integration by parts in $p$ shows that the difference kernel is smooth and that all of its $C^\infty$-seminorms on compact subsets are $\mathcal O(h^N)$ for every $N\in\mathbb N$. Thus, the two choices of cut-off functions yield operators that differ by an element of $h^\infty\Psi_h^{-\infty}(M;E,F)$.
\end{proof}

\subsection{Arbitrary $\tau$-quantizations}

The Weyl quantization defined above is special in the sense that it involves the midpoint $z$ between $x = z - \frac{u}{2}$ and $y = z + \frac{u}{2}$. Although working with the midpoint has certain advantages (see, for instance, Section~\ref{Subsec:Hermitian} below), it is also possible to define other quantizations based at different points on the geodesic segment connecting $x$ and $y$. We refer to these as $\tau$-quantizations, which are based at $z_\tau = x + \tau(y-x)$, with $\tau \in [0, 1]$. In particular, $\tau = \frac{1}{2}$ gives the Weyl quantization, and $\tau = 0$ is known as the Kohn-Nirenberg quantization \cite{KohnNirenberg,Hormander}.

In uniform pseudodifferential calculus on $\mathbb{R}^{n}$, it is well known that one can transition between these different quantization schemes. That is, any pseudodifferential operator on $\mathbb{R}^{n}$ can be realized as the $\tau$-quantization of a symbol for any $\tau \in [0,1]$, and the corresponding symbols for different values of $\tau$ are related by an asymptotic expansion of one another; see, for example, the classical treatment in \cite[Thm.~23.2 and 23.3]{Shubin}. In this section, we establish the analogous property for the covariant quantization framework on manifolds developed in this article, generalizing the results obtained in \cite{Safarov,dls_weyl}.

To define the $\tau$-quantization, we introduce the $h$-dependent \textit{$\tau$-Wigner function}
\begin{equation}
         W_{h,\tau}[\Psi, \Phi](z, p) := \int_{T_{z} M } \J{z}{x_\tau} \Psi^* \left( x_\tau \right) \otimes  \J{z}{y_\tau} \Phi \left(y_\tau\right) e^{-\frac{i p \cdot u}{h}} \chi\left(x_\tau, y_\tau \right)\Delta^{-\gamma}\left(x_\tau, y_\tau \right) \dfrac{ d\mu_{T_z M} (u) }{(2\pi h)^d} ,
\end{equation}
where $x_\tau := z - \tau u$ and $y_\tau := z + (1 - \tau) u$; see also \cite{PflaumWeyl, LiuQian}, which contains a similar definition in the scalar case with $\gamma=0$.

The $\tau$-symbol $a_\tau=(a_{h,\tau})_{h \in (0, h_0]}$ of a family of operators
$\hat A=(\hat A_h)_{h \in (0, h_0]}$ is defined by
\begin{equation}
    (\Psi, \hat{A}_h \Phi)_{F} = \int_{T^*M} a_{h,\tau} \cdot W_{h,\tau} [\Psi,\Phi]\, d\mu_{T^*M} \qquad \forall h \in (0,h_0].
\end{equation}
Using this definition, we can express the kernel in terms of the $\tau$-symbol as
\begin{equation}
\label{kernel:tau}
    A_h (x,y) = \Delta^{1-\gamma}(x,y)\chi(x,y) \int_{T^*_{z_\tau} M} \J{x}{z_\tau} a_{h,\tau} (z_\tau, p) \J{z_\tau}{y} e^{-\frac{i p \cdot u_\tau}{h}} \frac{d \mu_{T^*_{z_\tau} M}(p)}{(2 \pi h)^d},
\end{equation}
where $u_\tau = \J{z_\tau}{x}(y-x)$. By taking the inverse Fourier transform, we obtain a $\tau$-symbol
\begin{equation}
\label{symbol:tau}
    a_{h,\tau} (z,p) = \int_{T_z M} \J{z}{x_\tau} A_h(x_\tau, y_\tau) \J{y_\tau}{z} e^{\frac{i p \cdot u}{h}} \Delta^{\gamma - 1}(x_\tau, y_\tau)\, d\mu_{T_z M}(u)\qquad \mathrm{mod} \quad h^\infty \mathcal{S}_h^{-\infty} .
\end{equation}

In general, given an operator $\hat{A}$ with kernel $A(x,y)$, its corresponding symbols $a_\tau$ and $a_{\tau'}$ for different $\tau$-quantizations can be asymptotically related as in \cite{Safarov,dls_weyl,Shubin}.

\begin{proposition}\label{Prp:tauchange}
Let $\tau,\sigma\in[0,1]$, and let $a_\tau=(a_{h,\tau})_{0<h\le h_0}$ and $a_\sigma=(a_{h,\sigma})_{0<h\le h_0}$ be elements of $h^{-l}\mathcal{S}_h^k(M;E,F)$ which are, respectively, $\tau$- and $\sigma$-symbols whose corresponding quantized operators coincide modulo $h^\infty\Psi_h^{-\infty}(M;E,F)$. Let $\gamma,\gamma'\in\mathbb R$ be the powers of the van Vleck-Morette determinants used in the $\tau$- and $\sigma$-quantizations, respectively. Then, $a_\tau$ and $a_\sigma$ are related by the asymptotic expansion
\begin{equation}
    a_{h,\tau}(z,p) \sim \sum_{\boldsymbol\beta} f_{\boldsymbol\beta}(z) \left(-i h\partial_p\right)^{\boldsymbol\beta} e^{-ih(\sigma-\tau)\partial_p\cdot\hnabla} a_{h,\sigma}(z,p),
\end{equation}
in $h^{-l}\mathcal{S}_h^k(M;E,F)$, modulo $h^\infty\mathcal{S}_h^{-\infty}(M;E,F)$. Here, $\partial_p$ denotes the vertical derivative, so that $\partial_p\cdot\hnabla=\vnabla^\mu\hnabla_\mu$, and the coefficients $f_{\boldsymbol\beta}(z)$ are defined by the Taylor expansion around $z$ of the Van Vleck-Morette term
\begin{equation}
    \Delta^{\gamma-\gamma'}(x_\tau,y_\tau) \sim \sum_{\boldsymbol\beta} u^{\boldsymbol\beta}f_{\boldsymbol\beta}(z).
\end{equation}
The coefficients $f_{\boldsymbol\beta}$ depend on $\tau,\gamma,\gamma'$.
\end{proposition}

\begin{proof}
Using \cref{symbol:tau} to express the $\tau$-symbol in terms of the kernel and \cref{kernel:tau} to express the kernel in terms of the $\sigma$-symbol, we obtain
\begin{align}
       &a_{h,\tau} (z,p) = \int_{T_z M} \J{z}{x_\tau} A_h(x_\tau, y_\tau) \J{y_\tau}{z} e^{\frac{i p \cdot u}{h}} \Delta^{\gamma - 1}(x_\tau, y_\tau)\, d\mu_{T_z M}(u) \nonumber \\
       &= \int_{T_z M \times T^*_{z_{\sigma}}M} \J{z}{z_\sigma} a_{h,\sigma} (z_\sigma, p') \J{z_\sigma}{z} \, e^{\frac{i (p \cdot u - p' \cdot u_\sigma)}{h}} \chi(x_\tau, y_\tau) \Delta^{\gamma-\gamma'}(x_\tau, y_\tau)\frac{{d\mu}_{T_z M \times T^*_{z_{\sigma}}M }(u, p')}{(2\pi h)^d} \nonumber \\
       &=\int_{T_z M \times T^*_z M} \J{z}{z_\sigma} a_{h,\sigma} (z_\sigma, \J{z_\sigma}{z} q ) \J{z_\sigma}{z} \, e^{\frac{i (p - q) \cdot u}{h}} \chi(x_\tau, y_\tau) \Delta^{\gamma-\gamma'}(x_\tau, y_\tau) \frac{{d\mu}_{T_z M \times T^*_z M }(u, q)}{(2\pi h)^d},
\end{align}
where the above equations hold modulo $h^\infty \mathcal{S}_h^{-\infty}(M;E,F)$, and the third equality follows after the change of coordinates $p' = \J{z_\sigma}{z} q$.
Next, we Taylor-expand around $z$ the individual terms in the above integral. We have
\begin{subequations}
\begin{align}
    \J{z}{z_\sigma} a_{h,\sigma}(z_\sigma, \J{z_\sigma}{z} q ) \J{z_\sigma}{z} &\sim \sum_{\boldsymbol{\alpha}} \frac{(\sigma - \tau)^{|\boldsymbol{\alpha}|}}{\boldsymbol{\alpha}!} \, u^{\boldsymbol{\alpha}} \hnabla_{\boldsymbol{\alpha}} \, a_{h,\sigma}(z,q), \\
    \Delta^{\gamma-\gamma'}(x_\tau, y_\tau) &\sim \sum_{\boldsymbol{\beta}} u^{\boldsymbol{\beta}} f_{\boldsymbol{\beta}} (z), \\
    \chi(x_\tau, y_\tau) &\sim 1,
\end{align}    
\end{subequations}
where $\boldsymbol{\alpha}, \boldsymbol{\beta}$ are multi-indices, we used \cref{lemma:derivative} for the expansion of the symbol, where $\hnabla_{\boldsymbol{\alpha}}$ denotes the corresponding symmetrized horizontal derivative. The expansion of the Van Vleck-Morette determinant can be calculated to arbitrary orders as in \cite{visser,Poisson}, and the expansion of the cut-off function around $z$ is trivial by definition. Using these expressions, we obtain
\begin{align}
    &a_{h,\tau}(z,p) \sim \int_{T_z M \times T^*_z M} \left[ \sum_{\boldsymbol{\alpha}} \frac{(\sigma - \tau)^{|\boldsymbol{\alpha}|}}{\boldsymbol{\alpha}!} \, u^{\boldsymbol{\alpha}} \hnabla_{\boldsymbol{\alpha}} \, a_{h,\sigma}(z,q) \right] e^{\frac{i (p - q) \cdot u}{h}} \sum_{\boldsymbol{\beta}} u^{\boldsymbol{\beta}} f_{\boldsymbol{\beta}}(z) \frac{{d\mu}_{T_z M \times T^*_z M }(u, q)}{(2\pi h)^d} \nonumber \\  
    &=  \int_{T_z M \times T^*_z M} \left\{ \sum_{\boldsymbol{\alpha}, \boldsymbol{\beta}} \frac{(\sigma - \tau)^{|\boldsymbol{\alpha}|}}{\boldsymbol{\alpha}!} \, \left[ f_{\boldsymbol{\beta}}(z) (i h \partial_q)^{\boldsymbol{\beta}} (i h \partial_q)^{\boldsymbol{\alpha}}  e^{\frac{i(p - q) \cdot u}{h} }\right] \hnabla_{\boldsymbol{\alpha}} \, a_{h,\sigma}(z,q) \right\} \frac{{d\mu}_{T_z M \times T^*_z M }(u, q)}{(2\pi h)^d} \nonumber \\
    &= \int_{T_z M \times T^*_z M} e^{\frac{i(p - q) \cdot u}{h} } \left\{ \sum_{\boldsymbol{\alpha}, \boldsymbol{\beta}} \frac{(\sigma - \tau)^{|\boldsymbol{\alpha}|}}{\boldsymbol{\alpha}!} f_{\boldsymbol{\beta}}(z) \left( - i h \partial_q\right)^{\boldsymbol{\beta}} \left( - i h \partial_q\right)^{\boldsymbol{\alpha}} \hnabla_{\boldsymbol{\alpha}} a_{h,\sigma}(z,q) \right\} \frac{{d\mu}_{T_z M \times T^*_z M }(u, q)}{(2\pi h)^d} \nonumber \\
    &= \int_{T^*_z M} \delta(p - q) \left\{ \sum_{\boldsymbol{\alpha}, \boldsymbol{\beta}} \frac{(\sigma - \tau)^{|\boldsymbol{\alpha}|}}{\boldsymbol{\alpha}!} f_{\boldsymbol{\beta}}(z) \left( - i h \partial_q\right)^{\boldsymbol{\beta}} \left( - i h \partial_q\right)^{\boldsymbol{\alpha}} \hnabla_{\boldsymbol{\alpha}} a_{h,\sigma}(z,q) \right\} d\mu_{T^*_z M }(q) \nonumber \\
    &= \sum_{\boldsymbol{\alpha}, \boldsymbol{\beta}} \frac{(\sigma - \tau)^{|\boldsymbol{\alpha}|}}{\boldsymbol{\alpha}!} f_{\boldsymbol{\beta}}(z) \left( - i h \partial_p\right)^{\boldsymbol{\beta}} \left( - i h \partial_p\right)^{\boldsymbol{\alpha}} \hnabla_{\boldsymbol{\alpha}} a_{h,\sigma}(z,p) \nonumber \\
    &= \sum_{\boldsymbol{\beta}} f_{\boldsymbol{\beta}}(z)  \left( - i h \partial_p\right)^{\boldsymbol{\beta}} e^{- i h (\sigma - \tau)\partial_p \cdot \hnabla} \, a_{h,\sigma}(z,p).
\end{align}
All equalities above are understood modulo $h^\infty\mathcal{S}_h^{-\infty}(M;E,F)$. In the first line of the above equation, we inserted the Taylor expansions of the individual terms. The second line follows by replacing the $u$ terms with $\partial_q$ derivatives acting on the exponential, and the third line is obtained after integration by parts. Then, we integrate over $u$ to obtain the fourth line, and we integrate over $q$ to obtain the fifth line. The final equality follows by summing over the multi-index $\boldsymbol{\alpha}$ and using the multinomial expansion of the exponential. Note that, using the commutation properties of vertical and horizontal derivatives, we can swap the order of $\left( - i h \partial_p\right)^{\boldsymbol{\beta}}$ and $e^{- i h (\sigma - \tau)\partial_p \cdot \hnabla}$.
\end{proof}

\subsection{Star product}

Consider two linear differential operators $\hat{A}\colon \Gamma(F)\to \Gamma(G)$ and $\hat{B}\colon \Gamma(E)\to \Gamma(F)$ of orders $k_1$ and $k_2$, respectively. As discussed in \cref{example:semi_diff_op}, we may regard them as semiclassical differential operators $\hat A\in h^{-k_1}\mathrm{Diff}_h^{k_1}(M;F,G)$ and $\hat B\in h^{-k_2}\mathrm{Diff}_h^{k_2}(M;E,F)$. Let $a=(a_h)_{h \in (0, h_0]}\in h^{-k_1}\mathcal{S}_h^{k_1}(M;F,G)$ and $b=(b_h)_{h \in (0, h_0]}\in h^{-k_2}\mathcal{S}_h^{k_2}(M;E,F)$ denote their corresponding Weyl symbols. Thus, for each fixed $h$, the symbols are sections
\begin{equation}
    a_h\in\Gamma(T^*M,\pi^*(F^*\otimes G)), \qquad
    b_h\in\Gamma(T^*M,\pi^*(E^*\otimes F)).
\end{equation}
Then the composition $\hat A\circ\hat B$ is again a linear differential operator of order at most $k_1+k_2$, and hence has a Weyl symbol $c=(c_h)_{h \in (0, h_0]}\in h^{-(k_1+k_2)}\mathcal S_h^{k_1+k_2}(M;E,G)$, modulo $h^\infty \mathcal{S}_h^{-\infty}(M;E,G)$. We define
\begin{equation}
    c:=a\star b,
\end{equation}
where, for every fixed $h$, one has $c_h:=a_h\star b_h \in\Gamma(T^*M,\pi^*(E^*\otimes G))$. The so-defined associative binary operation $\star$ is called the \emph{star product}. The goal of this section is to provide an explicit description of this product.

To define the star product in the pseudodifferential setting, we begin by recalling the following terminology. Let
$\hat D=(\hat D_h)_{h \in (0, h_0]}$ be a family of linear and continuous operators
$\hat D_h\colon\Gamma_{\mathrm c}(E)\to\Gamma(F)$, with corresponding Schwartz kernels $D_h(x,y)\in\mathcal{D}^{\prime}(M\times M,F\boxtimes E^*)$. The operator $\hat D$ is said to be \emph{properly supported} if there exists a closed set
$K_D\subset M\times M$ such that $\mathrm{supp}(D_h)\subset K_D$ for all $h\in(0,h_0]$ and such that the two canonical projections
\begin{equation}
    \mathrm{pr}_{1,2}\colon K_D\to M
\end{equation}
are proper maps. That is, the preimage of any compact subset of $M$ under either projection is compact. In particular, each properly supported operator $\hat D_h\colon\Gamma_{\mathrm c}(E)\to\Gamma(F)$ admits a unique continuous extension $\hat D_h \colon \Gamma(E) \to \Gamma(F)$.

The following lemma is the semiclassical version of a well-known result from pseudodifferential calculus; see \cite[Rmk.~3.3]{Sjorstrand} and \cite[Prop.~3.3]{Shubin}.

\begin{lemma} \label{Lemma:PropSup}
Let $d=(d_h)_{h \in (0, h_0]}\in h^{-l}\mathcal S_h^k(M;E,F)$ be a semiclassical symbol, and let $\hat D=(\hat D_h)_{h \in (0, h_0]}$ be its Weyl quantization. Then, there exists a properly supported family of operators $\hat D^{\mathrm{prop}}\in h^{-l}\Psi_h^k(M;E,F)$ such that $\hat D-\hat D^{\mathrm{prop}} \in h^\infty\Psi_h^{-\infty}(M;E,F)$. In this sense, $\hat D$ is properly supported modulo $h^\infty\Psi_h^{-\infty}(M;E,F)$.
\end{lemma}

\begin{proof}
Let $\eta\in C^\infty(M\times M)$ be such that $\eta(x,y)=1$ in a neighborhood of the diagonal $\delta=\{(x,x)\mid x\in M\}$ and such that the projections $\mathrm{pr}_{1,2}\colon\mathrm{supp}(\eta)\to M$ are proper maps (for the existence of such a cut-off function, see \cite[p.~29]{Sjorstrand}). We decompose the Schwartz kernel of $\hat D_h$ as
\begin{align} 
        D_h(x,y) =& \Delta^{1 - \gamma}(x,y)\chi(x,y)\eta(x,y) \int_{T^*_z M} \J{x}{z} d_h(z,p) \J{z}{y} e^{-\frac{i p \cdot u}{h}} \frac{d \mu_{T^*_z M}(p)}{(2 \pi h)^d}\nonumber\\ 
        &+\Delta^{1 - \gamma}(x,y)\chi(x,y)[1-\eta(x,y)] \int_{T^*_z M} \J{x}{z} d_h(z,p) \J{z}{y} e^{-\frac{i p \cdot u}{h}} \frac{d \mu_{T^*_z M}(p)}{(2 \pi h)^d} ,
    \end{align}
where $z=x+\frac{1}{2}(y-x)$ and $u=\J{z}{x}(y-x)$. Note that the product $\chi(x,y)\eta(x,y)$ still has the property that the projections $\mathrm{pr}_{1,2}\colon\mathrm{supp}(\chi\eta)\to M$ are proper maps. Therefore, the first term in the above equation defines a properly supported semiclassical pseudodifferential operator in $h^{-l}\Psi_h^k(M;E,F)$.

It remains to consider the second term. Since $1-\eta$ vanishes in a neighborhood of the diagonal, the phase $\varphi(p,u)=p\cdot u$ has no stationary point in $p$ on the support of this term. Hence, by the same integration-by-parts argument as in the proof of \cref{Prop:3.1}, this second summand is a smooth kernel whose $C^\infty$-seminorms on compact subsets are $\mathcal O(h^N)$ for every $N\in\mathbb N$. Therefore, it defines an element of $h^\infty\Psi_h^{-\infty}(M;E,F)$.

Thus, the first term gives the desired properly supported representative $\hat D^{\mathrm{prop}}$, and the difference $\hat D-\hat D^{\mathrm{prop}}$ belongs to $h^\infty\Psi_h^{-\infty}(M;E,F)$.
\end{proof}

After this general discussion, consider two families of linear operators $\hat{A}=(\hat A_h)_{h \in (0, h_0]}$ and $\hat{B}=(\hat B_h)_{h \in (0, h_0]}$, with $\hat A_h\colon \Gamma_{\mathrm c}(F)\to \Gamma(G)$ and $\hat B_h\colon \Gamma_{\mathrm c}(E)\to \Gamma(F)$, obtained as the Weyl quantizations of arbitrary semiclassical symbols $a=(a_h)_{h \in (0, h_0]}\in h^{-l_1}\mathcal{S}_h^{k_1}(M;F,G)$ and $b=(b_h)_{h \in (0, h_0]}\in h^{-l_2}\mathcal{S}_h^{k_2}(M;E,F)$. For each fixed $h$, these symbols are sections $a_h\in \Gamma(T^*M,\pi^*(F^*\otimes G))$ and $b_h\in \Gamma(T^*M,\pi^*(E^*\otimes F))$. Following \cref{Lemma:PropSup}, the operators $\hat A$ and $\hat B$ can be replaced, modulo $h^\infty\Psi_h^{-\infty}$, by properly supported representatives. Hence, their composition $\hat A\circ\hat B$ is well defined modulo $h^\infty\Psi_h^{-\infty}(M;E,G)$. The resulting operator belongs to $h^{-(l_1+l_2)}\Psi_h^{k_1+k_2}(M;E,G)$ and again admits a Weyl symbol $c:=a\star b$. In particular, the star product gives rise to a well-defined map
\begin{align}\label{eq:starQuo}
    \star\colon \frac{h^{-l_1}\mathcal S_h^{k_1}(M;F,G)}{h^\infty\mathcal S_h^{-\infty}(M;F,G)}
    \times
    \frac{h^{-l_2}\mathcal S_h^{k_2}(M;E,F)} {h^\infty\mathcal S_h^{-\infty}(M;E,F)}
    \to
    \frac{h^{-(l_1+l_2)}\mathcal S_h^{k_1+k_2}(M;E,G)}{h^\infty\mathcal S_h^{-\infty}(M;E,G)} .
\end{align}

The following theorem is the vector bundle generalization of the corresponding result in \cite[Thm.~3.9]{dls_weyl}. The proof follows the same general steps, but additional terms arise from the parallel transport operators. These terms will contribute to the asymptotic expansion of the star product.

\begin{theorem}[Star product] \label{th:star_product}
Let $\hat A=(\hat A_h)_{h \in (0, h_0]}$ and $\hat B=(\hat B_h)_{h \in (0, h_0]}$ be the Weyl quantizations of the symbols $a=(a_h)_{h \in (0, h_0]} \in h^{-l_1}\mathcal{S}_h^{k_1}(M;F,G)$ and $b=(b_h)_{h \in (0, h_0]} \in h^{-l_2}\mathcal{S}_h^{k_2}(M;E,F)$, respectively. Thus, $\hat A_h\colon\Gamma_{\mathrm c}(F)\to\Gamma(G)$ and $\hat B_h\colon\Gamma_{\mathrm c}(E)\to\Gamma(F)$ for every $h$. After replacing $\hat A$ and $\hat B$ with properly supported representatives modulo $h^\infty\Psi_h^{-\infty}$, their composition defines a family $\hat C=(\hat C_h)_{h\in(0,h_0]}$ with $\hat C_h:=\hat A_h\circ\hat B_h$.

Then, the Weyl symbol $c=(c_h)_{h \in (0, h_0]} \in h^{-(l_1+l_2)}\mathcal S_h^{k_1+k_2}(M;E,G)$ of the operator family $\hat C = (\hat{C}_h)_{h \in (0, h_0]}$ is given, modulo $h^\infty\mathcal S_h^{-\infty}(M;E,G)$, by the star product $c=a\star b$, with $c_h = a_h \star b_h$. For each fixed $h\in(0,h_0]$, the symbol $c_h=a_h\star b_h$ is given by the following oscillatory integral and admits the corresponding semiclassical asymptotic expansion:  

\begin{subequations}
\begin{align} 
    (a_h &\star b_h)(z,p) = (\pi h)^{-2 d} \int_{N_{z}} d\mu_{N_z}(u_1,u_2,p_1,p_2)\, \Lambda(z, u_1, u_2)e^{\frac{2 i p \cdot (w + u_1 - u_2)}{h}} e^{\frac{2 i (p_2 \cdot u_1 - p_1 \cdot u_2) }{h}} \nonumber \\
    &\quad\times \Big[\J{z}{z - w} \cdot \J{z - w}{z+v_1} a_h(z+v_1, \J{z+v_1}{z} (p+p_1)) \J{z+v_1}{z+\tilde{w}} \Big] \nonumber \\
    &\quad\times\Big[ \J{z+\tilde{w}}{z+v_2} b_h(z+v_2, \J{z+v_2}{z} (p+p_2)) \J{z+v_2}{z+w} \cdot \J{z+w}{z} \Big] \quad\mathrm{mod}\quad h^{\infty}\mathcal{S}_h^{-\infty}(M;E,G) \label{eq:star1} \\
    \sim& e^{\frac{i h}{2} \left( \partial_{u_1} \cdot \partial_{p_2} - \partial_{u_2} \cdot \partial_{p_1} \right)} \Lambda(z, u_1, u_2) e^{\frac{2 i p \cdot (w + u_1 - u_2)}{h}} \nonumber \\
    &\quad \times {\mathbb{H}_z(\nabla^{\pi^*G})} \Big[\J{z}{z+v_1} a_h(z+v_1, \J{z+v_1}{z} (p+p_1)) \J{z+v_1}{z} \Big] \mathbb{H}_z(\nabla^{\pi^*F}) \nonumber \\
    &\quad\times \Big[ \J{z}{z+v_2} b_h(z+v_2, \J{z+v_2}{z} (p+p_2)) \J{z+v_2}{z} \Big] {\mathbb{H}_z(\nabla^{\pi^*E})} \Bigg|_{\substack{u_1 = u_2 = 0\\p_1=p_2=0}}, \label{eq:star2}
\end{align}
\end{subequations}
where $N_{z}:= T_zM \times T_zM \times T^*_zM \times T^*_zM$, the geometric factor $\Lambda(z,u_{1},u_{2})$ is defined by
\begin{align}\label{eq:GeomFact}
    \Lambda(z,u_{1},u_{2}):=2^{-d} \bigg\vert\frac{\partial (w,\tilde{w})}{\partial (u_{1},u_{2})}\bigg\vert \frac{ \Delta(z-w,z+\tilde{w})^{1-\gamma} \Delta(z+w,z+\tilde{w})^{1-\gamma} }{ \Delta(z-w,z+w)^{1-\gamma}\Delta(z,z+\tilde{w})},
\end{align}
and where $v_{1},v_{2},w,\tilde{w}\in T_{z}M$ are tangent vectors on $z$ depending on $u_{1},u_{2}\in T_{z}M$ and are completely determined by the set of equations
\begin{subequations}\label{eq:Vectors}
\begin{align}
    z-w=(z+v_{1})-\J{z+v_{1}}{z}(u_{2}),\\
    z+w=(z+v_{2})-\J{z+v_{2}}{z}(u_{1}),\\
    z+\tilde{w}=(z+v_{1})+\J{z+v_{1}}{z}(u_{2}),\\
    z+\tilde{w}=(z+v_{2})+\J{z+v_{2}}{z}(u_{1}).
\end{align}
\end{subequations}
The term $\mathbb{H}_z(\nabla^{\pi^*G})$ is the holonomy of the connection $\nabla^{\pi^*G}$ at $z$ along the geodesic triangle $z \mapsto z+v_1 \mapsto z-w \mapsto z$ (around the green triangle in \cref{Fig:Triangle}), $\mathbb{H}_z(\nabla^{\pi^*F})$ is the holonomy of the connection $\nabla^{\pi^*F}$ at $z$ along the geodesic loop $z \mapsto z+v_2 \mapsto z+\tilde{w} \mapsto z+v_1 \mapsto z$ (around the blue quadrilateral in \cref{Fig:Triangle}), and $\mathbb{H}_z(\nabla^{\pi^*E})$ is the holonomy of the connection $\nabla^{\pi^*E}$ at $z$ along the geodesic triangle $z \mapsto z+w \mapsto z+v_2 \mapsto z$ (around the orange triangle in \cref{Fig:Triangle}).
\end{theorem}

\begin{proof}
The proof follows the same steps as in \cite[Thm.~3.9]{dls_weyl}, but additional terms related to the holonomies of the nontrivial bundle connections arise. The geometric definition of the vectors $v_{1},v_{2},w,\tilde{w}\in T_{z}M$ and their relation to $u_{1}, u_{2}$ is the same as in \cite[Thm.~3.9]{dls_weyl}. We illustrate this in Figure~\ref{Fig:Triangle}. All equalities between operator families below are understood modulo $h^\infty\Psi_h^{-\infty}$, and all equalities between symbol families are understood modulo $h^\infty\mathcal S_h^{-\infty}$.     

\begin{figure}[t!]
\centering
\includegraphics[scale=1.2]{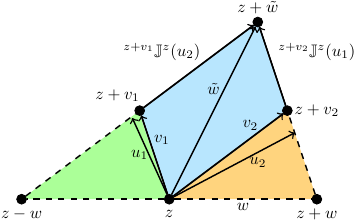}
\caption{The vectors $v_{1},v_{2},w,\tilde{w}\in T_{z}M$ depending on $u_{1},u_{2}\in T_{z}M$ and $z\in M$ as defined by the relations~\eqref{eq:Vectors}. The points $z,z+v_{1},z+v_{2}$ are the midpoints of the edges of the triangle spanned by $z-w,z+w,z+\tilde{w}$. Furthermore, note that $u_{1}=v_{1}$ and $u_{2}=v_{2}$ in the flat case. The colored geodesic triangles correspond to the holonomy terms that arise in the definition of the star product.
\label{Fig:Triangle}} 
\end{figure} 

By \cref{Lemma:PropSup}, we may replace $\hat A$ and $\hat B$ with properly supported representatives. This does not change the resulting Weyl symbol modulo $h^\infty\mathcal S_h^{-\infty}(M;E,G)$. For each fixed $h\in(0,h_0]$, let $A_h(x,y)$, $B_h(x,y)$, and $C_h(x,y)$ denote the Schwartz kernels of $\hat A_h$, $\hat B_h$, and $\hat C_h:=\hat A_h\circ\hat B_h$, respectively. Then
\begin{equation} \label{eq:C_def}
    C_h(x, y) = \int_M A_h(x, \tilde{z}) B_h(\tilde{z}, y) d \mu_M(\tilde{z}).
\end{equation}
We denote by $z$ the midpoint between $x$ and $y$, and we introduce $w, \tilde{w} \in T_z M$ by
\begin{equation}
    x = z - w, \qquad y = z + w, \qquad \tilde{z} = z + \tilde{w}. 
\end{equation}
Then, we can rewrite \cref{eq:C_def} as
\begin{align} \label{eq:C_intermediate}
    C_h(z - w, z + w) &= \int_M A_h(z - w, \tilde{z}) B_h(\tilde{z}, z + w) d \mu_M(\tilde{z}) \nonumber \\
    &= \int_{T_z M} A_h(z - w, z+ \tilde{w}) B_h(z+ \tilde{w}, z + w) \Delta^{-1}(z, z + \tilde{w}) d\mu_{T_z M}(\tilde{w}).
\end{align}
In the second line of the above equation, we performed the change of variables $\tilde{w} = \tilde{z} - z$, for which $d \mu_{T_z M}(\tilde{w}) = \Delta(z, \tilde{z}) d \mu_{M}(\tilde{z})$. From the above equation, the Weyl symbol $c_h = a_h \star b_h$ is obtained by applying \cref{eq:kernel2symbol}, which gives ($\mathrm{mod} \, \, h^{\infty}\mathcal{S}_h^{-\infty}$)
\begin{align} \label{eq:C_intermediate2}
    c_h(z, p) &= \int_{T_z M} \J{z}{z - \frac{u}{2}} C_h \left( z - \frac{u}{2}, z + \frac{u}{2} \right) \J{z + \frac{u}{2}}{z} e^{\frac{i p \cdot u}{h}} \Delta^{\gamma - 1} \left( z - \frac{u}{2}, z + \frac{u}{2} \right) \, d\mu_{T_z M}(u) \nonumber \\
    &= 2^d \int_{T_z M} \J{z}{z - w} C_h ( z - w, z + w ) \J{z + w}{z} e^{\frac{2 i p \cdot w}{h}} \Delta^{\gamma - 1} ( z - w, z + w ) \, d\mu_{T_z M}(w) \nonumber \\
    &= 2^d \int_{T_z M \times T_z M} \J{z}{z - w} A_h(z - w, z+ \tilde{w}) B_h(z+ \tilde{w}, z + w) \J{z + w}{z} e^{\frac{2 i p \cdot w}{h}} \nonumber \\
    & \qquad \times \Delta^{-1}(z, z + \tilde{w}) \Delta^{\gamma - 1} ( z - w, z + w ) \, d\mu_{T_z M}(w) d\mu_{T_z M}(\tilde{w}) \nonumber \\
    &= 2^d \int_{T_z M \times T_z M} \J{z}{z - w} A_h(z - w, z+ \tilde{w}) B_h(z+ \tilde{w}, z + w) \J{z + w}{z} e^{\frac{2 i p \cdot w}{h}} \nonumber \\
    & \qquad \times \Delta^{-1}(z, z + \tilde{w}) \Delta^{\gamma - 1} ( z - w, z + w ) \, \bigg| \frac{\partial (w, \tilde{w})}{\partial (u_1, u_2)} \bigg|  d\mu_{T_z M}(u_1) d\mu_{T_z M}(u_2).
\end{align}
The second equality in the above equation follows from the change of variable $w = \frac{u}{2}$, and for the third equality we inserted \cref{eq:C_intermediate}. Finally, the fourth equality follows from the change of variables from $(w, \tilde{w})$ to $(u_1, u_2)$, as defined in \cref{eq:Vectors}. 

We now express the kernels $A_h$ and $B_h$ in terms of their Weyl symbols. Let $z_1:=z+v_1$ be the midpoint between $z-w$ and $z+\tilde w$, and let $z_2:=z+v_2$ be the midpoint between $z+\tilde w$ and $z+w$. Then, using \cref{eq:symbol2kernel}, we get
\begin{align} \label{eq:A_starprod}
    A_h(z - w, z+ \tilde{w}) &= \Delta^{1-\gamma}(z - w, z+ \tilde{w})\chi(z - w, z+ \tilde{w}) \nonumber \\
    &\qquad \times \int_{T^*_{z_1} M} \J{z - w}{z_1} a_h(z_1, p_1) \J{z_1}{z + \tilde{w}} e^{-\frac{2 i p_1 \cdot \J{z_1}{z} u_2}{h}} \frac{d \mu_{T^*_{z_1} M}(p_1)}{(2 \pi h)^d} \nonumber \\
    &= \Delta^{1-\gamma}(z - w, z+ \tilde{w})\chi(z - w, z+ \tilde{w}) \nonumber \\
    &\qquad \times \int_{T^*_{z} M} \J{z - w}{z_1} a_h(z_1, \J{z_1}{z}q_1) \J{z_1}{z + \tilde{w}} e^{-\frac{2 i q_1 \cdot u_2}{h}} \frac{d \mu_{T^*_{z} M}(q_1)}{(2 \pi h)^d},
\end{align}
where in the second equality we used the change of variables $p_1 = \J{z_1}{z} q_1$, which gives $d\mu_{T^*_{z_1}M}(p_1) = d\mu_{T^*_{z}M}(q_1)$. Similarly, for the kernel $B_h$, we have
\begin{align} \label{eq:B_starprod}
    B_h(z + \tilde{w}, z + w) &= \Delta^{1-\gamma}(z + \tilde{w}, z + w)\chi(z + \tilde{w}, z + w) \nonumber \\
    &\qquad \times \int_{T^*_{z_2} M} \J{z + \tilde{w}}{z_2} b_h(z_2, p_2) \J{z_2}{z + w} e^{\frac{2 i p_2 \cdot \J{z_2}{z} u_1}{h}} \frac{d \mu_{T^*_{z_2} M}(p_2)}{(2 \pi h)^d} \nonumber \\
    &= \Delta^{1-\gamma}(z + \tilde{w}, z + w)\chi(z + \tilde{w}, z + w) \nonumber \\
    &\qquad \times \int_{T^*_{z} M} \J{z + \tilde{w}}{z_2} b_h(z_2, \J{z_2}{z}q_2) \J{z_2}{z + w} e^{\frac{2 i q_2 \cdot u_1}{h}} \frac{d \mu_{T^*_{z} M}(q_2)}{(2 \pi h)^d}.
\end{align}
Inserting these results into \cref{eq:C_intermediate2}, we obtain 
\begin{align}
    c_h (z, p) &= (\pi h)^{-2 d} \int_{N_z} \J{z}{z - w} \cdot \J{z - w}{z + v_1} a_h(z + v_1, \J{z+v_1}{z} q_1) \J{z+v_1}{z+ \tilde{w}} \nonumber\\
    & \qquad \qquad \qquad \times \J{z+\tilde{w}}{z+ v_2} b_h(z + v_2, \J{z+v_2}{z} q_2) \J{z+v_2}{z + w} \cdot \J{z+w}{z}  \nonumber \\
    & \qquad \qquad \qquad \times e^{\frac{2 i}{h} ( q_2 \cdot u_1 - q_1 \cdot u_2 + p \cdot w )} \Lambda(z, u_1, u_2) d\mu_{N_z}(u_1, u_2, q_1, q_2) \nonumber \\
    &= (\pi h)^{-2 d} \int_{N_z} \J{z}{z - w} \cdot \J{z - w}{z + v_1} a_h(z + v_1, \J{z+v_1}{z} ( p + p_1) ) \J{z+v_1}{z+ \tilde{w}} \nonumber\\ & \qquad \qquad \qquad \times \J{z+\tilde{w}}{z+ v_2} b_h(z + v_2, \J{z+v_2}{z} (p + p_2)) \J{z+v_2}{z + w} \cdot \J{z+w}{z}  \nonumber \\
    & \qquad \qquad \qquad \times e^{\frac{2 i}{h} ( p_2 \cdot u_1 - p_1 \cdot u_2)} e^{\frac{2 i}{h} p \cdot ( u_1 -  u_2 + w )} \Lambda(z, u_1, u_2) d\mu_{N_z}(u_1, u_2, p_1, p_2) \nonumber \\
    &\sim e^{\frac{i h}{2} \left( \partial_{u_1} \cdot \partial_{p_2} - \partial_{u_2} \cdot \partial_{p_1} \right)} \Lambda(z, u_1, u_2) e^{\frac{2 i p \cdot (w + u_1 - u_2)}{h}} \nonumber \\
    &\qquad \times \Big[\J{z}{z - w} \cdot \J{z - w}{z + v_1} a_h (z+v_1, \J{z+v_1}{z} (p+p_1)) \J{z+v_1}{z+\tilde{w}} \Big]\nonumber\\
    &\qquad\times \Big[ \J{z+\tilde{w}}{z+v_2} b_h(z+v_2, \J{z+v_2}{z} (p+p_2)) \J{z+v_2}{z + w} \cdot \J{z+w}{z} \Big] \Bigg|_{\substack{u_1 = u_2 = 0\\p_1=p_2=0}} ,
\end{align}
where $N_{z}:= T_zM \times T_zM \times T^*_zM \times T^*_zM$, and $\Lambda(z, u_1, u_2)$ is defined as in \cref{eq:GeomFact}. The cut-off functions in \cref{eq:A_starprod,eq:B_starprod} can be absorbed into the integrand. Since they are identically $1$ in a neighborhood of the diagonal, the terms obtained by replacing them with $1$ are supported away from the stationary set of the phase. By the same integration-by-parts argument as in the proof of
\cref{Prop:3.1}, these terms contribute only to $h^\infty\mathcal S_h^{-\infty}(M;E,G)$. We therefore suppress them. The second equality follows after the change of variables $(q_1, q_2) = (p_1 + p, p_2 + p)$, and the final relation is obtained by applying the standard oscillatory integral expansion to the variables
$(u_1,u_2,p_1,p_2)$, as given in \cite[Prop.~1.2.4.]{duistermaat2010}. Note that this result has the same form as in \cite[Thm.~3.9]{dls_weyl}, the only difference being the additional parallel transport operators. In particular, when the symbols are scalar valued and the parallel transport operators are dropped, the above result coincides with \cite[Thm.~3.9]{dls_weyl}. 

To complete the proof, note that we can rewrite the composition of parallel transport operators in the above equation as 
\begin{subequations}
\begin{align}
    \J{z}{z - w} \cdot \J{z - w}{z + v_1} &= \mathbb{H}_z(\nabla^{\pi^*G}) \cdot \J{z}{z+v_1}, \\
    \J{z+v_1}{z+\tilde{w}} \cdot \J{z+\tilde{w}}{z+v_2} &= \J{z+v_1}{z} \cdot \mathbb{H}_z(\nabla^{\pi^*F}) \cdot \J{z}{z+v_2}, \\
    \J{z+v_2}{z + w} \cdot \J{z+w}{z} &= \J{z+v_2}{z} \cdot \mathbb{H}_z(\nabla^{\pi^*E}),
\end{align}
\end{subequations}
where
\begin{subequations}
\begin{align}
    \mathbb{H}_z(\nabla^{\pi^*G}) &= \J{z}{z-w} \cdot \J{z-w}{z+v_1} \cdot \J{z+v_1}{z} , \\
    \mathbb{H}_z(\nabla^{\pi^*F}) &= \J{z}{z+v_1} \cdot \J{z+v_1}{z+\tilde{w}} \cdot \J{z+\tilde{w}}{z+v_2} \cdot \J{z+v_2}{z}, \\
    \mathbb{H}_z(\nabla^{\pi^*E}) &= \J{z}{z+v_2} \cdot \J{z+v_2}{z+w} \cdot \J{z+w}{z}.
\end{align}
\end{subequations}
Here, $\mathbb{H}_z(\nabla^{\pi^*G})$ is the holonomy of the connection $\nabla^{\pi^*G}$ at $z$ along the geodesic triangle $z \mapsto z+v_1 \mapsto z-w \mapsto z$, $\mathbb{H}_z(\nabla^{\pi^*F})$ is the holonomy of the connection $\nabla^{\pi^*F}$ at $z$ along the geodesic loop $z \mapsto z+v_2 \mapsto z+\tilde{w} \mapsto z+v_1 \mapsto z$ (this is the blue loop on \cref{Fig:Triangle}), and $\mathbb{H}_z(\nabla^{\pi^*E})$ is the holonomy of the connection $\nabla^{\pi^*E}$ at $z$ along the geodesic triangle 
$z \mapsto z+w \mapsto z+v_2 \mapsto z$. Using this result, the star product can be written as
\begin{align}\label{eq:StarProductHol}
    c_h(z,p) = (a_h \star b_h)(z,p) &\sim e^{\frac{i h}{2} \left( \partial_{u_1} \cdot \partial_{p_2} - \partial_{u_2} \cdot \partial_{p_1} \right)} \Lambda(z, u_1, u_2) e^{\frac{2 i p \cdot (w + u_1 - u_2)}{h}} \nonumber \\
    &\qquad \times {\mathbb{H}_z(\nabla^{\pi^*G})} \Big[\J{z}{z+v_1} a_h(z+v_1, \J{z+v_1}{z} (p+p_1)) \J{z+v_1}{z} \Big] \mathbb{H}_z(\nabla^{\pi^*F}) \nonumber \\
    &\qquad\times \Big[ \J{z}{z+v_2} b_h(z+v_2, \J{z+v_2}{z} (p+p_2)) \J{z+v_2}{z} \Big] {\mathbb{H}_z(\nabla^{\pi^*E})} \Bigg|_{\substack{u_1 = u_2 = 0\\p_1=p_2=0}} ,
\end{align}
which concludes the proof.
\end{proof}

\subsection{Asymptotic expansion of the star product}
\label{Sec:AsympStarProd}

In this section, we explicitly calculate the first terms of the semiclassical asymptotic expansion of the star product defined in \cref{th:star_product}, up to third order in the semiclassical parameter $h$. In principle, higher-order terms can be computed following the same systematic procedure outlined below. Many of the terms in the expansion coincide with those appearing in the scalar case discussed in \cite[Sec.~3.9]{dls_weyl}, since the geometric factor $\Lambda(z,u_1,u_2)$, the oscillatory exponential $\exp\left(\frac{2i}{h}p\cdot(w+u_1-u_2)\right)$, and the vectors $v_1,v_2,w,\tilde w$ of the geodesic triangles depicted in \cref{Fig:Triangle} have the same local expansions. The main difference arises from the bundle holonomy terms introduced in \cref{eq:star2}, which account for the parallel transport operations that act on the fibers of the corresponding vector bundles.

\begin{proposition}[Star product expansion] \label{Prp:StarProdExp}
Let $a=(a_h)_{h\in(0,h_0]}\in h^{-l_1}\mathcal{S}_h^{k_1}(M;F,G)$ and $b=(b_h)_{h\in(0,h_0]}\in h^{-l_2}\mathcal{S}_h^{k_2}(M;E,F)$. Then the star product $a \star b$ is well defined modulo elements in $h^\infty\mathcal{S}_h^{-\infty}(M;E,G)$ and satisfies
\begin{equation}
    a \star b = (a_h \star b_h)_{h\in(0,h_0]} \in h^{-(l_1+l_2)}\mathcal S_h^{k_1+k_2}(M;E,G).
\end{equation}
Moreover, for every $h\in(0,h_0]$, one has the expansion
\begin{equation}\label{eq:ExpStarGen}
    a_h\star b_h = \sum_{j=0}^3 h^j (a_h\star b_h)_j + r_{4,h},
\end{equation}
where the remainder satisfies $r_4=(r_{4,h})_{h\in(0,h_0]} \in h^{4-(l_1+l_2)}\mathcal{S}_h^{k_1+k_2-4}(M;E,G)$ and the coefficient families satisfy
\begin{equation}
    \big((a_h\star b_h)_j\big)_{h\in(0,h_0]} \in h^{-(l_1+l_2)}\mathcal S_h^{k_1+k_2-j}(M;E,G), \qquad j = 0, 1, 2, 3.
\end{equation}
To keep the notation simple, the subscript $h$ is suppressed in the following formulas, so that $a$ and $b$ stand for $a_h$ and $b_h$, and $(a\star b)_j$ stands for $(a_h\star b_h)_j$. The coefficients are then given by
\begin{subequations}
\begin{align}
    (a \star b)_0 =& \, a b, \\
    (a \star b)_1 =& \, \tfrac{i}{2} \left( a_\alpha b^\alpha - a^\alpha b_\alpha \right),\\
    (a \star b)_2 =& - \tfrac{1}{8} \left( a_{\alpha_1\alpha_2}b^{\alpha_1\alpha_2} - 2 a^{\alpha_2}_{\alpha_1}b_{\alpha_2}^{\alpha_1} + a^{\alpha_1\alpha_2} b_{\alpha_1\alpha_2} \right) + \tfrac{3-4\gamma}{12} R_{\alpha_1 \alpha_2} a^{\alpha_1}b^{\alpha_2} \nonumber\\
    & -\tfrac{1}{24} R\indices{^\beta_{\alpha_1\alpha_2\alpha_3}} p_\beta \left( a^{\alpha_2}b^{\alpha_1\alpha_3} + a^{\alpha_1\alpha_3}b^{\alpha_2} \right) \nonumber\\
    &-{\tfrac{1}{8} a^{\alpha_1} b^{\alpha_2} E_{\alpha_1 \alpha_2} }- \tfrac{1}{4} a^{\alpha_1} \tensor{F}{_{\alpha_{1} \alpha_{2}}} b^{\alpha_2} -\tfrac{1}{8} G_{\alpha_1 \alpha_2} a^{\alpha_1} b^{\alpha_2} , \\
    (a \star b)_3 =& - \tfrac{i}{48} \left( a_{\alpha_1\alpha_2\alpha_3}b^{\alpha_1\alpha_2\alpha_3} - 3 a_{\alpha_1\alpha_2}^{\alpha_3}b^{\alpha_1\alpha_2}_{\alpha_3} + 3 a^{\alpha_2 \alpha_3}_{\alpha_1} b_{\alpha_2 \alpha_3}^{\alpha_1} - a^{\alpha_1\alpha_2\alpha_3} b_{\alpha_1\alpha_2\alpha_3} \right)  \nonumber\\
    &+ \tfrac{i (3 - 4\gamma)}{24} R_{\alpha_1 \alpha_2} \left( a_{\alpha_3}^{\alpha_2} b^{\alpha_1 \alpha_3} - a^{\alpha_2 \alpha_3} b_{\alpha_3}^{\alpha_1} \right) \nonumber\\
    &- \tfrac{i}{16} R\indices{^\beta_{\alpha_1\alpha_2\alpha_3}} \left( a^{\alpha_1 \alpha_3}_{\beta}b^{\alpha_2} - a^{\alpha_2}b^{\alpha_1 \alpha_3}_{\beta} \right) \nonumber\\
    &+ \tfrac{i}{48} R\indices{^\beta_{\alpha_1\alpha_2\alpha_3}} p_\beta  \left( - a^{\alpha_1\alpha_3}_{\alpha_4}b^{\alpha_2\alpha_4} - a^{\alpha_2}_{\alpha_4} b^{\alpha_1\alpha_3\alpha_4}+a^{\alpha_1\alpha_3\alpha_4}b^{\alpha_2}_{\alpha_4} + a^{\alpha_2\alpha_4} b^{\alpha_1\alpha_3}_{\alpha_4} \right)  \nonumber\\
    &+ \tfrac{i (3 - 4\gamma)}{48} R_{\alpha_1 \alpha_2 ; \alpha_3} \left( a^{\alpha_3}b^{\alpha_1\alpha_2} - a^{\alpha_1\alpha_2} b^{\alpha_3} \right)  \nonumber\\
    &+ \tfrac{i}{48} R\indices{^\beta_{\alpha_1\alpha_2\alpha_3;\alpha_4}} p_\beta \left( a^{\alpha_1\alpha_3\alpha_4} b^{\alpha_2} - a^{\alpha_2} b^{\alpha_1\alpha_3\alpha_4} \right)  \nonumber\\
    &- \tfrac{i}{16} (a^{\alpha_1}_{\alpha_3} b^{\alpha_2 \alpha_3}  + a^{\alpha_2 \alpha_3} b^{\alpha_1}_{\alpha_3} ) \tensor{E}{_{\alpha_{1}\alpha_{2}}} +\tfrac{i}{48} (a^{\alpha_1} b^{\alpha_2\alpha_3} -2 a^{\alpha_2\alpha_3} b^{\alpha_1}) \tensor{E}{_{\alpha_1\alpha_2;\alpha_3}} \nonumber \\
    &- \tfrac{i}{8} (a^{\alpha_1}_{\alpha_3}\tensor{F}{_{\alpha_{1}\alpha_{2}}}b^{\alpha_2 \alpha_3} + a^{\alpha_2 \alpha_3}\tensor{F}{_{\alpha_{1}\alpha_{2}}} b^{\alpha_1}_{\alpha_3})  -\tfrac{i}{16} (a^{\alpha_1}\tensor{F}{_{\alpha_1\alpha_2;\alpha_3}}b^{\alpha_2\alpha_3} + a^{\alpha_2\alpha_3}\tensor{F}{_{\alpha_1\alpha_2;\alpha_3}}b^{\alpha_1}) \nonumber \\
    &- \tfrac{i}{16} \tensor{G}{_{\alpha_{1}\alpha_{2}}} (  a^{\alpha_1}_{\alpha_3} b^{\alpha_2 \alpha_3} + a^{\alpha_2 \alpha_3} b^{\alpha_1}_{\alpha_3}) +\tfrac{i}{48} \tensor{G}{_{\alpha_1\alpha_2;\alpha_3}} (a^{\alpha_2\alpha_3} b^{\alpha_1} -2a^{\alpha_1} b^{\alpha_2\alpha_3}) .
\end{align}    
\end{subequations}
    In the above expressions, we used the same notation as in \cite{dls_weyl}, in which lower indices of the symbols denote horizontal derivatives and upper indices of the symbols denote vertical derivatives:
    \begin{equation}
        a_{\alpha_1 \cdots \alpha_n} = \hnabla_{\alpha_n} \cdots \hnabla_{\alpha_1} a , \qquad
        a^{\alpha_1 \cdots \alpha_n} =\vnabla^{\alpha_n} \cdots \vnabla^{\alpha_1} a .
    \end{equation}
    The tensors $E$, $F$, and $G$ are the curvature tensors of the connections $\nabla^E$, $\nabla^F$, and $\nabla^G$, respectively. Frame indices in the bundles $\pi^*(F^*\otimes G)$ and $\pi^*(E^*\otimes F)$ are suppressed, so the order of the factors in the above expressions is essential. For example, terms such as $a b$ should be understood as $a\indices{^A_B} b\indices{^B_C}$ once the bundle indices are restored.
\end{proposition}

\begin{proof}
The starting point of the expansion is \cref{eq:star2}. We work componentwise in $h\in(0,h_0]$ and, throughout the proof, suppress the subscript $h$ on the fixed-$h$ symbols $a_h$ and $b_h$. Thus, in this proof, $a$ and $b$ stand for $a_h$ and $b_h$, respectively, and $a\star b$ stands for $a_h\star b_h$. We also introduce the compact notation 
    \begin{subequations}\label{eq:CompactNotation}
        \begin{align}
             \Lambda &:= \Lambda(z, u_1, u_2), \\
            \mathbb{E} &:= e^{\frac{2 i p \cdot (w + u_1 - u_2)}{h}}, \\ 
            \mathbb{H}^E & :=\mathbb{H}_z(\nabla^{\pi^*E}), \\
            \mathbb{H}^F & :=\mathbb{H}_z(\nabla^{\pi^*F}), \\
            \mathbb{H}^G &:=\mathbb{H}_z(\nabla^{\pi^*G}), \\
            A &:= \Big[\J{z}{z+v_1} a(z+v_1, \J{z+v_1}{z} (p+p_1)) \J{z+v_1}{z} \Big], \\
            B &:= \Big[ \J{z}{z+v_2} b(z+v_2, \J{z+v_2}{z} (p+p_2)) \J{z+v_2}{z} \Big],
        \end{align}
    \end{subequations}
    for the relevant functions of $u_{1},u_{2},p_{1},p_{2}$. The aim is to compute the expansion
\begin{align}\label{eq:ExpStruc}
    a_h\star b_h \sim
    \bigg[ \mathbf 1 + \sum_{r=1}^{\infty}
        \frac{i^r h^r}{2^r r!}
        \left( \partial_{u_1}\cdot\partial_{p_2} -\partial_{u_2}\cdot\partial_{p_1} \right)^r \bigg]
    \big( \Lambda\cdot\mathbb E\cdot\mathbb H^G\cdot A \cdot\mathbb H^F\cdot B\cdot\mathbb H^E \big) \bigg|_{\substack{u_1=u_2=0\\p_1=p_2=0}}
\end{align}
    explicitly up to the order $h^{3}$. Note that the derivatives in this expression themselves carry frame indices. For example, the $u_{i}$-derivative of the map $\Lambda\colon T_{z}M\times T_{z}M\to \mathbb{R}$ defined by $(u_{1},u_{2})\mapsto \Lambda(z,u_{1},u_{2})$ is a map $\partial_{u_{i}}\Lambda\:T_{z}M\times T_{z}M\to T^{\ast}_{z}M$ defined by
    \begin{align}
        \bigg(\cfrac{\partial\Lambda}{\partial u_{i}}\bigg)_{\mu}=\cfrac{\partial\Lambda}{\partial u_{i}^{\mu}} .
    \end{align}
    In particular, the operator $\partial_{u_1} \cdot \partial_{p_2}$ in \cref{eq:ExpStruc} must be understood as $\partial_{u_1^{\mu}} \cdot \partial_{(p_2)_{\mu}}$, and one must carefully keep track of the contracted indices when computing the expansion.

    First of all, we note that the $u_{i}$-derivatives of the symbols $a$ and $b$ can be computed using \cref{lemma:derivative}. For example, it holds that
    \begin{align}\label{eq:horderSym}
        \cfrac{\partial}{\partial u_1^\mu} A \Bigg|_{\substack{u_1 = u_2 = 0\\p_1=p_2=0}} &= \cfrac{\partial}{\partial u_1^\mu} \Big[\J{z}{z+v_1} a(z+v_1, \J{z+v_1}{z} (p+p_1)) \J{z+v_1}{z} \Big] \Bigg|_{\substack{u_1 = u_2 = 0\\p_1=p_2=0}}  \nonumber \\
        &= \cfrac{\partial v_1^\nu}{\partial u_1^\mu} \cfrac{\partial}{\partial v_1^\nu} \Big[\J{z}{z+v_1} a(z+v_1, \J{z+v_1}{z} (p+p_1)) \J{z+v_1}{z} \Big] \Bigg|_{\substack{u_1 = u_2 = 0\\p_1=p_2=0}}\nonumber \\
        &= \cfrac{\partial v_1^\nu}{\partial u_1^\mu} \Bigg|_{u_1 = u_2 = 0} \hnabla_\nu a(z,p).
    \end{align}
    Thus, the expansions of the symbols $a$ and $b$ take exactly the same form as in \cite[Sec.~4.7]{dls_weyl}, with the only difference being that the horizontal covariant derivatives should now be understood as acting on $a \in \Gamma(T^*M, \pi^*(F^* \otimes G))$ or $b \in \Gamma(T^*M, \pi^*(E^* \otimes F))$. In particular, note that these expansions will not contain any curvature terms related to the connections on the vector bundles $E$, $F$, and $G$. There will only be Riemann curvature terms coming from the expansions of $v_1$ and $v_2$ in terms of $u_1$ and $u_2$. However, the main difference is the presence of the holonomy terms, which can be expanded in terms of $u_1$ and $u_2$ using the same methods explained in \cite{Vines_holonomy}. The resulting terms will contain bundle curvatures and their derivatives. We present the details of the expansion of the holonomy terms in \cref{Appendix:Holonomy}. The expansions of the geometric factor $\Lambda$ defined in \cref{eq:GeomFact}, the exponential $\mathbb{E}$, and the vectors $v_{1},v_{2},w,\widetilde{w}$ as defined in \cref{eq:Vectors} (cf. \cref{Fig:Triangle}) have been derived in \cite[Sec.~4]{dls_weyl} and are recalled in \cref{Subsec:Exp2}.
    
    We start by noting that the variable $p$ in $\mathbb{E}$ is a free variable independent of the momenta $p_{1,2}$. Furthermore, we observe that the only quantities that depend on $p_{1}$ and $p_{2}$ are the symbols $A$ and $B$, while all remaining terms depend exclusively on $u_{1}$ and $u_{2}$. We restrict the development to third order in $h$, which requires us to take into account those terms at fourth order that contain an extra factor of $h^{-1}$, where there is a contribution from the oscillatory exponential factor $\mathbb{E}$. At zeroth order, the symbol of the star product is given by the product of the symbols
    \begin{align}\label{exp:ord0}
        &\left(\Lambda \cdot \mathbb{E} \cdot \mathbb{H}^G \cdot A\cdot\mathbb{H}^F \cdot B \cdot \mathbb{H}^E \right)\Big|_{\substack{u_i = 0\\p_i=0}}=ab ,
    \end{align}
    where the bundle indices are suppressed, as explained in the statement of \cref{Prp:StarProdExp}. The contribution of $r=1$ to \cref{eq:ExpStruc} is given by
    \begin{align}
       \frac{{i h}}{ 2 } \left( \partial_{u_1} \cdot \partial_{p_2} - \partial_{u_2} \cdot \partial_{p_1} \right)  &\left(\Lambda \cdot \mathbb{E} \cdot \mathbb{H}^G \cdot A\cdot\mathbb{H}^F \cdot B \cdot \mathbb{H}^E\right)  
     \Big|_{\substack{u_i = 0\\p_i=0}} =\nonumber \\
     &\qquad=\frac{{i h}}{ 2 } \left(\partial_{u_1} A \cdot \partial_{p_2} B -  \partial_{p_1} A\cdot \partial_{u_2} B \right) \Big|_{\substack{u_i = 0\\p_i=0}} ,
    \end{align}
    where we used $\partial_{u_{i}}\Lambda\vert_{u_{1}=u_{2}=0}=0$, $\partial_{u_{i}}\mathbb{E}\vert_{u_{1}=u_{2}=0}=0$, and $\partial_{u_{i}}\mathbb{H}^{E,F,G}\vert_{u_{1}=u_{2}=0}=0$. Now, using \cref{lemma:derivative} and the fact that below third order we have $\partial_{u_1}\sim\partial_{v_1}$, as one can see from the expansion of $v_{i}$ in \cref{Exp:Vectors}, we can write the $u_{i}$-derivatives of the symbols as (symmetrized) horizontal derivatives, as explained in \cref{eq:horderSym}. We get for the first order:
    \begin{align}\label{exp:ord1}
       &\frac{{i h}}{ 2 } \left( \partial_{u_1} \cdot \partial_{p_2} - \partial_{u_2} \cdot \partial_{p_1} \right) \left( \Lambda \cdot \mathbb{E} \cdot \mathbb{H}^G \cdot A\cdot\mathbb{H}^F \cdot B \cdot \mathbb{H}^E \right)  
        \Big|_{\substack{u_i = 0\\p_i=0}} =\frac{{i h}}{ 2 } \left(a_\alpha  b^\alpha -   a^\alpha  b_\alpha \right).
    \end{align}
    Following similar steps, the contribution for $r=2$ in \cref{eq:ExpStruc} is given by
    \begin{align}\label{exp:ord2}
       -\frac{{h^2}}{ 8 }(\partial_{u_1} \cdot \partial_{p_2} &- \partial_{u_2} \cdot \partial_{p_1} )^2 (\Lambda \cdot \mathbb{E} \cdot \mathbb{H}^G \cdot A\cdot\mathbb{H}^F \cdot B \cdot \mathbb{H}^E)
     \Big|_{\substack{u_i = 0\\p_i=0}}=  \nonumber \\
     &= -\frac{{h^2}}{ 8 } \bigg[ 
     a_{\alpha_1\alpha_2}b^{\alpha_1\alpha_2}
     -2 a^{\alpha_1}_{\alpha_2}b_{\alpha_1}^{\alpha_2}
     + a^{\alpha_1\alpha_2} b_{\alpha_1\alpha_2}
     -\frac{2(3-4\gamma)}{3}R_{\alpha_1 \alpha_2} a^{\alpha_1}b^{\alpha_2} \nonumber\\
     &\qquad\qquad +{ a^{\alpha_1} b^{\alpha_2} \tensor{E}{_{\alpha_{1}\alpha_{2}}} }+ 2a^{\alpha_1}\tensor{F}{_{\alpha_{1}\alpha_{2}}}b^{\alpha_2} + { \tensor{G}{_{\alpha_{1}\alpha_{2}}} a^{\alpha_1} b^{\alpha_2} } \bigg],
    \end{align}
    where we used \cref{lemma:derivative}, as well as the expansions of $\Lambda$, $\mathbb{E}$, and $\mathbb{H}^{E,F,G}$ in \cref{Exp:GeomFactExp,eq:ExpHolonomy,eq:ExpHolonomy2}. At this point, it is important to stress again that we have suppressed the bundle indices in this computation. Thus, the order of the terms in the above expressions is important, and terms such as $a^{\alpha_{1}}\tensor{F}{_{\alpha_{1}\alpha_{2}}}b^{\alpha_2}$ should be understood with the relevant bundle indices restored and contracted in the natural order. At order $r=3$ in \cref{eq:ExpStruc}, we get
    \begin{align}\label{exp:ord3}
         - \frac{{i h^3}}{ 48} &(\partial_{u_1} \cdot \partial_{p_2} - \partial_{u_2} \cdot \partial_{p_1})^3 (\Lambda \cdot \mathbb{E} \cdot \mathbb{H}^G \cdot A\cdot\mathbb{H}^F \cdot B \cdot \mathbb{H}^E)\Big|_{\substack{u_i = 0\\p_i=0}} =  \nonumber\\
       &  - \frac{{i h^3}}{ 48}\bigg[ a_{\alpha_1\alpha_2\alpha_3}b^{\alpha_1\alpha_2\alpha_3}
       - 3a_{\alpha_1\alpha_2}^{\alpha_3}b^{\alpha_1\alpha_2}_{\alpha_3} 
       +3 a^{\alpha_1\alpha_2}_{\alpha_3}b_{\alpha_1\alpha_2}^{\alpha_3} 
       - a^{\alpha_1\alpha_2\alpha_3} b_{\alpha_1\alpha_2\alpha_3}\nonumber\\
       & +3R^\beta{}_{\alpha_1\alpha_2\alpha_3}(a^{\alpha_1 \alpha_3}_{\beta}b^{\alpha_2} -a^{\alpha_2}b^{\alpha_1 \alpha_3}_{\beta})
        + (3-4\gamma)R_{\alpha_1 \alpha_2 ; \alpha_3}( a^{\alpha_1\alpha_2} b^{\alpha_3}-a^{\alpha_3}b^{\alpha_1\alpha_2}) \nonumber\\
    &    +2(3-4\gamma)R_{\alpha_1 \alpha_3}( a^{\alpha_1\alpha_2}b_{\alpha_2}^{\alpha_3}- a_{\alpha_2}^{\alpha_3}b^{\alpha_1\alpha_2})\nonumber\\
    &
    +3 (a^{\alpha_1}_{\alpha_3} b^{\alpha_2 \alpha_3} + a^{\alpha_2 \alpha_3} b^{\alpha_1}_{\alpha_3}) \tensor{E}{_{\alpha_{1}\alpha_{2}}} - (a^{\alpha_1} b^{\alpha_2\alpha_3} - 2 a^{\alpha_2\alpha_3} b^{\alpha_1}) \tensor{E}{_{\alpha_1\alpha_2;\alpha_3}}
    \nonumber\\
    &
    +6 (a^{\alpha_1}_{\alpha_3}\tensor{F}{_{\alpha_{1}\alpha_{2}}}b^{\alpha_2 \alpha_3} + a^{\alpha_2 \alpha_3}\tensor{F}{_{\alpha_{1}\alpha_{2}}} b^{\alpha_1}_{\alpha_3}) + 3(a^{\alpha_1}\tensor{F}{_{\alpha_1\alpha_2;\alpha_3}}b^{\alpha_2\alpha_3} + a^{\alpha_2\alpha_3}\tensor{F}{_{\alpha_1\alpha_2;\alpha_3}}b^{\alpha_1}) 
    \nonumber\\
    &
    +3 \tensor{G}{_{\alpha_{1}\alpha_{2}}} (a^{\alpha_1}_{\alpha_3} b^{\alpha_2 \alpha_3} + a^{\alpha_2 \alpha_3} b^{\alpha_1}_{\alpha_3}) - \tensor{G}{_{\alpha_1\alpha_2;\alpha_3}} ( a^{\alpha_2\alpha_3} b^{\alpha_1} -2a^{\alpha_1} b^{\alpha_2\alpha_3}) 
    \nonumber\\
    &- 2 i h^{-1} p_\mu  R^\mu{}_{\alpha_1\alpha_2\alpha_3}(a^{\alpha_1\alpha_3}b^{\alpha_2}+a^{\alpha_2}b^{\alpha_1\alpha_3})\bigg] .
    \end{align}
    Note the term proportional to $h^{-1}$ in \cref{exp:ord3}, which will contribute to the second order of the expansion. For similar reasons, we also have to compute the order $r=4$ in \cref{eq:ExpStruc} to obtain the expansion of the star product up to the order $h^{3}$ due to the factor of $h^{-1}$ in the exponential. The relevant terms are given by
    \begin{align}\label{exp:ord4}
       \frac{{ h^4}}{ 384} ( \partial_{u_1} &\cdot \partial_{p_2} - \partial_{u_2} \cdot \partial_{p_1} )^4 (\Lambda \cdot \mathbb{E} \cdot \mathbb{H}^G \cdot A\cdot\mathbb{H}^F \cdot B \cdot \mathbb{H}^E) \Big|_{\substack{u_i = 0\\p_i=0}}=\nonumber\\
        &=\frac{{ i h^3}}{ 48} \Big[ p_\mu  R^\mu{}_{\alpha_1\alpha_2\alpha_3} (a^{\alpha_2\alpha_4} b^{\alpha_1\alpha_3}_{\alpha_4}- a^{\alpha_1\alpha_3}_{\alpha_4}b^{\alpha_2\alpha_4} 
        - a^{\alpha_2}_{\alpha_4} b^{\alpha_1\alpha_3\alpha_4}+a^{\alpha_1\alpha_3\alpha_4}b^{\alpha_2}_{\alpha_4})\nonumber \\
        &\qquad+  p_\mu R^\mu{}_{\alpha_1\alpha_2\alpha_3;\alpha_4}(a^{\alpha_1\alpha_3\alpha_4}b^{\alpha_2} - a^{\alpha_2}b^{\alpha_1\alpha_3\alpha_4})\Big] +\mathcal R_{4,h},
    \end{align}
    where the family $(\mathcal R_{4,h})_{h\in(0,h_0]}$ belongs to $h^{4-(l_1+l_2)}\mathcal S_h^{k_1+k_2-4}(M;E,G)$. 

    By combining the various terms shown in \cref{exp:ord0,exp:ord1,exp:ord2,exp:ord3,exp:ord4} and collecting equal powers of $h$, we obtain the coefficient families $(a_h\star b_h)_j$ for $j=0,1,2,3$, which appear in \cref{eq:ExpStarGen}. The standard remainder estimate for the oscillatory integral expansion, together with the symbol estimates for $a\in h^{-l_1}\mathcal S_h^{k_1}(M;F,G)$ and $b\in h^{-l_2}\mathcal S_h^{k_2}(M;E,F)$, gives
\begin{equation}
    r_4=(r_{4,h})_{h\in(0,h_0]} \in h^{4-(l_1+l_2)} \mathcal{S}_h^{k_1+k_2-4}(M;E,G).
\end{equation}
This proves the proposition.
\end{proof}

\begin{remark}
\cref{Prp:StarProdExp} does not require the input symbols $a$ and $b$ to admit asymptotic expansions in powers of $h$. It applies to arbitrary semiclassical symbols $a \in h^{-l_1}\mathcal S_h^{k_1}(M;F,G)$ and $b \in h^{-l_2}\mathcal S_h^{k_2}(M;E,F)$. Suppose, however, that $a$ and $b$ do admit asymptotic expansions of the form
\begin{equation}
    a_h\sim h^{-l_1}\sum_{r=0}^{\infty}h^r a_{r}, \qquad
    b_h\sim h^{-l_2}\sum_{s=0}^{\infty}h^s b_{s},
\end{equation}
where $a_r\in\mathcal{S}^{k_1-r}(M;F,G)$ and $b_s\in\mathcal{S}^{k_2-s}(M;E,F)$. Then the star product admits the induced expansion
\begin{equation}
    a_h\star b_h \sim h^{-(l_1+l_2)} \sum_{N=0}^{\infty}h^N c_{N},
\end{equation}
with $c_N\in\mathcal S^{k_1+k_2-N}(M;E,G)$. Up to the order computed in \cref{Prp:StarProdExp}, these coefficients are obtained by collecting all contributions of total degree $N$ in $h$ after substituting the expansions of $a_h$ and $b_h$ into the star product expansion.
\end{remark}

\begin{remark}
     For $\gamma = \tfrac{1}{2}$, and if we additionally restrict ourselves to the scalar case, the bundle curvature terms vanish, and our results coincide with those presented in \cite[Sec.~3.9]{dls_weyl}. However, note that there is a misprint in \cite[Sec.~3.9]{dls_weyl}, where the term proportional to $\tfrac{i}{48} R\indices{^\beta_{\alpha_1 \alpha_2 \alpha_3}} p_\beta$ in $(a \star b)_3$ has the wrong sign.
\end{remark}

\section{Properties of the Weyl quantization}
\label{Sec:PropWeylQuant}

In this section, we discuss some properties of the Weyl quantization introduced in \cref{Def:WeylQuant}. In particular, we show that, modulo $h^\infty$-smoothing remainders, formal self-adjointness of Weyl symbols is equivalent to formal self-adjointness of their Weyl quantizations. This generalizes the corresponding property for operators on $\mathbb{R}^{d}$ to the bundle-valued geometric setting. Furthermore, we discuss the \textit{Moyal equation} satisfied by the Wigner function in this framework.

\subsection{Self-adjoint symbols and formally self-adjoint operators}\label{Subsec:Hermitian}

For scalar pseudodifferential operators on $\mathbb{R}^{n}$, it is well known that the complex conjugate symbol in the $\tau$-quantization defines the adjoint of the operator in the $(1-\tau)$-quantization; see, for example, \cite[Thm.~23.5]{Shubin}. In particular, this implies that the Weyl quantization, corresponding to $\tau=\frac{1}{2}$, maps real-valued symbols to formally self-adjoint operators. In this section, we establish the corresponding statement for the bundle-valued Weyl quantization introduced in \cref{Def:WeylQuant}. In the present geometric setting, the quantization depends on the choice of a cut-off function only modulo $h^\infty$-smoothing operators, and the symbol is recovered only modulo $h^\infty$-smoothing symbols. Accordingly, the equivalence between fiberwise self-adjoint Weyl symbols and formally self-adjoint Weyl quantizations will be formulated modulo these semiclassical smoothing remainders.

We start by defining the adjoint of bundle-valued symbols. Let $\mathbf A\in\Gamma(T^*M,\pi^*(E^*\otimes F))$. Its adjoint is the section $\mathbf A^\dagger\in\Gamma(T^*M,\pi^*(F^*\otimes E))$ defined fiberwise, using the pulled-back fiber metrics, by the identity
\begin{equation}
    \langle \eta,\mathbf A(z,p)\xi\rangle_F = \langle \mathbf A^\dagger(z,p)\eta,\xi\rangle_E \qquad \forall (z,p)\in T^*M, \xi\in E_z, \eta\in F_z.
\end{equation}
Equivalently, for pure tensors $\phi^*\otimes\psi$, with $\phi\in\Gamma(T^*M,\pi^*E)$, $\psi\in\Gamma(T^*M,\pi^*F)$, and $\phi^*:=\langle\phi,\cdot\rangle_E$, the adjoint is
\begin{equation}
    (\phi^*\otimes\psi)^\dagger = \psi^*\otimes\phi \in \Gamma(T^*M,\pi^*(F^*\otimes E)).
\end{equation}
This gives a well-defined conjugate-linear map ${}^\dagger\colon \Gamma(T^*M,\pi^*(E^*\otimes F)) \longrightarrow \Gamma(T^*M,\pi^*(F^*\otimes E))$. In local coordinates and local orthonormal trivializations of $E$ and $F$ given by frames $e_A$ and $f_A$ with $\eta^E_{A B} = \langle e_A, e_B \rangle_E$ and $\eta^F_{A B} = \langle f_A, f_B \rangle_F$, a section $\mathbf{A}$ can be identified with a smooth matrix-valued function $\mathbf{A}\colon U\times\mathbb{R}^{d}\to\mathbb{C}^{m \times n}$, and $\mathbf{A}^{\dagger}\colon U\times\mathbb{R}^{d}\to\mathbb{C}^{n\times m}$ is then represented by the matrix $\mathbf{A}^{\dagger} = (\eta^E)^{-1} \cdot \overline{\mathbf{A}}^{\mathrm{T}} \cdot \eta^F$. In the positive-definite case and orthonormal frames, this reduces to the usual Hermitian adjoint. For a semiclassical symbol $a=(a_h)_{h\in(0,h_0]}$, we define $a^\dagger:=(a_h^\dagger)_{h\in(0,h_0]}$ componentwise. If $a\in h^{-l}\mathcal S_h^k(M;E,F)$, then $a^\dagger\in h^{-l}\mathcal S_h^k(M;F,E)$.

Since the Wigner function depends on the cut-off function $\chi$, we temporarily write $W_{h,\chi}[\Psi,\Phi]$ for the expression in \cref{eq:W_def}. The following lemma describes how the Wigner function transforms under the adjoint operation.

\begin{lemma}\label{Lemma:WignerAdjoint}
Let $\Psi\in\Gamma(F)$ and $\Phi\in\Gamma(E)$, and define $\chi^{\mathrm T}(x,y):=\chi(y,x)$. Then
\begin{equation}
    W_{h, \chi}[\Psi,\Phi]^\dagger(z,p) = W_{h, \chi^{\mathrm T}}[\Phi,\Psi](z,p).
\end{equation}
Moreover, if $\Psi\in\Gamma_{\mathrm c}(F)$ and $\Phi\in\Gamma_{\mathrm c}(E)$, then for any two admissible cut-off functions $\chi$ and $\chi'$ that agree in a neighborhood of the diagonal, and for any $a\in h^{-l}\mathcal S_h^k(M;E,F)$, one has
\begin{equation}
    \int_{T^*M} a_h \cdot \left( W_{h,\chi}[\Psi,\Phi] - W_{h,\chi'}[\Psi,\Phi] \right) d\mu_{T^*M} = \mathcal O(h^\infty).
\end{equation}
In particular, this applies to $\chi'=\chi^{\mathrm T}$.
\end{lemma}

\begin{proof}
We first prove the adjoint identity. By \cref{Def:Wigner}, we have
\begin{align}
    W_{h,\chi}[\Psi, \Phi](z, p) &=
    \int_{T_{z} M } \J{z}{z-\frac{u}{2}} \Psi^* \left( z - \frac{u}{2} \right) \otimes  \J{z}{z+\frac{u}{2}} \Phi \left(z+\frac{u}{2}\right) e^{-\frac{i p \cdot u}{h}} \nonumber\\
    &\qquad\qquad\qquad \chi\left(z-\frac{u}{2}, z+\frac{u}{2}\right)\Delta^{-\gamma}\left(z-\frac{u}{2}, z+\frac{u}{2}\right) \dfrac{ d\mu_{T_zM} (u) }{(2\pi h)^d} .
\end{align}
Taking the fiberwise adjoint and using the fact that the connections are compatible with the fiber metrics, so that parallel transport commutes with taking adjoints, gives
\begin{align}
    W_{h,\chi}[\Psi, \Phi]^\dagger(z, p) &=
    \int_{T_{z} M } \J{z}{z+\frac{u}{2}} \Phi^* \left( z + \frac{u}{2} \right) \otimes  \J{z}{z-\frac{u}{2}} \Psi \left(z-\frac{u}{2}\right) e^{\frac{i p \cdot u}{h}} \nonumber\\
    &\qquad\qquad\qquad \chi\left(z-\frac{u}{2}, z+\frac{u}{2}\right)\Delta^{-\gamma}\left(z-\frac{u}{2}, z+\frac{u}{2}\right) \dfrac{ d\mu_{T_zM} (u) }{(2\pi h)^d} ,
\end{align}
where we used the fact that the cut-off functions are real-valued. Now perform the change of variables $u\mapsto -u$. The measure $d\mu_{T_zM}(u)$ is invariant under this change, and the Van Vleck-Morette determinant is symmetric in its arguments. Thus, we obtain
\begin{align}
    W_{h,\chi}[\Psi, \Phi]^\dagger(z, p) &= \int_{T_{z} M } \J{z}{z-\frac{u}{2}} \Phi^* \left( z - \frac{u}{2} \right) \otimes  \J{z}{z+\frac{u}{2}} \Psi \left(z+\frac{u}{2}\right) e^{-\frac{i p \cdot u}{h}} \nonumber\\
    &\qquad\qquad\qquad \chi\left(z+\frac{u}{2}, z-\frac{u}{2}\right)\Delta^{-\gamma}\left(z-\frac{u}{2}, z+\frac{u}{2}\right) \dfrac{ d\mu_{T_zM} (u) }{(2\pi h)^d} \nonumber\\
    &= W_{h,\chi^{\mathrm T}}[\Phi, \Psi](z, p).
\end{align}
This proves the first claim.

It remains to prove the cut-off independence statement in the weak sense stated above. Let $\delta\chi:=\chi-\chi'$. Since $\chi$ and $\chi'$ agree in a neighborhood of the diagonal, $\delta\chi$ vanishes there. Therefore, the factor $\delta\chi(z-\frac{u}{2},z+\frac{u}{2})$ is supported where $u$ is bounded away from zero on every compact subset of the relevant coordinate patch. The difference of the two pairings is
\begin{align}
    I_h &:= \int_{T^*M} a_h \cdot \left( W_{h,\chi}[\Psi,\Phi] -W_{h,\chi'}[\Psi,\Phi] \right) d\mu_{T^*M} \nonumber\\
    &= \int_{T^*M} \int_{T_zM} a_h(z,p) \cdot \left[ \J{z}{z-\frac{u}{2}} \Psi^*\left(z-\frac{u}{2}\right) \otimes \J{z}{z+\frac{u}{2}} \Phi\left(z+\frac{u}{2}\right) \right] e^{-\frac{i p\cdot u}{h}} \nonumber\\
    &\qquad\qquad\qquad \delta\chi\left(z-\frac{u}{2},z+\frac{u}{2}\right) \Delta^{-\gamma}\left(z-\frac{u}{2},z+\frac{u}{2}\right) \frac{d\mu_{T_zM}(u)}{(2\pi h)^d} d\mu_{T^*M}(z,p).
\end{align}
Since $\Psi$ and $\Phi$ are compactly supported and the cut-offs are supported in a fixed neighborhood of the diagonal, the variables $(z,u)$ range over a compact set. Moreover, on the support of $\delta\chi$, the phase $\varphi(p,u)=p\cdot u$ has no stationary point in $p$. Using a finite partition of unity, we may work in local coordinates. Repeated integration by parts in $p$, with the same non-stationary phase argument as in the proof of \cref{Prop:3.1}, gives a factor of $h$ at each integration by parts, while each $p$-derivative lowers the symbol order of $a_h$. More explicitly, after $N_0$ integrations by parts, one obtains an estimate of the form
\begin{equation}
    |I_h| \leq C_{N_0} h^{N_0-d-l} \int_{\mathbb R^d} (1+|p|)^{k-N_0} d^dp,
\end{equation}
provided $N_0$ is chosen large enough so that the integral in $p$ converges. Since $N_0$ can be taken arbitrarily large, it follows that for every $N\in\mathbb N$ there exists a constant $C_N>0$ such that $|I_h|\leq C_N h^N$. Thus $I_h=\mathcal O(h^\infty)$, which proves the second claim.
\end{proof}

Now, if $\hat{A}\colon\Gamma_{\mathrm{c}}(E) \to \Gamma(F)$ is a linear and continuous operator, we define its \textit{formal adjoint} in the standard way, namely as a linear operator $\hat{A}^{\dagger}\colon\Gamma_{\mathrm{c}}(F)\to\Gamma(E)$ such that
\begin{align}
   (\Psi,\hat{A}\Phi)_{F}=\int_{M}\langle \Psi,\hat{A}\Phi\rangle_{F}\,d\mu_{M}=\int_{M}\langle \hat{A}^{\dagger}\Psi,\Phi\rangle_{E}\,d\mu_{M}=(\hat{A}^{\dagger}\Psi,\Phi)_{E},
\end{align}
for all test sections $\Psi\in\Gamma_{\mathrm{c}}(F)$ and $\Phi\in\Gamma_{\mathrm{c}}(E)$. If $\hat A$ has Schwartz kernel $A(x,y)\in\mathcal D'(M\times M,F\boxtimes E^*)$, then the Schwartz kernel of $\hat A^\dagger$ is characterized by
\begin{equation}
    A^\dagger(y,x)=(A(x,y))^\dagger .
\end{equation}
In particular, every linear differential operator, or more generally, every pseudodifferential operator, admits a formal adjoint. Moreover, one checks that its adjoint is again a differential or pseudodifferential operator, respectively. 

For a semiclassical family $\hat A=(\hat A_h)_{h\in(0,h_0]}$, we define its formal adjoint componentwise by $\hat A^\dagger:=(\hat A_h^\dagger)_{h\in(0,h_0]}$. In particular, if $\hat A\in h^{-l}\Psi_h^k(M;E,F)$, then $\hat A^\dagger\in h^{-l}\Psi_h^k(M;F,E)$. If two families differ by an element of $h^\infty\Psi_h^{-\infty}(M;E,F)$, then their formal adjoints differ by an element of $h^\infty\Psi_h^{-\infty}(M;F,E)$.

For the Weyl quantization, we have the following result, which generalizes the corresponding result for pseudodifferential operators on $\mathbb{R}^{n}$ (see~\cite[Thm.~23.5]{Shubin}, \cite{HormanderWeyl}, and \cite[Sec.~18.5]{HormanderIII}), as well as for the scalar-valued case on manifolds, as discussed, for instance, in \cite[Thm.~2.1(ii)]{PflaumWeyl} and \cite[Thm.~7.1]{Safarov}. In the present semiclassical geometric setting, the result is naturally formulated modulo $h^\infty$-smoothing symbols and operators.

\begin{proposition}\label{Prop:SelfAdjOp}
Let $a=(a_h)_{h\in(0,h_0]}\in h^{-l}\mathcal S_h^k(M;E,E)$ and let $\hat A=(\hat A_h)_{h\in(0,h_0]}$ be its Weyl quantization. Then $a$ is fiberwise self-adjoint modulo $h^\infty\mathcal S_h^{-\infty}(M;E,E)$ if and only if $\hat A$ is formally self-adjoint modulo $h^\infty\Psi_h^{-\infty}(M;E,E)$, that is,
\begin{equation}
    a=a^\dagger \quad \mathrm{mod} \quad h^\infty\mathcal S_h^{-\infty}(M;E,E) \qquad \Longleftrightarrow \qquad 
    \hat A=\hat A^\dagger \quad \mathrm{mod} \quad h^\infty\Psi_h^{-\infty}(M;E,E).
\end{equation}
As a special case, if $E$ is the trivial line bundle, then real-valued Weyl symbols give rise to formally self-adjoint Weyl quantizations modulo $h^\infty\Psi_h^{-\infty}(M;E,E)$.
\end{proposition}

\begin{proof}
We first record a more general observation. Let $a=(a_h)_{h\in(0,h_0]}\in h^{-l}\mathcal S_h^k(M;E,F)$ and let $\hat A=(\hat A_h)_{h\in(0,h_0]}$ be its Weyl quantization. Then $\hat A^\dagger$ is, modulo $h^\infty\Psi_h^{-\infty}(M;F,E)$, the Weyl quantization of $a^\dagger$. Equivalently, the Weyl symbol of $\hat A^\dagger$ is $a^\dagger$ modulo $h^\infty\mathcal{S}_h^{-\infty}(M;F,E)$. To show this, let $\Psi\in\Gamma_{\mathrm c}(E)$ and $\Phi\in\Gamma_{\mathrm c}(F)$. Using the definition of the formal adjoint and the convention that the Hermitian pairings are linear in the second argument, we get
\begin{align}
    (\Psi,\hat A_h^\dagger\Phi)_E = \overline{(\hat A_h^\dagger\Phi,\Psi)_E} = \overline{(\Phi,\hat A_h\Psi)_F} =
    \overline{ \int_{T^*M} a_h\cdot W_{h,\chi}[\Phi,\Psi] d\mu_{T^*M} }.
\end{align}
For bundle-valued symbols and Wigner functions, the fiberwise duality pairing satisfies
\begin{equation}
    \overline{a_h\cdot W_{h,\chi}[\Phi,\Psi]} = a_h^\dagger\cdot W_{h,\chi}[\Phi,\Psi]^\dagger .
\end{equation}
By \cref{Lemma:WignerAdjoint}, this gives
\begin{align}
    (\Psi,\hat A_h^\dagger\Phi)_E = \int_{T^*M} a_h^\dagger \cdot W_{h,\chi^{\mathrm T}}[\Psi,\Phi] d\mu_{T^*M} =
    \int_{T^*M} a_h^\dagger \cdot W_{h,\chi}[\Psi,\Phi] d\mu_{T^*M} + \mathcal{O}(h^\infty).
\end{align}
The first equality is the weak defining relation for the Weyl quantization of $a^\dagger$ using the cut-off $\chi^{\mathrm T}$. Since $\chi^{\mathrm T}$ is another admissible cut-off, the cut-off independence of Weyl quantization implies that replacing $\chi^{\mathrm T}$ by $\chi$ changes the resulting operator by an element of $h^\infty\Psi_h^{-\infty}(M;F,E)$. The last equality is precisely the weak cut-off independence statement in \cref{Lemma:WignerAdjoint}.

We now specialize to $F=E$. Suppose first that $a-a^\dagger\in h^\infty\mathcal{S}_h^{-\infty}(M;E,E)$. By the previous observation, $\hat A^\dagger$ is the Weyl quantization of $a^\dagger$ modulo $h^\infty\Psi_h^{-\infty}(M;E,E)$. Since Weyl quantization maps $h^\infty\mathcal{S}_h^{-\infty}$ into $h^\infty\Psi_h^{-\infty}$, it follows that $\hat A-\hat A^\dagger \in h^\infty\Psi_h^{-\infty}(M;E,E)$.

Conversely, suppose that $\hat A-\hat A^\dagger \in h^\infty\Psi_h^{-\infty}(M;E,E)$. Again, by the observation above, $\hat A^\dagger$ has Weyl symbol $a^\dagger$ modulo $h^\infty\mathcal{S}_h^{-\infty}(M;E,E)$, while $\hat A$ has Weyl symbol $a$. The uniqueness of Weyl symbols modulo $h^\infty\mathcal{S}_h^{-\infty}(M;E,E)$, as expressed by the inverse formula \cref{eq:kernel2symbol}, therefore implies $a-a^\dagger \in h^\infty\mathcal{S}_h^{-\infty}(M;E,E)$. This completes the proof.
\end{proof}

\begin{remark}
    The previous proposition is actually just a special case of a more general fact: For uniform pseudodifferential calculus on $\mathbb{R}^{n}$, it is well known that the adjoint of the $\tau$-quantization of some symbol $a$ is given by the $(1-\tau)$-quantization of the complex conjugate symbol $a^{\dagger}$ for $\tau\in [0,1]$, see, e.g.~\cite[Thm.~23.5]{Shubin}. A similar result has been established on manifolds with the covariant $\tau$-quantization for scalar-valued operators developed in \cite{PflaumWeyl}. In the case of the Weyl quantization on bundle-valued operators developed in this article, the same holds true modulo $h^\infty$-smoothing symbols and operators, as can be proved analogously to \cref{Prop:SelfAdjOp} with some minor modifications.
\end{remark}

\begin{proposition}[Adjoint of the star product]\label{Prop:AdjointStar}
Let $a=(a_h)_{h\in(0,h_0]}\in h^{-l_1}\mathcal S_h^{k_1}(M;F,G)$ and $b=(b_h)_{h\in(0,h_0]}\in h^{-l_2}\mathcal S_h^{k_2}(M;E,F)$. Then the adjoint of their star product satisfies
\begin{equation}
    (a\star b)^\dagger = b^\dagger\star a^\dagger \qquad\mathrm{mod}\qquad h^\infty\mathcal S_h^{-\infty}(M;G,E).
\end{equation}
Equivalently, we have
\begin{equation}
    \left( (a_h\star b_h)^\dagger - b_h^\dagger\star a_h^\dagger \right)_{h\in(0,h_0]} \in h^\infty\mathcal S_h^{-\infty}(M;G,E).
\end{equation}
In particular, since the two sides differ by an element of $h^\infty\mathcal S_h^{-\infty}(M;G,E)$, any asymptotic expansions of the two sides agree term by term to all orders in $h$.
\end{proposition}

\begin{proof}
We work componentwise in $h\in(0,h_0]$. Taking the adjoint of the oscillatory
formula in \cref{eq:star1}, we obtain
\begin{align}
    &(a_h \star b_h)^\dagger = (\pi h)^{-2 d} \int_{N_{z}} d\mu_{N_z}(u_1,u_2,p_1,p_2)\, \Lambda(z, u_1, u_2)e^{-\frac{2 i p \cdot (w + u_1 - u_2)}{h}} e^{-\frac{2 i (p_2 \cdot u_1 - p_1 \cdot u_2) }{h}} \nonumber \\
    &\quad\times \Big[\J{z}{z - w} \cdot \J{z - w}{z+v_1} a_h(z+v_1, \J{z+v_1}{z} (p+p_1)) \J{z+v_1}{z+\tilde{w}} \nonumber \\
    &\qquad\cdot \J{z+\tilde{w}}{z+v_2} b_h(z+v_2, \J{z+v_2}{z} (p+p_2)) \J{z+v_2}{z+w} \cdot \J{z+w}{z} \Big]^\dagger \quad\mathrm{mod}\quad h^{\infty}\mathcal{S}_h^{-\infty}(M;G,E) \nonumber \\
    &= (\pi h)^{-2 d} \int_{N_{z}} d\mu_{N_z}(u_1,u_2,p_1,p_2)\, \Lambda(z, u_1, u_2)e^{-\frac{2 i p \cdot (w + u_1 - u_2)}{h}} e^{-\frac{2 i (p_2 \cdot u_1 - p_1 \cdot u_2) }{h}} \nonumber \\
    &\quad\times \Big[\J{z}{z + w} \cdot \J{z + w}{z+v_2} b_h^\dagger(z+v_2, \J{z+v_2}{z} (p+p_2)) \J{z+v_2}{z+\tilde{w}} \nonumber \\
    &\qquad\cdot \J{z+\tilde{w}}{z+v_1} a_h^\dagger(z+v_1, \J{z+v_1}{z} (p+p_1))  \J{z+v_1}{z-w} \cdot \J{z-w}{z} \Big] \quad\mathrm{mod}\quad h^{\infty}\mathcal{S}_h^{-\infty}(M;G,E),
\end{align}
where we used that the geometric factor $\Lambda$ is real-valued and that the connections are compatible with the fiber metrics, so that taking adjoints reverses the parallel transport operators. We also used that taking the fiberwise adjoint maps $h^\infty\mathcal S_h^{-\infty}(M;E,G)$ to $h^\infty\mathcal S_h^{-\infty}(M;G,E)$.

Next, we perform the change of variables that swaps the pairs $p_1 \leftrightarrow p_2$ and $u_1 \leftrightarrow u_2$. Note that by \cref{eq:GeomFact,eq:Vectors}, the change of variables $u_1 \leftrightarrow u_2$ implies $v_1 \leftrightarrow v_2$, $w \rightarrow -w$, while $\tilde{w}$ and $\Lambda$ remain unchanged. Therefore, we obtain
\begin{align}
    &(a_h \star b_h)^\dagger = (\pi h)^{-2 d} \int_{N_{z}} d\mu_{N_z}(u_1,u_2,p_1,p_2)\, \Lambda(z, u_1, u_2)e^{\frac{2 i p \cdot (w + u_1 - u_2)}{h}} e^{\frac{2 i (p_2 \cdot u_1 - p_1 \cdot u_2) }{h}} \nonumber \\
    &\quad\times \Big[\J{z}{z - w} \cdot \J{z - w}{z+v_1} b_h^\dagger(z+v_1, \J{z+v_1}{z} (p+p_1)) \J{z+v_1}{z+\tilde{w}} \nonumber \\
    &\qquad\cdot \J{z+\tilde{w}}{z+v_2} a_h^\dagger(z+v_2, \J{z+v_2}{z} (p+p_2))  \J{z+v_2}{z+w} \cdot \J{z+w}{z} \Big] \quad\mathrm{mod}\quad h^{\infty}\mathcal{S}_h^{-\infty}(M;G,E) \nonumber \\
    &= b_h^\dagger \star a_h^\dagger,
\end{align}
which is exactly the star product formula \cref{eq:star1} applied to the pair
$b_h^\dagger\in h^{-l_2}\mathcal S_h^{k_2}(M;F,E)$ and $a_h^\dagger\in h^{-l_1}\mathcal S_h^{k_1}(M;G,F)$. Thus
\begin{equation}
    (a_h\star b_h)^\dagger = b_h^\dagger\star a_h^\dagger \qquad \mathrm{mod} \qquad h^\infty\mathcal S_h^{-\infty}(M;G,E).
\end{equation}
Since this holds as a statement about the family indexed by $h\in(0,h_0]$, we obtain
\begin{equation}
    (a\star b)^\dagger = b^\dagger\star a^\dagger \qquad \mathrm{mod} \qquad h^\infty\mathcal S_h^{-\infty}(M;G,E).
\end{equation}
This also implies that any asymptotic expansions of the two sides agree term by term to all orders in $h$.
\end{proof}

\subsection{Moyal equation}

In this section, we show how the Weyl calculus can be used to derive a Moyal equation \cite{Moyal_1949} for the Wigner function $W_h[\Phi_h, \Phi_h]$ associated with an $h$-dependent family of sections $\Phi = (\Phi_h)_{h \in (0, h_0]}$ satisfying a semiclassical differential equation $\hat{D}_h \Phi_h = 0$. 

\begin{remark}
Let $W_h[\Psi_h, \Phi_h](z, p)$ be the Wigner function of two sections $\Psi_h \in \Gamma(F)$ and $\Phi_h \in \Gamma(E)$, as given in \cref{Def:Wigner}. Then, treating the Wigner function $W_h[\Psi_h, \Phi_h](z, p)$ as a Weyl symbol, it follows from \cref{eq:kernel2symbol} that the corresponding Schwartz kernel is
\begin{equation} \label{eq:density_operator}
    \rho_h[\Psi_h, \Phi_h](x, y) =\frac{1}{(2 \pi h)^d} \chi(y, x) \Delta^{1-2\gamma}(y, x) \Psi_h^*(y) \otimes \Phi_h(x).
\end{equation}
In particular, for $\gamma = \tfrac{1}{2}$ the Van Vleck-Morette determinant is eliminated, and $\rho_h[\Phi_h, \Phi_h](x, x)$ corresponds to the density operator in quantum mechanics, up to the normalization factor $(2 \pi h)^{-d}$.
\end{remark}

Since we now allow non-compactly supported sections, we choose the cut-off $\chi$ in the definition of the Wigner function to be properly supported near the diagonal.

\begin{proposition}[Moyal equation]\label{prop:Moyal}
Fix the power of the Van Vleck-Morette determinant in the quantization to $\gamma = \tfrac{1}{2}$. Let $\hat D=(\hat D_h)_{h \in (0, h_0]} \in h^{-l}\mathrm{Diff}_h^k(M;E,F)$ be a semiclassical differential operator with Weyl symbol $d=(d_h)_{h \in (0, h_0]} \in h^{-l}\mathcal S_h^k(M;E,F)$. Let $\Phi=(\Phi_h)_{h \in (0, h_0]}$ be a family of sections of $E$. Then, for every family $\Psi=(\Psi_h)_{h \in (0, h_0]}$ of sections of $F$, one has
\begin{equation}
    W_h[\Psi_h,\Phi_h]\in \mathcal S^{-\infty}(M;F,E), \qquad
    d_h\star W_h[\Psi_h,\Phi_h]\in \mathcal S^{-\infty}(M;F,F) \qquad \forall \, h \in (0, h_0],
\end{equation}
as well as the relation
\begin{equation}\label{eq:SemiclassicalMoyalMomentIdentity}
    \int_{T_x^*M} \left(d_h\star W_h[\Psi_h,\Phi_h]\right)(x,p)\, d\mu_{T_x^*M}(p) =
    \Psi_h^*(x)\otimes(\hat D_h\Phi_h)(x) \qquad \forall \, x \in M, h \in (0, h_0].
\end{equation}
Furthermore, if $\hat{D}_h \Phi_h = 0$ for all $h\in(0,h_0]$, and if $C_h(x, y)$ denotes the Schwartz kernel corresponding to $d_h \star W_h$, then 
\begin{equation}
    C_h\in C^\infty(M\times M,F\boxtimes F^*), \qquad \mathrm{supp}\, C_h  \subset \mathcal {T},
\end{equation}
where $\mathcal T := \mathrm{cl} \{(x,y): \chi(y,x)\text{ is not locally constant in the } x\text{-variable}\}$. In particular, $C_h$ vanishes in a neighborhood of the diagonal.
\end{proposition}

\begin{proof}
Fix $h \in (0, h_0]$ and $\gamma = \tfrac{1}{2}$. By \cref{eq:density_operator}, the kernel corresponding to $W_h[\Psi_h,\Phi_h]$ is
\begin{equation}\label{eq:MoyalProofDensityKernel}
    \rho_h[\Psi_h, \Phi_h](x, y) = \frac{1}{(2\pi h)^d} \chi(y,x) \Psi_h^*(y)\otimes \Phi_h(x).
\end{equation}
For each fixed $h$, this kernel is smooth. Since $\chi$ is properly supported, $\rho_h[\Psi_h, \Phi_h]$ is properly supported as a kernel from sections of $F$ to sections of $E$. Moreover, in local midpoint coordinates $(z,u)$, the assumptions on $\chi$ imply that, over every compact $K\Subset M$ in the base variable $z$, the relative variable $u$ is confined to a compact set. Thus, the local amplitude in the formula for $W_h[\Psi_h,\Phi_h]$ is smooth and compactly supported in $u$, uniformly for $z\in K$. Differentiating under the integral and integrating by parts in $u$ then gives rapid decay in $p$, with all $z$- and $p$-derivatives. Hence, $W_h[\Psi_h,\Phi_h]\in\mathcal S^{-\infty}(M;F,E)$.

Next, by definition, $d_h\star W_h[\Psi_h,\Phi_h]$ is the Weyl symbol of the composition of $\hat D_h$ with the smoothing operator whose kernel is $\rho_h[\Psi_h, \Phi_h]$. Since $\hat D_h$ is a differential operator with kernel $D_h(x, y)$, this composition has the kernel
\begin{align}\label{eq:C_Kernel}
    C_h(x,y) &= \int_M D_h(x, z) \rho_h[\Psi_h, \Phi_h](z, y)  d\mu_M(z) \nonumber \\
    &=\frac{1}{(2 \pi h)^d} \Psi_h^*(y) \otimes \int_M D_h(x, z) \chi(y, z) \Phi_h(z) d\mu_M(z) \nonumber \\
    &= \hat D_{h,x} \rho_h[\Psi_h, \Phi_h](x,y) \nonumber\\
    &= \frac{1}{(2\pi h)^d} \Psi_h^*(y) \otimes \hat D_{h,x} \Big[ \chi(y,x)\Phi_h(x) \Big] \nonumber \\
    &= \frac{1}{(2\pi h)^d} \chi(y,x) \Psi_h^*(y) \otimes \hat D_{h,x} \Phi_h(x) + \frac{1}{(2\pi h)^d} \Psi_h^*(y) \otimes [\hat D_{h,x},\chi(y,x)] \Phi_h(x).
\end{align}
Here, the proper support of $\chi$ ensures that $C_h(x,y)$ is properly supported. Since $C_h$ is smooth and has the same proper-support and local relative-support properties as $\rho_h[\Psi_h, \Phi_h]$, the same argument as above shows that its Weyl symbol is smoothing. Thus, we obtain $d_h\star W_h[\Psi_h,\Phi_h]\in \mathcal S^{-\infty}(M;F,F)$. 

Using \cref{eq:C_Kernel} and evaluating the kernel--symbol relation \eqref{eq:symbol2kernel} between $C_h(x,y)$ and $d_h\star W_h$ on the diagonal, we obtain
\begin{align}
    C_h(x, x) &= \frac{1}{(2\pi h)^d} \chi(x,x) \Psi_h^*(x) \otimes (\hat D_{h,x} \Phi_h)(x) + \frac{1}{(2\pi h)^d} \Psi_h^*(x) \otimes [\hat D_{h,x},\chi(y,x)] \Phi_h(x) \Big|_{y=x} \nonumber\\
    &= \frac{1}{(2\pi h)^d} \Psi_h^*(x) \otimes (\hat D_{h,x} \Phi_h)(x) \nonumber\\
    &=\frac{1}{(2\pi h)^d} \int_{T_x^*M} (d_h\star W_h[\Psi_h,\Phi_h])(x,p) d\mu_{T_x^*M}(p). 
\end{align}
Since $\chi=1$ in a neighborhood of the diagonal, all derivatives of $\chi(y,x)$ with respect to the $x$-variable vanish at $y=x$. As a consequence, we have $[\hat D_{h,x},\chi(y,x)] \Phi_h(x) |_{y=x} = 0$, since this term always contains derivatives of $\chi$ evaluated on the diagonal. Therefore
\begin{equation}
    \int_{T_x^*M}
    \left(d_h\star W_h[\Psi_h,\Phi_h]\right)(x,p)\,
    d\mu_{T_x^*M}(p)
    =
    \Psi_h^*(x)\otimes(\hat D_h\Phi_h)(x) \qquad \forall \, x \in M, h \in (0, h_0].
\end{equation}
Finally, if $\hat{D}_h \Phi_h = 0$, then the kernel of $d_h \star W_h$ is given by
\begin{align}
    C_h(x,y) = \frac{1}{(2\pi h)^d} \Psi_h^*(y) \otimes [\hat D_{h,x},\chi(y,x)] \Phi_h(x).
\end{align}
Since the commutator $[\hat D_{h,x},\chi(y,x)] \Phi_h(x)$ always contains $x$-derivatives of $\chi(y,x)$, its support is contained in the transition region $\mathcal{T}$ where $\chi(y,x)$ is not locally constant in the $x$-variable. 
\end{proof}

\begin{remark}
Consider the particular cases of Minkowski spacetime or globally geodesically convex manifolds together with $\gamma = \tfrac{1}{2}$. Then, the quantization can be defined without a cut-off function, the Van Vleck-Morette determinant is eliminated from \cref{eq:density_operator}, and a stronger version of the Moyal equation can be derived. In this case, if $\hat D_h\Phi_h=0$, then we have
\begin{align}
    \hat{D}_{h,x} \rho_h[\Psi_h, \Phi_h](x, y) &= \frac{1}{(2 \pi h)^d} \Psi_h^*(y) \otimes (\hat{D}_h \Phi_h)(x) = 0,
\end{align}
and we obtain the Moyal equation $d_h \star W_h[\Psi_h, \Phi_h] = 0$.
\end{remark}

\section{Examples of Weyl symbols}
\label{Sec:Examples}

In this section, we show how the Weyl symbols of certain operators can be calculated. We focus our attention on second-order linear differential operators and calculate the Weyl symbols corresponding to the scalar wave equation, Maxwell's equations, linearized Yang-Mills equations, linearized Einstein equations, and the Dirac equation. The principal parts of these symbols also arise when geometric optics/WKB-type approximations of the corresponding wave equations are considered \cite{Oancea_2020,GSHE_GW,GSHE_Dirac,oancea2025quantum}, and govern the high-frequency/semiclassical dynamics of the waves. 

\begin{proposition}[Quantization of polynomial symbols]
\label{prop:polynomial_quantization}
Let $l\in\mathbb R$, and let
\begin{equation}
    d=(d_h)_{h\in(0,h_0]} \in h^{-l}\mathcal{S}_h^2(M;E,F)
\end{equation}
be a polynomial symbol of degree at most two in the cotangent variable. More explicitly, assume that there exist semiclassical families of smooth coefficient fields 
\begin{equation}
    a=(a_h)_{h\in(0,h_0]}, \qquad b=(b_h)_{h\in(0,h_0]}, \qquad c=(c_h)_{h\in(0,h_0]},
\end{equation}
with
\begin{equation}
    a_h\in\Gamma(M,E^*\otimes F\otimes TM^{\otimes_{s}2}), \qquad
    b_h\in\Gamma(M,E^*\otimes F\otimes TM), \qquad
    c_h\in\Gamma(M,E^*\otimes F),
\end{equation}
whose local $C^\infty$-seminorms on compact subsets are uniformly bounded in $h$. In local coordinates and local trivializations, the frame components of $d_h$ are
given by
\begin{equation}
\label{eq:polynomial_symbol}
    d\indices{_{h}^A_{B}}(z,p) = h^{-l} \left[ a\indices{_{h}^A_B^{\mu\nu}}(z)p_\mu p_\nu + b\indices{_{h}^A_B^{\mu}}(z)p_\mu + c\indices{_{h }^A_{B}}(z) \right].
\end{equation}
Then, the Weyl quantization of $d$ is the semiclassical differential operator
\begin{equation}
    \hat D=(\hat D_h)_{h\in(0,h_0]}
    \in h^{-l}\mathrm{Diff}_h^2(M;E,F)
\end{equation}
given by
\begin{align}
\label{eq:polynomial_operator}
    \hat D\indices{_{h}^A_{B}} &= h^{-l} \bigg\{ -h^2 a\indices{_{h}^A_B^{\mu\nu}} \left[ \nabla_{(\mu}\nabla_{\nu)} - \frac{\gamma}{3}R_{\mu\nu} \right]
    - h^2 \left[ \nabla_{(\mu}a\indices{_{h}^A_B^{\mu\nu}} \right] \nabla_{\nu)} \nonumber\\
    &\qquad\qquad- \frac{h^2}{4} \left[ \nabla_{(\mu}\nabla_{\nu)} a\indices{_{h}^A_B^{\mu\nu}} \right] - i h b\indices{_{h}^A_B^{\mu}}\nabla_\mu - \frac{i h}{2} \left[ \nabla_\mu b\indices{_{h}^A_B^{\mu}} \right] + c\indices{_{h}^A_{B}} \bigg\}.
\end{align}
Here, the covariant derivatives are those induced by $\nabla^E$, $\nabla^F$, and the Levi-Civita connection.
\end{proposition}

\begin{proof}
We work componentwise in $h\in(0,h_0]$. Starting from \cref{Def:WeylQuant}, we consider a Wigner function $W_h \in \Gamma(T^*M, \pi^*(F^* \otimes E))$ with frame components $W\indices{_h_A^B}$ and write
\begin{align}
    \int_{M} \langle \Psi, \hat{D}_h \Phi \rangle_F\, d \mu_M = \int_{M} \Psi^*_A \hat{D}\indices{_h^A_B} \Phi^B \, d \mu_M = \int_{T^*M} d\indices{_h^A_B} W\indices{_h_A^B} [\Psi,\Phi]\, d\mu_{T^*M}.
\end{align}
We evaluate the three polynomial terms in \cref{eq:polynomial_symbol} separately. For the term quadratic in $p$, we have
\begin{align}
    &h^{-l} \int_{T^*M} a\indices{_h^A_B^{\mu \nu}}(z) p_\mu p_\nu W\indices{_h_A^B}(z, p) \, d\mu_{T^*M} (z,p) = \nonumber \\
    &= h^{-l} \int_{T^*M \times T_z M} a\indices{_h^A_B^{\mu \nu}}(z) p_\mu p_\nu \J{z}{x} \Psi^*_A \left( x \right)  \J{z}{y} \Phi^B \left(y\right)  \chi\left(x, y \right)\Delta^{-\gamma}\left(x, y \right) e^{-\frac{i p \cdot u}{h}} \, \frac{d\mu_{T^*M \times T_z M}}{(2\pi h)^d} \nonumber \\
    &= -h^{2-l} \int_{T^*M \times T_z M} a\indices{_h^A_B^{\mu \nu}}(z) \J{z}{x} \Psi^*_A \left( x \right)  \J{z}{y} \Phi^B \left(y\right) \chi\left(x, y \right) \Delta^{-\gamma}\left(x, y \right) \partial_{u^\mu} \partial_{u^\nu} e^{-\frac{i p \cdot u}{h}} \, \frac{d\mu_{T^*M \times T_z M}}{(2\pi h)^d} \nonumber \\
    &= -h^{2-l} \int_{T^*M \times T_z M} a\indices{_h^A_B^{\mu \nu}}(z) e^{-\frac{i p \cdot u}{h}} \partial_{u^\mu} \partial_{u^\nu} \left[ \J{z}{x} \Psi^*_A \left( x \right)  \J{z}{y} \Phi^B \left(y\right) \chi\left(x, y \right) \Delta^{-\gamma}\left(x, y \right) \right] \, \frac{d\mu_{T^*M \times T_z M}}{(2\pi h)^d} \nonumber \\
    &= -h^{2-l} \int_{M} a\indices{_h^A_B^{\mu \nu}}(z) \partial_{u^\mu} \partial_{u^\nu} \left[ \J{z}{x} \Psi^*_A \left( x \right)  \J{z}{y} \Phi^B \left(y\right) \chi\left(x, y \right) \Delta^{-\gamma}\left(x, y \right) \right] \Big|_{u = 0} \, d\mu_M(z).
\end{align}
In the above equation, we used the notation $x = z - \frac{u}{2}$ and $y = z + \frac{u}{2}$. The first equality is obtained by using \cref{Def:Wigner} of the Wigner function. The second equality follows by replacing $p_\mu p_\nu$ with $u$ derivatives acting on the exponential factor. In the third equality, we integrate by parts, and the final line is obtained by integrating over $T^*_z M \times T_z M$ using an integral representation of the Dirac delta function. 

We continue with the evaluation of the remaining $u$ derivatives. By definition, the cut-off function is identically $1$ in a neighborhood of the diagonal; thus, $\chi(z, z) = 1$ and all derivatives of $\chi$ vanish at $u = 0$. The derivatives of the Van Vleck-Morette determinant can be evaluated following \cite[Sec.~II.H]{visser}. This gives
\begin{align}
    \Delta^{-\gamma}\left(x, y \right) \Big|_{u = 0} = 1, \qquad
    \partial_{u^\mu} \Delta^{-\gamma}\left(x, y \right) \Big|_{u = 0} = 0, \qquad
    \partial_{u^\mu} \partial_{u^\nu} \Delta^{-\gamma}\left(x, y \right) \Big|_{u = 0} = -\frac{\gamma}{3} R_{\mu \nu}(z).
\end{align}
Using Lemma~\ref{lemma:derivative}, the derivatives of the transported sections are
\begin{subequations}
\begin{alignat}{2}
    \partial_{u^\mu} \J{z}{x} \Psi^*_A(x) \Big|_{u = 0}
    &= - \frac{1}{2} \nabla_\mu \Psi^*_A(z),
    &\qquad
    \partial_{u^\mu} \J{z}{y} \Phi^B(y) \Big|_{u = 0}
    &= + \frac{1}{2} \nabla_\mu \Phi^B(z), \\
    \partial_{u^\mu} \partial_{u^\nu} \J{z}{x} \Psi^*_A(x) \Big|_{u = 0}
    &= \frac{1}{4} \nabla_{(\mu} \nabla_{\nu)} \Psi^*_A(z),
    &\qquad
    \partial_{u^\mu} \partial_{u^\nu} \J{z}{y} \Phi^B(y) \Big|_{u = 0}
    &= \frac{1}{4} \nabla_{(\mu} \nabla_{\nu)} \Phi^B(z).
\end{alignat}
\end{subequations}
Bringing everything together, we obtain
\begin{align}
    &h^{-l} \int_{T^*M} a\indices{_h^A_B^{\mu \nu}}(z) p_\mu p_\nu W\indices{_h_A^B}(z, p) \, d\mu_{T^*M}(z,p) = \nonumber \\
    &=-h^{2-l} \int_{M} a\indices{_h^A_B^{\mu \nu}} \bigg[ \frac{1}{4} \nabla_{(\mu} \nabla_{\nu)} \Psi^*_A \Phi^B + \frac{1}{4} \Psi^*_A \nabla_{(\mu} \nabla_{\nu)} \Phi^B - \frac{1}{2} \nabla_{(\mu} \Psi^*_A \nabla_{\nu)} \Phi^B - \frac{\gamma}{3} \Psi^*_A \Phi^B R_{\mu \nu} \bigg] \, d\mu_M(z) \nonumber \\
    &= \int_{M} \Psi^*_A \bigg[ -h^{2-l} a\indices{_h^A_B^{\mu \nu}} \left( \nabla_{(\mu} \nabla_{\nu)}  - \frac{\gamma}{3}  R_{\mu \nu} \right) -h^{2-l} \left(\nabla_{(\mu} a\indices{_h^A_B^{\mu \nu}} \right) \nabla_{\nu)} \nonumber\\
    &\qquad\qquad\qquad\qquad- \frac{h^{2-l}}{4} \nabla_{(\mu} \nabla_{\nu)} a\indices{_h^A_B^{\mu \nu}}  \bigg] \Phi^B \, d\mu_M(z),
\end{align}
where the last equality is obtained after integration by parts. A similar calculation gives the quantization of the other terms in \cref{eq:polynomial_symbol} as 
\begin{subequations}
\begin{align}
    &h^{-l} \int_{T^*M} b\indices{_h^A_B^{\mu}} p_\mu W\indices{_h_A^B} \, d\mu_{T^*M}(z,p) = -i h^{1-l} \int_M \Psi^*_A \left[ b\indices{_h^A_B^{\mu}} \nabla_\mu + \frac{1}{2} \left( \nabla_\mu b\indices{_h^A_B^{\mu}} \right)\right] \Phi^B \, d\mu_{M}(z), \\
    &h^{-l} \int_{T^*M} c\indices{_h^A_B}(z) W\indices{_h_A^B}(z, p) \, d\mu_{T^*M}(z,p) = h^{-l}\int_M \Psi^*_A c\indices{_h^A_B} \Phi^B \, d\mu_{M}(z),
\end{align}
which concludes the proof.
\end{subequations}
\end{proof}

We now turn to several important examples of second-order linear differential operators that are important for physical applications, all of which are specific instances of \cref{prop:polynomial_quantization}. To start with, we consider the scalar wave operator, the (spin) Dirac operator, and the Maxwell operator.

\begin{example}[Scalar wave operator] \label{Exam:WaveOp}
Consider a massless scalar field $\psi \in C^\infty(M)$. The Weyl symbol of the Laplace-Beltrami/d'Alembert operator, or scalar wave operator, $\hat{D} \psi  = g^{\mu \nu} \nabla_{\mu} \nabla_{\nu} \psi $ on $(M,g)$ is given by 
\begin{align} 
    d(z,p) = - \frac{1}{h^2} g^{\mu \nu}(z) p_{\mu} p_{\nu} +  \frac{\gamma}{3} R(z), 
\end{align} where $R$ denotes the Ricci scalar. In particular, we recover the result of \cite{fulling,dls_weyl} when $\gamma = \tfrac{1}{2}$.
\end{example}

\begin{example}[Dirac operator and spinor wave operator] Let $(M,g)$ be a $d$-dimensional Lorentzian spin manifold equipped with a spin structure, and consider the vector bundle
\begin{align}
    SM:=\mathrm{Spin}_{d-1,1}(M)\times_{\rho}\mathbb{C}^{N},
\end{align}
where $N:=2^{\lfloor\frac{d}{2}\rfloor}$ (i.e.~$N=2^{\frac{d}{2}}$ for $d$ even and $N=2^{\frac{d-1}{2}}$ for $d$ odd), associated to the spin principal bundle $\mathrm{Spin}_{d-1,1}(M)$ via the spin representation $\rho\:\mathrm{Spin}_{0}(d-1,1)\to\mathrm{Aut}(\mathbb{C}^{N})$ (see~\cite{MR685274,lawson,Hamilton}). The Dirac operator with mass $m$ is defined by
\begin{equation}
    \hat{D} \psi = \left( i h \boldsymbol{\gamma}^\mu \nabla_\mu - m \right) \psi ,
\end{equation}
for all $\psi\in\Gamma(SM)$, where $\nabla$ denotes the spin connection obtained as a lift of the Levi-Civita connection, and where $\boldsymbol{\gamma}\:TM\to\mathrm{End}(SM)$ with $\boldsymbol{\gamma}_{\mu}:=\boldsymbol{\gamma}(\partial_{\mu})\in\Gamma^{\infty}(\mathrm{End}(SM))$ are the spacetime gamma matrices, which satisfy $\boldsymbol{\gamma}_{(\mu} \boldsymbol{\gamma}_{\nu)} = -  g_{\mu \nu}$, as well as the compatibility condition $\nabla_\mu \boldsymbol{\gamma}_\nu = 0$. Then, the Weyl symbol of the Dirac operator $\hat{D}$ is
\begin{align}
    d(z, p) = - \boldsymbol{\gamma}^\mu(z) p_\mu - m.
\end{align}
Furthermore, we can also introduce the spinor wave operator
\begin{equation} \label{eq:def_C}
    \hat{C} \psi =  g^{\mu \nu} \nabla_\mu \nabla_\nu \psi,
\end{equation}
which has the Weyl symbol
\begin{equation} \label{eq:symbol_spin_wave}
    c(z, p) = -\frac{1}{h^2} g^{\mu \nu}(z) p_\mu p_\nu + \frac{\gamma}{3} R(z),
\end{equation}
where $R$ is the Ricci scalar. The spinor wave operator $\hat{C}$ is related to the square of $\hat{D}_0 = \boldsymbol{\gamma}^\mu \nabla_\mu$ by the Lichnerowicz–Weitzenböck formula as follows:
\begin{align} \label{eq:LW_formula}
    \hat{D}_0^2 \psi&= \left( \boldsymbol{\gamma}^\mu \nabla_\mu \right) \left( \boldsymbol{\gamma}^\nu \nabla_\nu \right) \psi\nonumber\\
    &= \left( - g^{\mu \nu} \nabla_\mu \nabla_\nu - \frac{1}{8} R_{\alpha \beta \mu \nu} \boldsymbol{\gamma}^\alpha \boldsymbol{\gamma}^\beta \boldsymbol{\gamma}^\mu \boldsymbol{\gamma}^\nu \right) \psi \nonumber\\
    &= \left( - \hat{C} + \frac{1}{4} R \right) \psi.
\end{align}

\end{example}

\begin{example}[Maxwell's equations] 
Let $(M,g)$ be a $d$-dimensional Lorentzian manifold and consider the exterior derivative $d:\Omega^{k}(M)\to\Omega^{k+1}(M)$ and its formal adjoint with respect to the non-degenerate Hodge bilinear form
\begin{align}\label{eq:Hodge}
    (\alpha,\beta)_{\Omega^{k}}:=\frac{1}{k!}\int_{M}\alpha_{\mu_{1}\dots\mu_{k}}\beta^{\mu_{1}\dots\mu_{k}}\,d\mu_{M}=\int_{M}\alpha\wedge\ast\beta\qquad\qquad\forall\alpha,\beta\in\Omega^{k}_{\mathrm{c}}(M),
\end{align}
the codifferential $\delta\colon\Omega^{k+1}(M)\to\Omega^{k}(M)$, defined by
\begin{subequations}
     \begin{align}
       &(d\omega)_{\alpha_{0}\dots\alpha_{k}}:=(1+k)\nabla_{[\alpha_{0}}\omega_{\alpha_{1}\dots\alpha_{k}]}\\
       &\delta:=(-1)^{k+1}\ast^{-1} d\ast\qquad\text{i.e.}\qquad(\delta\eta)_{\alpha_{1}\dots\alpha_{k}}=-g^{\mu\nu}\nabla_{\mu}\eta_{\nu\alpha_{1}\dots\alpha_{k}},
    \end{align}
\end{subequations}
for $\omega\in\Omega^{k}(M)$ and $\eta\in\Omega^{k+1}(M)$, where $\ast\:\Omega^{j}(M)\to\Omega^{d-j}(M)$ for $j\in\{0,\dots,d\}$ denotes the Hodge $\ast$-operator. Now, consider the Maxwell operator $\hat{D}:=\delta d\colon\Omega^{1}(M)\to\Omega^{1}(M)$, i.e.
\begin{equation}
    \hat{D}\indices{_\alpha^\beta} A_\beta = -2\delta_{[\alpha}^{\beta}\nabla^{\lambda}\nabla_{\lambda]}A_{\beta}= \left( \nabla^\beta \nabla_\alpha - \delta^\beta_\alpha g^{\mu \nu} \nabla_\mu \nabla_\nu \right) A_\beta = 0\,
\end{equation}
for all $A\in\Omega^{1}(M)$. Then, the Weyl symbol of the operator $\hat{D}\indices{_\alpha^\beta}$ is
\begin{align}
  d\indices{_\alpha^\beta}(z,p) = -\frac{1}{h^2} \bigg[ p^\beta p_\alpha - \delta_\alpha^\beta g^{\mu \nu}(z) p_\mu p_\nu \bigg] + \frac13 \bigg[ \frac{3 + 2 \gamma}{2} R\indices{_\alpha^\beta}(z) - \gamma R(z) \, \delta_\alpha^\beta \bigg].
\end{align}
Note that we also used \cref{eq:commutator_covd} to swap the order of some covariant derivatives in \cref{eq:polynomial_operator}. 
\end{example}

As a generalization of the previous example, we consider the linearized Yang-Mills equations. To avoid technicalities and to keep the notation simple, we consider the special case of gauge theories with trivial principal bundles. For a mathematical discussion of the Yang-Mills theory, see~\cite{Hamilton,Rudolph}. A detailed discussion of the linearized theory can be found, for example, in \cite[Sec.~2.2.2]{Schmid}; see also \cite[Ex.~3.7]{HackSchenkel} for a brief overview.

\begin{example}[Linearized Yang-Mills equations]
    Let $(M,g)$ be a Lorentzian manifold and $G$ be a finite-dimensional, real, and semisimple Lie group with Lie algebra $(\mathfrak{g},[\cdot,\cdot])$. Furthermore, consider the bundles $E_{k}:=\mathfrak{g}\otimes\bigwedge^{k}T^{\ast}M$ whose sections are $\mathfrak{g}$-valued $k$-forms, i.e.~$\Omega^{k}(M,\mathfrak{g}):=\Gamma(E_{k})$. We equip the bundles with the natural bundle metrics induced by the canonical bundle metric for differential forms, cf.~\cref{eq:Hodge}, and the Killing form of $\mathfrak{g}$. For a given $1$-form $A\in\Omega^{1}(M,\mathfrak{g})$, we define its curvature by means of Cartan's structure equation, i.e.
    \begin{align}\label{eq:LinYMCur}
        F:=dA+\frac{1}{2}[A\wedge A],
    \end{align}
    where $[\cdot\wedge\cdot]$ denotes the natural wedge product on $\Omega^{\bullet}(M,\mathfrak{g})$ induced by $[\cdot,\cdot]$. Finally, we define the exterior covariant derivative and covariant codifferential by
    \begin{subequations}
        \begin{align}
            d_{A}&\:\Omega^{k}(M,\mathfrak{g})\to\Omega^{k+1}(M,\mathfrak{g}),\qquad d_{A}:=d+[A\wedge\cdot]\\
            \delta_{A}&\:\Omega^{k+1}(M,\mathfrak{g})\to\Omega^{k}(M,\mathfrak{g}),\qquad \delta_{A}:=d_{A}^{\ast}=(-1)^{k+1}\ast^{-1} d_{A}\ast.
        \end{align}
    \end{subequations}
    Now, assume that a given $A\in\Omega^{1}(M,\mathfrak{g})$ with corresponding curvature $F$ defined by \cref{eq:LinYMCur} satisfies the nonlinear Yang-Mills equations $\delta_{A}F=0$. The linearized Yang-Mills equations around $A$ are given by
    \begin{align}
        \hat{D}\alpha:=\delta_{A}d_{A}\alpha-\ast[\ast F\wedge \alpha]=0
    \end{align}
    for all $\alpha\in\Omega^{1}(M,\mathfrak{g})$. Choosing a basis $T_{A}$ of $\mathfrak{g}$ and denoting the corresponding structure constants by $f_{BC}^{A}$, i.e.~$[T_{B},T_{C}]=f_{BC}^{A}T_{A}$, a straightforward computation in local coordinates shows that
    \begin{align}
        \tensor{\hat{D}}{^A_{B\mu}^{\nu}}\alpha_{\nu}^{B}=\bigg[&-2\delta_{B}^{A}\delta_{[\mu}^{\nu}\nabla^{\lambda}\nabla_{\lambda]}+2f_{BC}^{A} \left( A_{[\lambda}^{C}\delta_{\mu]}^{\nu}\nabla^{\lambda}+A^{\lambda C}\delta_{[\mu}^{\nu}\nabla_{\lambda]} \right) \nonumber\\ &\quad+2f_{BC}^{A}\delta^{\nu}_{[\mu}\nabla^{\lambda}A_{\lambda]}^{C}+f_{BC}^{A}F_{\mu\lambda}^{C}g^{\lambda\nu}+2f_{DE}^{A}f_{BC}^{E}A^{\lambda D}A^{C}_{[\lambda}\delta^{\nu}_{\mu]}\bigg]\alpha_{\nu}^{B}=0.
    \end{align}
    Then, the Weyl symbol of $\hat{D}$ is given by 
    \begin{align}
        \tensor{d}{^A_{B\mu}^{\nu}}(z,p)=&-\frac{1}{h^2} \delta^{A}_{B}\bigg[ p^\nu p_\mu - \delta_\mu^\nu g^{\alpha \beta}(z) p_\alpha p_\beta \bigg] + \frac13 \delta_{B}^{A} \bigg[ \frac{3 + 2 \gamma}{2} R\indices{_\mu^\nu}(z) - \gamma R(z) \, \delta_\mu^\nu \bigg]\nonumber\\&+\frac{i}{h}f_{BC}^{A}\bigg[2\delta_{\mu}^{\nu}A^{\lambda C}(z)p_{\lambda}-A^{C}_{\mu}(z)p^{\nu}-A^{C\nu}(z)p_{\mu}\bigg]\nonumber\\&+\frac{3}{2}f_{BC}^{A}F_{\mu\lambda}^{C}(z)g^{\lambda\nu}(z)\nonumber\\&-\frac{1}{2}f_{BC}^{A}f_{DE}^{C}A^{D}_{\mu}(z)A^{\nu E}(z)+2f_{DE}^{A}f_{BC}^{E}A^{\lambda D}(z)\delta^{\nu}_{[\mu}A^{C}_{\lambda]}(z).
    \end{align}
\end{example}

Another example of a linear gauge theory is provided by linearized gravity. This is the gravitational field linearized around an arbitrary background metric solving the nonlinear Einstein field equations in vacuum.

\begin{example}[Linearized Einstein field equations]
Consider a $d$-dimensional Lorentzian manifold $(M,g)$ with $d>2$, such that $g$ is a solution to the nonlinear Einstein field equations in vacuum with cosmological constant $\Lambda\in\mathbb{R}$:
\begin{align}
    R_{\mu\nu}-\frac{1}{2}Rg_{\mu\nu}+\Lambda g_{\mu\nu}=0.
\end{align}
Consider the linearized Einstein field equations around the background $g$ \cite{GerardMurroWrochna,FewsterHunt}:
\begin{align}
    \hat{D}\indices{_{\alpha \beta}^{\sigma \delta}} \chi_{\sigma \delta} \nonumber:= &\bigg[ \delta^{\sigma}_{\alpha}\delta^{\delta}_{\beta} \left(\nabla_\mu \nabla^\mu+\frac{4\Lambda}{d-2}\right)  -2 \delta_{(\alpha}^\delta \nabla^\sigma \nabla_{\beta)} \\&\quad + g^{\sigma \delta} \nabla_\alpha \nabla_\beta+g_{\alpha \beta}\nabla^\sigma \nabla^\delta  - g_{\alpha\beta}g^{\sigma \delta}\bigg( \nabla_\mu \nabla^\mu+\frac{2\Lambda}{d-2} \bigg)\bigg] \chi_{\sigma \delta}=0,
\end{align}
for the metric perturbation $\chi_{\mu \nu} \in \Gamma(T^{\ast}M^{\otimes_{s}2})$, where indices are raised and lowered via the background metric $g$, and where $\nabla$ denotes the Levi-Civita connection of $g$ (a detailed derivation of this equation can, for instance, be found in \cite[Prop.~2.38 and 2.40]{Schmid}). Then, the Weyl symbol of $\hat{D}\indices{_{\alpha \beta}^{\sigma \delta}}$ is
\begin{align}
    \tensor{d}{_{\alpha \beta}}{^{\sigma\delta}}(z,p) &= -\frac{1}{h^2} \bigg[ \delta^{\sigma}_{\alpha}\delta^{\delta}_{\beta}p_\mu p^\mu - g_{\alpha \beta}(z) g^{\sigma \delta}(z) p_\mu p^\mu + g^{\sigma \delta}(z) p_\alpha p_\beta+ g_{\alpha \beta}(z) p^\sigma p^\delta - 2\delta_{(\alpha}^\delta p^\sigma p_{\beta)} \bigg] \nonumber\\
    &\qquad + \frac{\gamma}{3} R(z) \bigg[\delta^{\sigma}_{\alpha}\delta^{\delta}_{\beta} - g_{\alpha \beta}(z) g^{\sigma \delta}(z) \bigg] + \frac{\gamma}{3} \bigg[ g^{\sigma \delta}(z) R_{\alpha \beta}(z) + g_{\alpha \beta}(z) R^{\sigma \delta}(z) \bigg]  \nonumber\\ 
    &\qquad - \frac{3+ 2 \gamma}{6} \bigg[\delta^\delta_\beta \tensor{R}{_\alpha}{^\sigma}(z) +\delta^\delta_\alpha \tensor{R}{_\beta}{^\sigma}(z) \bigg]  + \frac12 \bigg[\tensor{R}{^\delta}{_\beta}{^\sigma}{_\alpha}(z) + \tensor{R}{^\delta}{_\alpha}{^\sigma}{_\beta}(z) \bigg]\nonumber\\&\qquad+\frac{4\Lambda}{d-2}\bigg[\delta_{\alpha}^{\sigma}\delta_{\beta}^{\delta}-\frac{1}{2}g_{\alpha\beta}(z)g^{\sigma\delta}(z)\bigg].
\end{align}
\end{example}

\begin{remark}
The Weyl symbols of the second-order differential operators given above can also be calculated using the star product. For example, consider the scalar wave operator $\hat{D} = \nabla^\mu \nabla_\mu = \hat{A}^\mu \hat{B}_\mu$ as a composition of two operators $ \hat{A}^\mu: \Gamma(T^*M) \to C^\infty(M)$ and $\hat{B}_\mu: C^\infty(M) \to \Gamma(T^*M)$ defined by
\begin{align}
     \hat{A}^\mu = \nabla^\mu\,\qquad\text{and}\qquad \hat{B}_\mu = \nabla_\mu .
\end{align}
By \cref{prop:polynomial_quantization}, the Weyl symbols of the operators $\hat{A}^\mu$ and $\hat{B}_\mu$ are given by
\begin{equation}
    a^\mu = \frac{i}{h}p^\mu, \qquad b_\mu = \frac{i}{h} p_\mu.
\end{equation}
Furthermore, as noted in \cref{Subsec:HorVertCD}, the horizontal and vertical covariant derivatives satisfy
\begin{align}
    \hnabla_\alpha p^\mu = 0 = \hnabla_\alpha p_\mu, \quad \vnabla^\alpha p^\mu = g^{\alpha \mu}, \quad \vnabla^\alpha p_\mu = \delta^\alpha_\mu, \quad \vnabla^\alpha \vnabla^\beta p^\mu = 0 = \vnabla^\alpha \vnabla^\beta p_\mu.
\end{align}
Then, using the star product expansion derived in \cref{Prp:StarProdExp}, we can calculate the symbol of the scalar wave operator $\hat{D}$ as
\begin{align}
    d = a \star b &= \sum_{n=0}^2 h^n (a \star b)_n + \mathcal{O}(h^3) \nonumber \\
    &= - \frac{1}{h^2} p^\mu p_\mu - \frac{3- 4\gamma}{12} R + \frac{1}{4} g^{\mu \alpha} F\indices{_\mu^\nu_{\alpha \beta}} \delta_\nu^\beta \nonumber \\
    &= - \frac{1}{h^2} p^\mu p_\mu + \frac{\gamma}{3}R,
\end{align}
which matches the result obtained in \cref{Exam:WaveOp}. In the above equation, the bundle curvature of $T^*M$ is $F\indices{_\mu^\nu_{\alpha \beta}} = -R\indices{^\nu_{\mu \alpha \beta}}$, the line bundle curvature $2$-forms $E_{\alpha \beta}$ and $G_{\alpha \beta}$ corresponding to the trivial scalar bundle vanish identically, and all terms of $\mathcal{O}(h^3)$ in the expansion contain horizontal derivatives or second-order vertical derivatives of $a$ and $b$ that vanish. 

As a second example, we consider the spinor wave operator $\hat{C}$ introduced in \cref{eq:def_C} and we show that its Weyl symbol $c(z,p)$ given in \cref{eq:symbol_spin_wave} can be recovered by using the star product. Using the Lichnerowicz–Weitzenböck formula \eqref{eq:LW_formula}, we have
\begin{equation}
    \hat{C} = - \hat{D}_0 \hat{D}_0 + \frac{1}{4}R.
\end{equation}
Note that in this case $\hat{D}_0 : \Gamma(SM) \to \Gamma(SM)$, so all the bundles $E$, $F$, and $G$ in the definition of the star product are nontrivial and identified with $S M$. The Weyl symbol of $\hat{D}_0$ is 
\begin{equation}
    d_0(z,p) = \frac{i}{h} \boldsymbol{\gamma}^\mu(z) p_\mu.
\end{equation}
Thus, the only nonzero horizontal and vertical derivatives are $\vnabla^\alpha d_0 = \frac{i}{h}\boldsymbol{\gamma}^\alpha$, and using the star product expansion from \cref{Prp:StarProdExp}, we obtain
\begin{align}
    d_0 \star d_0 = - \frac{1}{h^2} \boldsymbol{\gamma}^\mu \boldsymbol{\gamma}^\nu p_\mu p_\nu - \frac{3-4\gamma}{12} R_{\alpha \beta} \boldsymbol{\gamma}^\alpha \boldsymbol{\gamma}^\beta + \frac{1}{8} \boldsymbol{\gamma}^\alpha \boldsymbol{\gamma}^\beta E_{\alpha \beta} + \frac{1}{4} \boldsymbol{\gamma}^\alpha F_{\alpha \beta} \boldsymbol{\gamma}^\beta + \frac{1}{8} G_{\alpha \beta} \boldsymbol{\gamma}^\alpha \boldsymbol{\gamma}^\beta.
\end{align}
Since the bundles $E$, $F$, and $G$ are identified with $S M$, the curvature $2$-forms are
\begin{equation}
    E_{\alpha \beta} = F_{\alpha \beta} = G_{\alpha \beta} = - \frac{1}{4} R_{\alpha \beta \mu \nu} \boldsymbol{\gamma}^\mu \boldsymbol{\gamma}^\nu,
\end{equation}
Using this expression, together with the anticommutation relation $\boldsymbol{\gamma}^{(\mu} \boldsymbol{\gamma}^{\nu)} = - g^{\mu \nu}$, we can write
\begin{align}
    d_0 \star d_0 &= \frac{1}{h^2} g^{\mu \nu} p_\mu p_\nu + \frac{3-4\gamma}{12} R - \frac{1}{32} R_{\alpha \beta \mu \nu} \left( \boldsymbol{\gamma}^\alpha \boldsymbol{\gamma}^\beta \boldsymbol{\gamma}^\mu \boldsymbol{\gamma}^\nu + 2 \boldsymbol{\gamma}^\alpha \boldsymbol{\gamma}^\mu \boldsymbol{\gamma}^\nu \boldsymbol{\gamma}^\beta + \boldsymbol{\gamma}^\mu \boldsymbol{\gamma}^\nu \boldsymbol{\gamma}^\alpha \boldsymbol{\gamma}^\beta \right) \nonumber\\
    &= \frac{1}{h^2} g^{\mu \nu} p_\mu p_\nu + \frac{3-4\gamma}{12} R - \frac{1}{16} R_{\alpha \beta \mu \nu} \boldsymbol{\gamma}^\alpha \left( \boldsymbol{\gamma}^\beta \boldsymbol{\gamma}^\mu \boldsymbol{\gamma}^\nu + \boldsymbol{\gamma}^\mu \boldsymbol{\gamma}^\nu \boldsymbol{\gamma}^\beta \right) \nonumber\\
    &= \frac{1}{h^2} g^{\mu \nu} p_\mu p_\nu + \frac{3-4\gamma}{12} R - \frac{1}{16} R_{\alpha \beta \mu \nu} \boldsymbol{\gamma}^\alpha \left( 2 \boldsymbol{\gamma}^{[\beta} \boldsymbol{\gamma}^\mu \boldsymbol{\gamma}^{\nu]} -2 g^{\mu \nu} \boldsymbol{\gamma}^\beta \right) \nonumber\\
    &= \frac{1}{h^2} g^{\mu \nu} p_\mu p_\nu + \frac{3-4\gamma}{12} R.
\end{align}
Thus, we obtain
\begin{equation}
    c = -d_0 \star d_0 + \frac{1}{4}R = -\frac{1}{h^2} g^{\mu \nu} p_\mu p_\nu + \frac{\gamma}{3} R,
\end{equation}
which is consistent with \cref{eq:symbol_spin_wave}.

Similar calculations can also be performed to obtain the Weyl symbols for the Maxwell and linearized Yang-Mills equations, as well as for linearized gravity.
\end{remark}

\appendix

\section{Details on the expansion of the star product}\label{Appendix:StarProd}
This appendix provides additional details on the proof of Proposition~\ref{Prp:StarProdExp}, which concerns the expansion of the star product. Specifically, \cref{Subsec:Exp2} recalls the expansions of the geometric factor in \cref{eq:GeomFact}, the oscillatory exponential, as well as the vectors in \cref{eq:Vectors}. These form the geodesic triangle in \cref{Fig:Triangle}, as derived in \cite{dls_weyl}. The expansion of the holonomy is derived in~\cref{Appendix:Holonomy}.

\subsection{Expansion of the geometric factor, exponential term and vectors}
\label{Subsec:Exp2}

The asymptotic expansions of the oscillatory exponential $\mathrm{exp}(\frac{2i}{h}p \cdot (w+u_{1}-u_{2}))$ and the geometric factor $\Lambda(z,u_{1},u_{2})$ in the case $\gamma=\frac{1}{2}$ have been given in \cite[Sec.~4]{dls_weyl}. The expansion of $\Lambda$ for arbitrary values of $\gamma$ can be obtained in the same way, simply by multiplying the expansions of the biscalar $\zeta(x, y) = \ln \Delta^{1/2}(x,y)$ in \cite[Sec.~4.4]{dls_weyl} with the appropriate factors of $\gamma$. The results are as follows:
\begin{subequations}\label{Exp:GeomFactExp}
\begin{align}
    \Lambda (z,u_{1},u_{2}) &= \exp \bigg\{
     \tfrac{3 - 4\gamma}{3} R_{\alpha_1 \alpha_2} u_1^{\alpha_1} u_2^{\alpha_2} \nonumber \\ 
    &\qquad\qquad+ \tfrac{3 - 4 \gamma}{6} R_{\alpha_1 \alpha_2 ; \alpha_3} \left( u_1^{\alpha_1} u_1^{\alpha_2} u_2^{\alpha_3} + u_2^{\alpha_1} u_2^{\alpha_2} u_1^{\alpha_3} \right)
    + \dotsb \bigg\}, \\
    e^{\frac{2 i p \cdot (w + u_1 - u_2)}{h}} &= \exp \bigg\{ \frac{i p_\mu}{h} \Big[ \tfrac{1}{3} R\indices{^\mu_{\alpha_1\alpha_2\alpha_3}} \left( u_1^{\alpha_1} u_2^{\alpha_2} u_1^{\alpha_3} - u_2^{\alpha_1} u_1^{\alpha_2} u_2^{\alpha_3} \right)  \nonumber \\
    &\qquad\qquad+ \tfrac{1}{3} R\indices{^\mu_{\alpha_1\alpha_2\alpha_3;\alpha_4}} \left( u_1^{\alpha_1} u_2^{\alpha_2} u_1^{\alpha_3} u_1^{\alpha_4} - u_2^{\alpha_1} u_1^{\alpha_2} u_2^{\alpha_3} u_2^{\alpha_4} \right) + \dotsb \Big] \bigg\}.
\end{align}
\end{subequations}

The expansions of the vectors $v_{1},v_{2},w,\widetilde{w}\in T_{z}M$, associated with the geodesic triangles shown in \cref{Fig:Triangle} and defined in \cref{eq:Vectors}, in terms of $u_{1},u_{2}\in T_{z}M$ are derived in~\cite[Sec.~4.5]{dls_weyl}. Up to fourth order, they are given by:
\begin{subequations}\label{Exp:Vectors}
\begin{align}
    v_1^\mu &= u_1^\mu + \tfrac12 R\indices{^\mu_{\alpha_1\alpha_2\alpha_3}} u_2^{\alpha_1} u_1^{\alpha_2} u_2^{\alpha_3} + \tfrac{1}{24} R\indices{^\mu_{\alpha_1\alpha_2\alpha_3;\alpha_4}} \big( 5 u_2^{\alpha_1} u_1^{\alpha_2} u_2^{\alpha_3} u_1^{\alpha_4} - u_1^{\alpha_1} u_2^{\alpha_2} u_1^{\alpha_3} u_2^{\alpha_4} \nonumber\\
    &\qquad\qquad -2 u_1^{\alpha_1} u_2^{\alpha_2} u_1^{\alpha_3} u_1^{\alpha_4} + 2 u_2^{\alpha_1} u_1^{\alpha_2} u_2^{\alpha_3} u_2^{\alpha_4} \big) + \dotsb ,\\
    v_2^\mu &= u_2^\mu + \tfrac12 R\indices{^\mu_{\alpha_1\alpha_2\alpha_3}} u_1^{\alpha_1} u_2^{\alpha_2} u_1^{\alpha_3} + \tfrac{1}{24} R\indices{^\mu_{\alpha_1\alpha_2\alpha_3;\alpha_4}} \big( 5 u_1^{\alpha_1} u_2^{\alpha_2} u_1^{\alpha_3} u_2^{\alpha_4} - u_2^{\alpha_1} u_1^{\alpha_2} u_2^{\alpha_3} u_1^{\alpha_4} \nonumber\\
    &\qquad\qquad -2 u_2^{\alpha_1} u_1^{\alpha_2} u_2^{\alpha_3} u_2^{\alpha_4} + 2 u_1^{\alpha_1} u_2^{\alpha_2} u_1^{\alpha_3} u_1^{\alpha_4} \big) + \dotsb ,\\
    w^\mu & = u_2^\mu - u_1^\mu + \tfrac{1}{6} R\indices{^\mu_{\alpha_1\alpha_2\alpha_3}} \big( u_1^{\alpha_1} u_2^{\alpha_2} u_1^{\alpha_3} - u_2^{\alpha_1} u_1^{\alpha_2} u_2^{\alpha_3} \big) \nonumber \\
    &\qquad\qquad + \tfrac{1}{6} R\indices{^\mu_{\alpha_1\alpha_2\alpha_3;\alpha_4}} \big( u_1^{\alpha_1} u_2^{\alpha_2} u_1^{\alpha_3} u_1^{\alpha_4} - u_2^{\alpha_1} u_1^{\alpha_2} u_2^{\alpha_3} u_2^{\alpha_4} \big) + \dotsb ,\\
    \tilde{w}^\mu &= u_1^\mu + u_2^\mu + \tfrac{1}{6} R\indices{^\mu_{\alpha_1\alpha_2\alpha_3}} \big( u_1^{\alpha_1} u_2^{\alpha_2} u_1^{\alpha_3} + u_2^{\alpha_1} u_1^{\alpha_2} u_2^{\alpha_3} \big)+ \dotsb .
\end{align}
\end{subequations}

\subsection{Expansion of the holonomy}
\label{Appendix:Holonomy}

In this section, we analyze the holonomy terms appearing in the star product~\eqref{eq:star2} and show how their derivatives can be calculated. We focus on $\mathbb{H}_z(\nabla^{\pi^*F})$, since the expansion of the other holonomy terms follows from the same arguments. The term $\mathbb{H}_z(\nabla^{\pi^*F})$ represents the holonomy of the connection $\nabla^{\pi^*F}$ along the geodesic quadrilateral $z \mapsto z+v_2 \mapsto z+\tilde{w} \mapsto z+v_1 \mapsto z$, and acts as a fiber map $\mathbb{H}_z(\nabla^{\pi^*F}): \pi^*F_z \to \pi^*F_z$.

First, note that the holonomy is defined as parallel transport along geodesics in the base manifold $M$. A geodesic in $M$ with tangent vector $\dot{\gamma}$ is represented in $T^*M$ by a curve whose tangent vector is a Hamiltonian vector field. However, if $\dot{\gamma}$ is geodesic with respect to the Levi-Civita connection, then its horizontal lift $\dot{\gamma}^h$ to $T^*M$ is the corresponding Hamiltonian vector field. Thus, in the case of $\mathbb{H}_z(\nabla^{\pi^*F})$, parallel transport should be understood along the horizontal lifts of the vectors that define the geodesic quadrilateral. 

Furthermore, given a local frame $f\indices{_A}$ of $F$, we can define a local frame $(f_A' = f\indices{_A} \circ \pi)$ for the pullback bundle $\pi^*F$. Then we have
\begin{equation}
    (\nabla^{\pi^*F})_{\dot{\gamma}^h} f_A' = (\pi^*\omega\indices{^B_A})(\dot{\gamma}^h) f_B' = \omega\indices{^B_A}(d\pi(\dot{\gamma}^h)) f_B' = \omega\indices{^B_A}(\dot{\gamma}) f_B'.
\end{equation}
Thus, parallel transport with respect to $\nabla^{\pi^*F}$ along $\dot{\gamma}^h$ is equivalent to parallel transport with respect to $\nabla^F$ along $\dot{\gamma}$, and in our case we can write $\mathbb{H}_z(\nabla^{\pi^*F}) = \mathbb{H}_z(\nabla^F)$.

We proceed with the expansion of $\mathbb{H}_z(\nabla^F)$ in terms of the vectors $v_1$, $v_2$, and $\tilde{w}$ following the same steps as in \cite{Vines_holonomy}. First, parallel transport along the geodesic quadrilateral can be split into parallel transport along two geodesic triangles:
\begin{equation}\label{eq:GeodTriangHol}
\begin{aligned}
    \mathbb{H}_z(\nabla^{\pi^*F}) = \mathbb{H}_z(\nabla^F) &= \J{z}{z+v_1} \cdot \J{z+v_1}{z+\tilde{w}} \cdot \J{z+\tilde{w}}{z+v_2} \cdot \J{z+v_2}{z}\\
    &= \underbrace{\left( \J{z}{z+v_1} \cdot \J{z+v_1}{z+\tilde{w}} \cdot \J{z+\tilde{w}}{z} \right)}_{=:\mathbb{H}_1} \cdot \underbrace{\left( \J{z}{z+\tilde{w}} \cdot \J{z+\tilde{w}}{z+v_2} \cdot \J{z+v_2}{z} \right)}_{=:\mathbb{H}_2} .
\end{aligned}
\end{equation}
The holonomy along the two geodesic triangles can now be treated as in \cite{Vines_holonomy}. Looking at the first triangle $z\mapsto z+\tilde{w}\mapsto z+v_{1}\mapsto z$, we introduce two affine parameters along the geodesics $z'(\tau) = z + \tau \tilde{w}$ and $z''(\sigma) = z + \sigma v_1$, and define
\begin{equation}
    \mathbb H_1(\tau,\sigma) := \J{z}{z''(\sigma)} \cdot \J{z''(\sigma)}{z'(\tau)} \cdot \J{z'(\tau)}{z}.
\end{equation}
Then, we have $\mathbb H_1=\mathbb H_1(1,1)$ and
\begin{equation}
    \mathbb H_1(\tau,\sigma) = \sum_{m,n=0}^{\infty} \frac{\tau^m\sigma^n}{m!\,n!}\mathbb H_{(m,n)} ,
\end{equation}
where $\mathbb{H}_{(m,n)}\colon F_z \to F_z$ are constant matrices. The Levi-Civita covariant derivative of tensor fields, together with the covariant bundle derivative $\nabla^F$ (as well as the covariant derivative on the dual bundle) can be pulled back to define covariant derivatives along the geodesics $z'(\tau)$ and $z''(\sigma)$. We will denote these covariant derivatives along these geodesics by $\frac{D}{D \tau}$ and $\frac{D}{D \sigma}$. By the definition of parallel transport operators and geodesics, we have
\begin{subequations}
\begin{align}
\begin{split}\label{eq:HolFirstDer}
    \frac{D}{D \tau} \J{z}{z+\tau \tilde{w}} &= 0 = \frac{D}{D \sigma} \J{z+ \sigma v_1}{z},
\end{split}\\
\begin{split}
    \frac{D}{D \tau} \dot{z}'(\tau) &= 0 = \frac{D}{D \sigma} \dot{z}''(\sigma)  .
\end{split} 
\end{align}
\end{subequations}
Using these properties, the terms in the expansion of $\mathbb{H}_1$ can be obtained by taking multiple covariant derivatives along the geodesics and then taking the coincidence limit by setting $\tau = \sigma = 0$. We obtain $\mathbb{H}_{(0,0)} = \mathbb{I}$, and as a consequence of \cref{eq:HolFirstDer}, $\mathbb{H}_{(1,0)}=0=\mathbb{H}_{(0,1)}$. For higher-order terms, we have
\begin{equation}\label{eq:ExpHol}
    \mathbb{H}_{(m,n)} = \left[ \left( \frac{D}{D \tau} \right)^m \left( \frac{D}{D \sigma} \right)^n \J{z+ \sigma v_1}{z+\tau \tilde{w}} \right] \Bigg|_{\tau = \sigma = 0}.
\end{equation}

To compute the higher-order terms, we choose a local frame $f_{A}$ of $F$ and denote the parallel transport operator from $x\in M$ to $x^{\prime}:=x+u$ for $u\in T_{x}M$ along the geodesic in coordinates by 
\begin{align}
    \J{x^{\prime}}{x}\:F_{x}\to F_{x^{\prime}}, \qquad (\J{x^{\prime}}{x}\Psi)^{A^{\prime}}(x^{\prime})=:\tensor{g}{^{A^{\prime}}_B}(x^{\prime},x)\Psi^{B}(x).
\end{align}
To simplify notation, we will generally omit writing the explicit dependence on the base point, as the same information is already encoded in the index notation. 

Now, for the computation of the coincidence limits~\eqref{eq:ExpHol}, we shall adopt the following index notation:
\begin{table}[H]
\centering
\begin{tabular}{c|c|c}
\hline
 \cellcolor{lightgray}\textbf{base point} & \cellcolor{lightgray}\textbf{bundle indices} & \cellcolor{lightgray}\textbf{spacetime indices} \\
\hline
$z$ & $A,B,\dots$ & $\alpha,\beta,\dots$ \\
\hline
$z^{\prime}:=z+\tau \tilde{w}$ & $A^{\prime},B^{\prime},\dots$ & $\alpha^{\prime},\beta^{\prime},\dots$ \\
\hline
$z^{\prime\prime}:=z+\sigma v_{1}$ & $A^{\prime\prime},B^{\prime\prime},\dots$ & $\alpha^{\prime\prime},\beta^{\prime\prime},\dots$ 
\end{tabular}
\end{table}

\noindent With this notation, the coincidence limit~\eqref{eq:ExpHol} can be written as 
\begin{equation}\label{eq:ExpHolA}
    \tensor{(\mathbb{H}_{(m,n)})}{^A_B} =\widetilde{w}^{\alpha_{1}}\dots \widetilde{w}^{\alpha_{m}}v_{1}^{\beta_{1}}\dots v_{1}^{\beta_{n}}[\tensor{g}{^A_{B^\prime}_{;\alpha^\prime_{1}}_{\dots}_{\alpha^\prime_{m}}_{\beta_{1}}_{\dots}_{\beta_{n}}}]_{z^{\prime}\to z},
\end{equation}
where the semicolon denotes the covariant derivatives with respect to the connection $\nabla^{F}$. Note that in this expression, we first took the coincidence limit $z^{\prime\prime}\to z$, which corresponds to $\sigma\to 0$, leaving us with the subsequent limit $z^{\prime}\to z$, which corresponds to $\tau \to 0$. As discussed in more detail in \cite[Sec.~3.1]{Vines_holonomy}, the final result remains independent of the order in which these limits are taken.

Using equation~\eqref{eq:ExpHolA}, we are therefore left to compute the following coincidence limits for the expansion of $\mathbb{H}_{1}$ up to the third order in $\sigma$ and $\tau$: 
\begin{align}\label{eq:ExpHolz}
    \tensor{(\mathbb{H}_{1})}{^A_B} = \tensor{\delta}{^A_B} &+\frac{\tau^{2}}{2}\tilde{w}^{\alpha_{1}} \tilde{w}^{\alpha_{2}}[\tensor{g}{^A_{B^\prime}_{;\alpha^\prime_{1}\alpha^\prime_{2}}}]+\frac{\sigma^{2}}{2}v_{1}^{\beta_{1}} v_{1}^{\beta_{2}}[\tensor{g}{^A_{B^\prime}_{;\beta_{1}\beta_{2}}}]+\tau\sigma \tilde{w}^{\alpha_{1}}v_{1}^{\beta_{1}}[\tensor{g}{^A_{B^\prime}_{;\alpha^\prime_{1}}_{\beta_{1}}}] \nonumber\\
    &+\frac{\tau^{3}}{6}\tilde{w}^{\alpha_{1}} \tilde{w}^{\alpha_{2}}\tilde{w}^{\alpha_{3}}[\tensor{g}{^A_{B^\prime}_{;\alpha^\prime_{1}}_{\alpha^\prime_{2}}_{\alpha^\prime_{3}}}]+\frac{\sigma^{3}}{6}v_{1}^{\beta_{1}} v_{1}^{\beta_{2}}v_{1}^{\beta_{3}}[\tensor{g}{^A_{B^\prime}_{;\beta_{1}}_{\beta_{2}}_{\beta_{3}}}] \nonumber\\
    &+\frac{\tau^{2}\sigma}{2}\tilde{w}^{\alpha_{1}} \tilde{w}^{\alpha_{2}}v_{1}^{\beta_{1}}[\tensor{g}{^A_{B^\prime}_{;\alpha^\prime_{1}}_{\alpha^\prime_{2}}_{\beta_{1}}}]+\frac{\tau\sigma^{2}}{2}\tilde{w}^{\alpha_{1}} v_{1}^{\beta_{1}}v_{1}^{\beta_{2}}[\tensor{g}{^A_{B^\prime}_{;\alpha^{\prime}_{1}}_{\beta_{1}}_{\beta_{2}}}]+\dots
\end{align}
where we already used the fact that the first order vanishes, see~\cref{eq:HolFirstDer}, and where all the coincidence limits are taken in the sense $z^{\prime}\to z$, i.e.~$[\dots]:=[\dots]_{z^{\prime}\to z}$.

Now, let us denote by $\sigma_{\alpha}:=\sigma_{;\alpha}$ and $\sigma_{\alpha^{\prime}}:=\sigma_{;\alpha^{\prime}}$ the covariant derivatives of Synge's world function $\sigma=\sigma(z,z^{\prime})$. By definition, the vector $\sigma_{\alpha}$ is proportional to the tangent vector along the geodesic connecting $z$ with $z^{\prime}$, which implies that $\tensor{g}{^A_{B^{\prime}}_{;\alpha}}\sigma^{\alpha}=0$. Taking the covariant derivative twice in the $z$-variable yields
\begin{align}\label{eq:DerPT}
    0=\tensor{g}{^A_{B^{\prime}}_{;\alpha\beta\gamma}}\sigma^{\alpha}+\tensor{g}{^A_{B^{\prime}}_{;\alpha\beta}}\tensor{\sigma}{^\alpha_\gamma}+ \tensor{g}{^A_{B^{\prime}}_{;\alpha\gamma}}\tensor{\sigma}{^{\alpha}_\beta}+\tensor{g}{^A_{B^{\prime}}_{;\alpha}}\tensor{\sigma}{^\alpha_{\beta\gamma}}.
\end{align}
Now, at the coincidence limit $z^{\prime}\to z$, we have $[\sigma^{\alpha}]_{z^{\prime}\to z}=[\tensor{\sigma}{^\alpha_{\beta\gamma}}]_{z^{\prime}\to z}=0$ as well as $[\tensor{\sigma}{^\alpha_\beta}]=\delta^{\alpha}_{\beta}$ \cite[Sec.~4.1]{Poisson}, which implies
\begin{align}
    0=[\tensor{g}{^A_{B^{\prime}}_{;\alpha\beta}}]_{z^{\prime}\to z}+[\tensor{g}{^A_{B^{\prime}}_{;\beta\alpha}}]_{z^{\prime}\to z}\qquad\Rightarrow\qquad [\tensor{g}{^A_{B^{\prime}}_{;\alpha\beta}}]_{z^{\prime}\to z}=-\frac{1}{2}\tensor{F}{^A_{B\alpha\beta}},
\end{align}
where $F\in\Omega^{2}(M,\mathrm{End}(F))$ denotes the curvature of the bundle $(F,\nabla^{F})$. Similar expressions for the primed derivatives of the parallel transport operator, as well as for mixed cases, can be derived by using \textit{Synge's rule} (see e.g.~\cite[Sec.~4.2]{Poisson}). The results obtained are
\begin{subequations}\label{eq:2ndOrderHol}
    \begin{align}
        &[\tensor{g}{^A_{B^{\prime}}_{;\alpha\beta}}]_{z^{\prime}\to z}=-\frac{1}{2}\tensor{F}{^A_{B\alpha\beta}},\qquad  [\tensor{g}{^A_{B^{\prime}}_{;{\alpha^{\prime}\beta^{\prime}}}}]_{z^{\prime}\to z}=\frac{1}{2}\tensor{F}{^A_{B\alpha\beta}},\\
        &[\tensor{g}{^A_{B^{\prime}}_{;{\alpha^{\prime}\beta}}}]_{z^{\prime}\to z}=-\frac{1}{2}\tensor{F}{^A_{B\alpha\beta}},\qquad [\tensor{g}{^A_{B^{\prime}}_{;{\alpha\beta^{\prime}}}}]_{z^{\prime}\to z}=\frac{1}{2}\tensor{F}{^A_{B\alpha\beta}}.
    \end{align}
\end{subequations}
Higher orders can be obtained by taking further covariant derivatives of $\tensor{g}{^A_{B^{\prime}}_{;\alpha}}\sigma^{\alpha}=0$. By differentiating \cref{eq:DerPT} along $\partial_{\delta}$ and using $[\sigma^{\alpha}]_{z^{\prime}\to z}=[\tensor{\sigma}{^\alpha_{\beta\gamma}}]_{z^{\prime}\to z}=0$, $[\tensor{\sigma}{^\alpha_\beta}]=\delta^{\alpha}_{\beta}$, and $[\tensor{g}{^A_B_{;\alpha}}]_{z^{\prime}\to z}=0$, we obtain
\begin{align}
    [\tensor{g}{^A_{B^{\prime}}_{;\delta\beta\gamma}}]_{z^{\prime}\to z}+[\tensor{g}{^A_{B^{\prime}}_{;\gamma\beta\delta}}]_{z^{\prime}\to z}+[\tensor{g}{^A_{B^{\prime}}_{;\beta\gamma\delta}}]_{z^{\prime}\to z}=0\,\Rightarrow\,\, [\tensor{g}{^A_{B^\prime}_{;\alpha\beta\gamma}}]_{z^{\prime}\to z}=-\frac{2}{3}\tensor{F}{^A_{B\alpha(\beta;\gamma)}},
\end{align}
where we recall that we use the convention in which the symmetrization of indices $(\dots)$ is defined with an appropriate normalization to be an idempotent operation, and where the covariant derivative in $\tensor{F}{^A_{B\alpha\beta;\gamma}}$ is the one acting on sections of $F\otimes F^{\ast}\otimes T^{\ast}M^{\otimes 2}$ induced by $\nabla^{F}$ and the Levi-Civita connection on $T^{\ast}M$. Using Synge's rule, \cref{eq:2ndOrderHol} and the Bianchi identity $d^{\nabla^{\mathrm{End}(F)}}F=0$, which in local coordinates translates to $\tensor{F}{^A_{B\alpha\beta;\gamma}}+\tensor{F}{^A_{B\beta\gamma;\alpha}}+\tensor{F}{^A_{B\gamma\alpha;\beta}}=0$, we obtain the other relevant coincidence limits for the third order:
\begin{subequations}
    \begin{align}
        [\tensor{g}{^A_{B^\prime}_{;\alpha\beta\gamma}}]_{z^{\prime}\to z}&=-\frac{2}{3}\tensor{F}{^A_{B\alpha(\beta;\gamma)}},\qquad [\tensor{g}{^A_{B^\prime}_{;\alpha^{\prime}\beta^{\prime}\gamma^{\prime}}}]_{z^{\prime}\to z}=\frac{2}{3}\tensor{F}{^A_{B\alpha(\beta;\gamma)}}, \\
        [\tensor{g}{^A_{B^\prime}_{;\alpha\beta\gamma^\prime}}]_{z^{\prime}\to z}&=-\frac{1}{3}\tensor{F}{^A_{B\gamma(\alpha;\beta)}},\qquad [\tensor{g}{^A_{B^\prime}_{;\alpha\beta^{\prime}\gamma^{\prime}}}]_{z^{\prime}\to z}=\frac{1}{3}\tensor{F}{^A_{B\alpha(\beta;\gamma)}}.
    \end{align}
\end{subequations}
Plugging the obtained coincidence limits into~\eqref{eq:ExpHolz}, we obtain the following expansion of the holonomy operator $\mathbb{H}_{1}$ up to the third order:
\begin{align}
     \tensor{(\mathbb{H}_{1})}{^A_B}=\tensor{\delta}{^A_B} + \frac{\tau\sigma}{2} \tensor{F}{^A_{B\alpha \beta}} v_{1}^{\alpha} \tilde{w}^{\beta} + \frac{\tau\sigma}{6}\tensor{F}{^A_{B\alpha\beta;\gamma}} v_{1}^{\alpha} \tilde{w}^{\beta} (\tau \tilde{w}^{\gamma}+\sigma v_{1}^{\gamma})+\dots
\end{align}
This equation reduces to the result in Ref. \cite[Eq.~(6)]{Vines_holonomy} when $F = T M$, $\nabla^F$ is the Levi-Civita connection, and therefore the curvature is given by the Riemann tensor as $\tensor{F}{^A_{B\alpha \beta}} = \tensor{R}{^\mu_{\nu \alpha \beta}}$. 

Following similar steps, we obtain an expansion of the holonomy $\mathbb{H}_{2}$ of the second geodesic triangle in~\eqref{eq:GeodTriangHol}. Combining both results and reabsorbing the parameters $\tau$ and $\sigma$ in the vectors, we obtain the following expansion of $\mathbb{H}_{z}(\nabla^{\pi^{\ast}F})$:
\begin{align}
     \tensor{\mathbb{H}_{z}(\nabla^{\pi^{\ast}F})}{^A_B} = \tensor{\delta}{^A_B}& + \frac{1}{2} \tensor{F}{^A_{B\alpha \beta}} v_{1}^{\alpha} \tilde{w}^{\beta} + \frac{1}{6}\tensor{F}{^A_{B\alpha\beta;\gamma}} v_{1}^{\alpha} \tilde{w}^{\beta} ( \widetilde{w}^{\gamma}+v_{1}^{\gamma}) \nonumber\\
     &+\frac{1}{2} \tensor{F}{^A_{B\alpha\beta}} \tilde{w}^{\alpha} v_{2}^{\beta} + \frac{1}{6}\tensor{F}{^A_{B\alpha\beta;\gamma}} \tilde{w}^{\alpha} v_{2}^{\beta} ( v_{2}^{\gamma}+\tilde{w}^{\gamma})+\dots
\end{align}
As a last step, we use the expansions of the vectors $v_{1},v_{2}$ and $\widetilde{w}$ in terms of $u_{1}$ and $u_{2}$, as stated in Equation~\eqref{Exp:Vectors}, which yields the following expansion up to third order: 
\begin{align}\label{eq:ExpHolonomy}
     \tensor{\mathbb{H}_{z}(\nabla^{\pi^{\ast}F})}{^A_B}=\tensor{\delta}{^A_B}&+\tensor{F}{^A_{B\alpha \beta}} u_{1}^{\alpha } u_{2}^{\beta} + \frac{1}{2}\tensor{F}{^A_{B\alpha\beta;\gamma}}u_{1}^{\alpha}u_{2}^{\beta}(u_{1}+u_{2})^{\gamma}+\dots .
\end{align}

Using the same approach, we can also expand the remaining holonomy terms  $\mathbb{H}_z(\nabla^{\pi^*E})$ and  $\mathbb{H}_z(\nabla^{\pi^*G})$ as
\begin{subequations}\label{eq:ExpHolonomy2}
\begin{align}
    \tensor{\mathbb{H}_{z}(\nabla^{\pi^{\ast}E})}{^A_B} &= \tensor{\delta}{^A_B} + \frac{1}{2} \tensor{E}{^A_{B\alpha \beta}} v_{2}^{\alpha } w^{\beta} + \frac{1}{6}\tensor{E}{^A_{B\alpha\beta;\gamma}} v_{2}^{\alpha} w^{\beta}(v_2 + w)^{\gamma}+\dots \nonumber\\
    &= \tensor{\delta}{^A_B} + \frac{1}{2} \tensor{E}{^A_{B\alpha \beta}} u_1^{\alpha } u_2^{\beta} + \frac{1}{6}\tensor{E}{^A_{B\alpha\beta;\gamma}} u_1^{\alpha} u_2^{\beta}(2u_2 - u_1)^{\gamma}+\dots,\\
    \tensor{\mathbb{H}_{z}(\nabla^{\pi^{\ast}G})}{^A_B} &= \tensor{\delta}{^A_B} + \frac{1}{2} \tensor{G}{^A_{B\alpha \beta}} v_{1}^{\alpha } w^{\beta} + \frac{1}{6}\tensor{G}{^A_{B\alpha\beta;\gamma}} v_{1}^{\alpha} w^{\beta}(v_1 - w)^{\gamma}+\dots \nonumber\\
    &= \tensor{\delta}{^A_B} + \frac{1}{2} \tensor{G}{^A_{B\alpha \beta}} u_{1}^{\alpha } u_2^{\beta} + \frac{1}{6}\tensor{G}{^A_{B\alpha\beta;\gamma}} u_{1}^{\alpha} u_2^{\beta}(2 u_1 - u_2)^{\gamma}+\dots,
\end{align}
\end{subequations}
where $\tensor{E}{^A_{B\alpha \beta}}$ and $\tensor{G}{^A_{B\alpha \beta}}$ are the curvature tensors of the connections $\nabla^E$ and $\nabla^G$, respectively.

\printbibliography

\end{document}